\documentclass[12pt]{article}

\usepackage{color}
\usepackage{amssymb}   
\usepackage{amsthm}    
\usepackage{amsmath}   
\usepackage{stmaryrd}  
\usepackage{titletoc}  
\usepackage{mathrsfs}  
\usepackage{graphicx}
\usepackage{multicol}
\usepackage{multirow}
\usepackage{subcaption}
\usepackage{caption}
\usepackage{float}
\usepackage{hyperref}
\usepackage{cleveref}
\usepackage{xcolor}
\usepackage{hyperref}
\usepackage[english]{babel}
\addto\extrasenglish{
  
}
\addto\extrasenglish{
  
}
\addto\extrasenglish{
  
}
\usepackage{comment}

\hypersetup{
  colorlinks   = true, 
  urlcolor     = red, 
  linkcolor    = blue, 
  citecolor   = blue 
}
\usepackage[round]{natbib}   
\newlength{\defbaselineskip}
\newcommand{\setlinespacing}[1]%
           {\setlength{\baselineskip}{#1 \defbaselineskip}}

\theoremstyle{plain}

\theoremstyle{definition}

\theoremstyle{plain}
\newtheorem{theorem}{Theorem}[section]
\newtheorem{lemma}[theorem]{Lemma}

\theoremstyle{definition}

\theoremstyle{remark}

\renewcommand{\theequation}{\thesection.\arabic{equation}}
\makeatletter\@addtoreset{equation}{section} \makeatother
 \allowdisplaybreaks
\begin{document}

\title{Variable Selection with Broken Adaptive Ridge Regression for Interval-Censored Competing Risks Data}

\author{
Fatemeh Mahmoudi$^a$, Chenxi Li $^b$,  Kaida Cai $^c$, 
Xuewen Lu$^d$  \footnote{Corresponding author. Email: xlu@ucalgary.ca. Phone: (1-403)220-6620. Fax: (1-403)282-5150} \\
 $^a$ Department of Mathematics and Computing \\
 Mount Royal University, Calgary, Alberta, T3E 6K6, Canada\\
 \and
 $^b$ Department of Epidemiology and Biostatistics \\
Michigan State University, East Lansing, Michigan 48824, USA \\
\and
 $^c$ School of Public Health, School of Mathematics,  \\
Key Laboratory of Environmental Medicine Engineering, Ministry of Education, \\
Southeast University, Nanjing, 210009, China
\and
$^d$ Department of Mathematics and Statistics \\
University of Calgary, 
Calgary, Alberta, T2N 1N4,  Canada 
}


\date{}

\maketitle

\begin{abstract}
Competing risks data refer to situations where the occurrence of one event precludes the possibility of other events happening, resulting in multiple mutually exclusive events. This data type is commonly encountered in medical research and clinical trials, exploring the interplay between different events and informing decision-making in fields such as healthcare and epidemiology. We develop a penalized variable selection procedure to handle such complex data in an interval-censored setting. We consider a broad class of semiparametric transformation regression models, including popular models such as proportional and non-proportional hazards models. To promote sparsity and select variables specific to each event, we employ the broken adaptive ridge (BAR) penalty. This approach allows us to simultaneously select important risk factors and estimate their effects for each event under investigation. We establish the oracle property of the BAR procedure and evaluate its performance through simulation studies. The proposed method is applied to a real-life HIV cohort dataset, further validating its applicability in practice.
\end{abstract}


\section{Introduction}\label{sec:intro}
In many biomedical studies, it is not possible to observe the exact time of an event or failure, such as the onset of a disease in clinical studies. Instead, the event is only known to have occurred within a certain time interval determined by periodic clinical visits \citep{sun2006statistical}. This phenomenon is called interval censoring.

Another frequently arising complication in survival analysis is the occurrence of multiple events of interest in real-life problems. This complicated setting is known as ``competing risks data", where the occurrence of any one event precludes the other events from happening. For example, in a study on HIV/AIDS disease \citep{hudgens2001nonparametric}, two competing events are viral subtypes B and E. This HIV data set is our motivating example that involves interval-censored competing risks data, where a portion of the data has a missing cause of failure. In the analysis of this data set by \cite{hudgens2001nonparametric}, various risk factors are investigated in order to develop better precaution procedures. In this paper, we consider variable selection for a general type of interval-censored competing risks data while allowing for unknown/missing causes of failure.

A comprehensive review of different types of interval-censored data and the relevant methods can be found in \cite{sun2006statistical}. This type of censoring is known to be more difficult to analyze compared to other basic types, such as right-censoring. One of the main challenges is the development of efficient estimation procedures and the corresponding computational algorithms \citep{guo2014overview}.
For example, under the Cox model, the well-known partial likelihood method for right-censored data does not apply to interval-censored data, and a nuisance parameter must be estimated in addition to the regression coefficient parameters. \cite{finkelstein1986proportional} introduced semiparametric inference for general interval-censored data, proposing a method to jointly estimate the regression parameters and the baseline hazard function. \cite{zeng2006semiparametric} studied case II interval-censored data under the additive risk model. \cite{wang2016flexible} proposed a new method for analyzing interval-censored data under the proportional hazards model using monotone splines to approximate the cumulative baseline hazards function, and \cite{zeng2016maximum} extended the analysis of interval-censored data to a class of transformation models. Another semiparametric regression analysis for interval-censored data, including left-truncation and cure fraction, was done by \cite{shen2019semiparametric}. Recently, a new method was developed by \cite{zhou2022new} to fit the proportional hazards model to interval-censored failure time data with missing covariates, and their method addresses the challenges posed by the presence of interval censoring and missing data, and provides a practical solution for analyzing such complex data in survival analysis.

To extend inference to interval-censored competing risks data, \cite{li2016fine} used a sieve maximum likelihood estimation methodology with B-splines to model the baseline hazard functions of the cumulative incidence function (CIF) under the proportional subdistribution hazards (PSH) a.k.a. the Fine-Gray model \citep{fine1999proportional}. Following \cite{fine1999proportional}, \cite{bakoyannis2017semiparametric} proposed a class of semiparametric generalized odds rate transformation models for the cause-specific CIF. Similarly, \cite{mao2017semiparametric} considered a general class of semiparametric regression models for this type of data, incorporating potentially time-varying external covariates. This class includes both proportional and non-proportional subdistribution hazards structures, and the authors used nonparametric maximum likelihood estimation (NPMLE) while allowing for mixed-case interval censoring \citep{sun2006statistical} and partially missing information on the causes of failure. 

In biomedical studies, it is common practice to collect and maintain a considerable number of variables or risk factors in a study. However, incorporating all covariates in a regression model without filtering them based on their effectiveness may not be advantageous. This approach could lower the accuracy of the prediction and make interpretation more difficult \citep{nordhausen2009elements}. Variable selection methods have increasingly been used to tackle these issues. Among different variable selection techniques, regularization-based or penalized variable selection procedures are computationally efficient compared to traditional methods such as subset selection and forward/backward selection, and can handle estimation and variable selection simultaneously \citep{desboulets2018review}.

Variable selection has been applied to many different models and types of data. One of the pioneering works on variable selection for interval-censored data analysis was done by \cite{wu2015penalized} under a parametric model. More recently, \cite{zhao2019simultaneous} proposed a penalized variable selection method, namely Broken Adaptive Ridge regression (BAR) under the Cox regression model. Similarly, \cite{li2020adaptive} proposed the Adaptive LASSO (ALASSO) penalty for variable selection in interval-censored data analysis. \cite{li2020penalized} considered penalized estimation under a semiparametric transformation model and proposed a novel Expectation Maximization (EM) algorithm to incorporate the computation algorithm. A substantial review of existing methods for variable selection based on interval-censored data can be found in \cite{du2022variable}.

\cite{kuk2013model} extended variable selection to competing risks data by making the basic stepwise selection applicable to the PSH model. However, the journey of variable selection for competing risks data did not stop there. \cite{fu2017penalized} generalized several popular variable selection methods and their group versions to accommodate the PSH model. Later, \cite{ahn2018group} extended their work to an adaptive group bridge penalty and showed the consistency of such selection at both group and individual predictor levels . \cite{li2019variable} used quantile regression for variable selection in competing risks data.

Besides variable selection techniques for different models and data, various penalty functions have been proposed for penalized variable selection as well. Although all of them share the same objective of inducing sparsity, they feature different properties. Some of the most popular norm-based penalties are Ridge ($L_2$-based), Lasso ($L_1$-based), and Adaptive Lasso ($L_1$-based) proposed by \cite{hoerl1970ridge}, \cite{tibshirani1996regression}, and \cite{zou2006adaptive}, respectively. Among all the norms, the $L_0$ norm is known to impose a penalty on the cardinality of the predictor set directly, and it has the optimal performance in variable selection and parameter estimation \citep{shen2012likelihood}. However, despite its theoretical advantages, it is almost impossible to employ $L_0$-based variable selection methods such as Mallow's $C_p$ \citep{mallows2000some}, Akaike's information criterion (AIC) \citep{akaike1974new} or the Bayesian information criterion (BIC) \citep{schwarz1978estimating} for variable selection purposes in high-dimensional data. This is because of its limitation in computation, and instability with high-dimensional problems \citep{breiman1996heuristics}. This limitation in computation stems from the fact that it is non-convex in nature. Discovering the global optima of this problem requires an exhaustive combinatorial search for the best subset, which is computationally infeasible, even for data with moderate dimensions. A recently proposed penalty function, called Broken Adaptive Ridge, is a computationally scalable surrogate for $L_0$-penalized regression. It is an iteratively reweighted squared $L_2$-based penalty function that approximates the $L_0$ norm penalty. BAR takes advantage of this approximation and enjoys fast and efficient computation, as well as oracle properties. Since its proposal by \cite{liu2016efficient} for complete data, BAR has been studied under different models and data structures, including the Cox model for right-censored data \citep{kawaguchi2017scalable}, the linear model \citep{dai2018broken}, and the Cox model for interval-censored data \citep{zhao2019simultaneous}. BAR has also been extended to semiparametric transformation models by \cite{li2020penalized} and to semiparametric accelerated failure time models by \cite{sun2022broken}. It has been shown to possess several interesting features \citep{dai2018broken}: First, it produces a sparser and more accurate model compared to some other penalty functions. Second, it inherits the beneficial properties of the $L_0$ penalty while avoiding its pitfalls. Third, BAR has a closed-form solution, which makes it self-sufficient and independent of complicated algorithms such as the coordinate descent algorithm. Fourth, it is consistent and possesses oracle properties. Lastly, it has a grouping effect, which allows it to handle correlated predictors. 

In this study, our primary focus is to employ BAR for variable selection under a class of transformation models for interval-censored competing risks data. To achieve this, we propose an iteratively reweighted least squares algorithm that approximates the likelihood function as a least squares problem, followed by an optimization procedure to solve it.

The literature on variable selection for competing risks data has a focus on one cause of failure only. This leads to a lack of information on other causes of failure in this setting. It is desirable to assess the importance of variables in a model for all causes of failure jointly. One variable may stay in the model as it is important for one cause but not for other causes. 
In this paper, we aim to take all the risks in competing risks data into consideration, simultaneously, and propose a variable selection method for interval-censored competing risks data under a class of transformation models. Our new contributions can be considered from three aspects:
\begin{itemize}
    \item [1.] First, we use a semiparametric transformation regression model that makes our method flexible, as it contains popular models such as the proportional and non-proportional hazards models as special cases. Our method can handle variable selection and parameter estimation simultaneously. Our proposed method allows for the importance of assessment of variables for multiple risks (i.e., submodels) simultaneously, whereas the Fine-Gray model only incorporates one of the risks in the inference. For instance, the variable selection strategy in \cite{fu2017penalized}, which is built on the Fine-Gray model, is based on one of the risks only. 
    Furthermore, the Fine-Gray model requires the determination of the distribution of censoring in the model. The purposes of the proposed joint analysis are to avoid modeling the censoring distribution and gain efficiency.  
    
    
      \item [2.] The second aspect of our proposed variable selection method is the investigation of the oracle properties of BAR in the context of competing risks data. Our proofs have sharpened and improved the existing techniques in the literature for the BAR regression and yielded a semiparametric information bound for the sparse estimator of the regression parameters. 
      \item [3.] The third aspect of our proposed variable selection method is the use of BAR as a penalty function to enhance the estimation accuracy and the efficiency in computation. Employing BAR enables the variable selection procedure to enjoy a fast and efficient computational algorithm.
        

\end{itemize}
The rest of this paper is structured as follows. \Cref{unpensec2} includes an introduction to the data type, notations, and model, along with the proposed method for simultaneously variable selection and parameter selection using a penalized maximum likelihood approach. \Cref{subsubproj12} introduces a penalized EM algorithm implementing the proposed method.  \Cref{chap2asymp} outlines the asymptotic properties of the proposed method. Specifically, the proposed BAR estimators of regression parameters are proven to have the oracle property. \Cref{sim-project1} and \Cref{realdata-proj1} present the simulation studies, and real-life data analysis, respectively. \Cref{discussionproj1} includes the conclusion and discussion. Finally, we present the proofs of the asymptotic
properties in the Appendix. 

\section{Method for Penalized Variable Selection}\label{unpensec2}
We consider a study of $n$ independent subjects who are potentially exposed to experiencing one of the $K$ competing events of interest. Let $\textit{T}$ be a failure time with $\textit{K}$ competing risks (i.e., causes of failure) and suppose that $\textit{D}\in\{1,...,\textit{K}\}$ indicates the risk or cause of failure. Let $\textbf{Z}(\cdot)$ represent a $d_n$-vector of potentially time-varying external covariates and $\boldsymbol{\beta}=(\boldsymbol{\beta}_1^\top,\boldsymbol{\beta}_2^\top,\ldots,\boldsymbol{\beta}_K^\top)^\top$  denote a set of regression parameters corresponding to $K$ risks in the model, where $\boldsymbol{\beta}_k=(\beta_{1k}, \beta_{2k},\ldots,\beta_{d_nk})^\top$ corresponds to the regression parameters for the $k$th risk, $k=1,\ldots, K$,  and $d_n$ denotes the number of variables for each of the risks. The total number of regression coefficients is denoted by $p_n=Kd_n$. We assume $d_n\longrightarrow \infty$, then $p_n\longrightarrow \infty$. Within the framework of models that deal with multivariate survival data, there are three commonly used approaches to incorporate regression coefficients parameters ($\boldsymbol{\beta}$) and covariates ($\boldsymbol{Z}$) into the model:
\begin{itemize}
    \item [1.] Cause-specific regression coefficients parameters ($\boldsymbol{\beta}_k$) and cause-specific covariates ($\boldsymbol{Z}_k$) considered in  \cite{reeder2023penalized}, 
 which is the most general format and can be converted into the other two forms.
    \item [2.] Cause-specific regression coefficients parameters ($\boldsymbol{\beta}_k$) and a common covariate matrix ($\boldsymbol{Z}$). This is the approach employed in \cite{mao2017semiparametric}.
    \item [3.] A single long vector of regression coefficients parameters ($\boldsymbol{\beta}$) that contains all the different parameters in $\boldsymbol{\beta}_k$ vectors and cause-specific covariates ($\boldsymbol{Z}_k$). This is a common viewpoint utilized in many multivariate failure type models in the literature, such as those presented in \cite{lin1994cox} and \cite{sun2004additive}.
\end{itemize}
Throughout this work, we consider the second approach which provides an excellent foundation for variable selection as we can interpret the potential heterogenous effects of the same set of variables corresponding to each of the risks separately after variable selection given that we don't know which variable has an effect on which risk in the beginning. In addition, we model the competing risks data by the conditional subdistribution hazard function defined as 
\begin{eqnarray*}
    \lambda_k(t|\boldsymbol{Z})=\lim_{\Delta t\longrightarrow 0}\frac{1}{\Delta t}P(t\leq T< t+\Delta t,D=k|(T\geq t)\cup \{(T<t)\cap (D\neq k) \},\boldsymbol{Z}).
\end{eqnarray*}
Based on this conditional hazard function, we consider a general class of semiparametric regression models with time-dependent covariates, where the cumulative hazard function of $T$ given the time-dependent covariates $\textbf{Z}(\cdot)$ is defined by 
\begin{equation}
\label{cum hazard semipar}
\Lambda_k(t;\boldsymbol{Z})=G_k\left\{\int_0^t e^{\boldsymbol{\beta}_k^\top \boldsymbol{Z}(s)}d\Lambda_k(s)\right\},
\end{equation}
$G_k(\cdot)$ is a known increasing function and $\Lambda_k(\cdot)$ is an arbitrary increasing function with $\Lambda_k(0)=0$. The transformation function has the form $G_k(x)=-\log \int_0^\infty \exp (-x\zeta_k)\phi(\zeta_k)d\zeta_k$, where $\phi(\zeta_k)$ is a known density function on $[0,\infty)$. A popular choice for $\phi(\zeta_k)$ is the gamma density function with mean 1 and variance $r_k$ for $k=1,\ldots,K$. In this case, $G_k(x)$ falls into the class of logarithmic transformation functions described as
\begin{equation}
\label{transformation-function}
  G_k(x)=\begin{cases}
    \frac{1}{r_k}\log(1+r_k x), & r_k>0,\\
    x, & r_k=0.
  \end{cases}
\end{equation}
When $r_k=0$, the transformation model is corresponding to the Cox PH model and when $r_k=1$, it is the proportional odds model. 

Additionally, following \cite{mao2017semiparametric}, suppose there exists a random sequence of examination times denoted by $U_1<\cdots<U_J$. Define $\boldsymbol\Delta=(\Delta_1,\ldots,\Delta_J)^\top$ where $\Delta_j=I(U_{j-1}< T \leq U_j)$; $j=1, \ldots,J$ and $U_0=0$. In addition, define $\widetilde{D}$ as $DI(\boldsymbol{\Delta \neq 0})$ so that it represents the cause of failure for the events that are observed to happen between two examination times. Since we allow for missing causes of failure in this study, another variable, $\xi$, needs to be considered to account for the cause of failure being missing. Finally, the observed data for a random sample of $n$ subjects is $\mathcal{O}_i=(J_i,\boldsymbol{U}_i,\boldsymbol{\Delta}_i,\xi_i,\xi_i\widetilde{D}_i,\boldsymbol{Z}_i)$ where $i=1,\ldots,n$. For the $i$th subject :
\begin{itemize}
    \item [1.] $J_i$: the total number of examination times.
    \item [2.] $\boldsymbol{U}_i=(U_{i0}, U_{i1}, \ldots, U_{i,J_{i}})^\top$: the vector of examination times.
    \item [3.] $\boldsymbol{\Delta}_i=(\Delta_{i1}, \Delta_{i2}, \ldots, \Delta_{i,J_{i}})^\top$: the vector of zero and ones showing whether the event time was censored or observed between any of the examination times.
    \item [4.] $\xi_i$: takes the value of zero when the cause of failure is missing and one otherwise.
    \item [5.] $\xi_i\widetilde{D}_i$: takes the value of zero if the cause of failure is missing and $k$ if the cause of failure is known $(k=1,\ldots, K)$.
\end{itemize}
The NPMLE approach is utilized to estimate two sets of parameters in \eqref{cum hazard semipar}, $\boldsymbol{\beta}=(\boldsymbol{\beta}_1^\top,\ldots,\boldsymbol{\beta}_K^\top)^\top$ and $\boldsymbol{\Lambda}=(\Lambda_1,\ldots,\Lambda_K)$. Assuming that $(T,D) \perp \!\!\! \perp (\boldsymbol{U},J)$, conditional on $\boldsymbol{Z}(\cdot)$, the likelihood function is constructed using three contributions from three different scenarios:
(i) The event of interest is observed (no censoring), and the cause of failure is known (i.e., $k$):  $I(\xi_i \widetilde{D}_i=k, \Delta_{ij}=1)=1$.
(ii) The event of interest is observed, but the cause of failure is missing:
$I(\xi_i=0, \Delta_{ij}=1)=1$.
(iii) The event of interest is censored, and as a result, there is no information available on the cause of failure: $I(\boldsymbol{\Delta}_i=\boldsymbol{0})=1$. 
The observed likelihood function for $\boldsymbol{\beta}$ and $\Lambda$ can be expressed as follows.
\begingroup
\allowdisplaybreaks
\begin{eqnarray*}
\label{L_1}
L_n(\boldsymbol{\beta},\boldsymbol{\Lambda})&=&\prod_{i=1}^n\Bigg[\prod_{j=1}^{J_i} \prod_{k=1}^K\Bigg(\exp\Bigg[-G_k\Bigg\{\int_{0}^{U_{i,{j-1}}}e^{\boldsymbol{\beta}_k^\top \boldsymbol{Z}_{i}(t)}d\Lambda_k(t)\Bigg\}\Bigg]\\
&-&\exp\Bigg[-G_k\Bigg\{\int_{0}^{U_{ij}} e^{\boldsymbol{\beta}_k^\top \boldsymbol{Z}_{i}}d\Lambda_k(t)\Bigg\}\Bigg]\Bigg)^{I(\xi_i \widetilde{D}_i=k,\Delta_{ij}=1)}\\
&\times&\Bigg\{ \sum_{k=1}^K\Bigg(\exp\Bigg[-G_k\Bigg\{\int_{0}^{U{i,{j-1}}}e^{\boldsymbol{\beta}_k^\top \boldsymbol{Z}_{i}(t)}d\Lambda_k(t)\Bigg\}\Bigg]\\
&-&\exp\Bigg[-G_k\Bigg\{\int_{0} ^{U_{ij}}e^{\boldsymbol{\beta}_k^\top \boldsymbol{Z}_{i}(t)}d\Lambda_k(t)\Bigg\}  \Bigg]\Bigg)\Bigg\}^{I(\xi_i=0,\Delta_{ij}=1)}\\
&\times&\Bigg(\sum_{k=1}^K \exp\Bigg[-G_k\Bigg\{\int_{0}^{U_{i,J_i}} e^{\boldsymbol{\beta}_k^\top\boldsymbol{ Z}_{i}(t)}d\Lambda_k(t)\Bigg\}\Bigg]-K+1\Bigg)^{I(\boldsymbol{\Delta}_i=0)}\Bigg].
\end{eqnarray*}
\endgroup
Now, assume that $(L_i, R_i]$ is the interval among $(U_{i0},U_{i1}],\ldots,(U_{i,J_i},\infty]$ that contains $T_i$, and let 
$t_{kj}$ $(j=1,\ldots,m_k)$ denote the distinct values of $L_i$ and $R_i$ with $\xi_i\widetilde{D}i=k$ or $\xi_i=0$. In addition, assume that $\lambda_{kj}$ is the size of the jump at $t_{kj}$ where $t_{k1}<\ldots<t_{k,m_k}$ for $k=1,2,\ldots,K$ and $j=1,\ldots,m_k$. Therefore, $\boldsymbol{Z}_{ikj}=\boldsymbol{Z}_i(t_{kj})$, and then the likelihood function can be expressed as follows,
\begingroup
\allowdisplaybreaks
\begin{eqnarray}
\label{L_1}
L_n(\boldsymbol{\beta},\boldsymbol{\Lambda})&=&\prod_{k=1}^K \prod_{i:\xi_i \widetilde{D}_i=k}\Bigg[\exp\Bigg\{-G_k\Bigg(\sum_{t_{kj}\leq L_i}\lambda_{kj}e^{\boldsymbol{\beta}_k^\top \boldsymbol{Z}_{ikj}}\Bigg)\Bigg\}\nonumber\\
&-&\exp\Bigg\{-G_k\Bigg(\sum_{t_{kj}\leq R_i} \lambda_{kj}e^{\boldsymbol{\beta}_k^\top \boldsymbol{Z}_{ikj}}\Bigg) \Bigg\}\Bigg]\nonumber\\
&\times& \prod _{i:\xi_i=0}\Bigg(\sum_{k=1}^K \Bigg[\exp\Bigg\{ -G_k\Bigg(\sum_{t_{kj}\leq L_i}\lambda_{kj}e^{\boldsymbol{\beta}_k^\top \boldsymbol{Z}_{ikj}}\Bigg)\Bigg\}\nonumber\\
&-&\exp \Bigg\{-G_k\Bigg(\sum_{t_{kj}\leq R_i} \lambda_{kj}e^{\boldsymbol{\beta}_k^\top \boldsymbol{Z}_{ikj}}\Bigg)  \Bigg\}\Bigg]\Bigg)\nonumber\\
&\times &\prod_{i:R_i=\infty} \Bigg[\sum_{k=1}^K \exp\Bigg\{ -G_k\Bigg(\sum_{t_{kj}\leq L_i} \lambda_{kj}e^{\boldsymbol{\beta}_k^\top \boldsymbol{Z}_{ikj}} \Bigg)\Bigg\}-K+1\Bigg].
\end{eqnarray}
\endgroup
In order to construct the objective function for variable selection and estimation, let 
\begin{eqnarray}
    \label{loglikforproof}\ell_n(\boldsymbol{\beta},\boldsymbol{\Lambda})=\log\{L_n(\boldsymbol{\beta},\boldsymbol{\Lambda})\}.
\end{eqnarray}
For fixed $\boldsymbol{\beta}$, denote $\widehat{\boldsymbol{\Lambda}}(\boldsymbol{\beta})=\text{argmax}_{\boldsymbol{\Lambda}}\ell_n(\boldsymbol{\beta},
\boldsymbol{\Lambda})$. We define profile log-likelihood as 
\begin{eqnarray*}
    \ell_p(\boldsymbol{\beta})=\max_{\boldsymbol{\Lambda}}\ell_n(\boldsymbol{\beta},\boldsymbol{\Lambda})=\ell_n(\boldsymbol{\beta},\widehat{\boldsymbol{\Lambda}}(\boldsymbol{\beta})).
\end{eqnarray*}
We propose to adopt the penalized likelihood method by minimizing the following penalized objective function:
\begin{eqnarray}
    \label{loglikbasedobj}
    -\ell_p(\boldsymbol{\beta})+\sum_{k=1}^{K}\sum_{j=1}^{d_n}p_{\tau_n}(|\beta_{jk}|),
\end{eqnarray}
where $p_{\tau_n}(\cdot)$ denotes a penalty function and $\tau_n$ is a non-negative tuning parameter that controls the model's complexity. Directly Minimizing \eqref{loglikbasedobj} is challenging because the parameters are high-dimensional and there is not a closed-form solution. 

Our proposed BAR method iteratively performs the following penalized likelihood  estimation, 
\begin{eqnarray*}
  \hat{\boldsymbol{\beta}}^{(m+1)}=
  \arg \min  \ell_{pp}(\boldsymbol{\beta}|\hat{\boldsymbol{\beta}}^{(m)})
  \equiv   \arg \min \left\{
  -\ell_p(\boldsymbol{\beta})+\tau_n\sum_{k=1}^{K}\sum_{j=1}^{d_n}\frac{\beta_{jk}^2}{(\hat{\beta}_{jk}^{(m)})^2}
  \right\},
\end{eqnarray*}
where $\boldsymbol{\beta}^{(0)}$ represents a consistent estimator of $\boldsymbol{\beta}$ with all the components being non-zero. 
Below we will discuss how we obtain this consistent estimator. 
If the iterative estimation converges numerically, i.e., $\hat{\boldsymbol{\beta}}^{(m)}$ converges to some $\hat{\boldsymbol{\beta}}^*$ as $m\to\infty$,
we expect 
\begin{eqnarray*}
(\hat{\beta}_{jk}^{(m+1)})^2 / (\hat{\beta}_{jk}^{(m)})^2
\rightarrow 
I(\hat{\beta}_{jk}^* \neq 0). 
\end{eqnarray*}
as $m$ goes to infinity. Hence, BAR is considered a surrogate for the $L_0$-penalization approach, which is generally viewed as impractical due to being an NP-hard problem. BAR has been shown to possess the desirable features of $L_0$ norm penalization while avoiding its computational infeasibility \citep{dai2018broken}. Additionally, BAR involves an adaptively reweighted procedure that allows for the weighted penalty strength to be intensified for zero components and reduced for nonzero ones simultaneously. This is the reason why BAR is powerful in selecting relevant variables in a variable selection problem. We also consider the Lasso penalty \citep{tibshirani1997lasso} function defined as 
\begin{eqnarray*}
    p_{\tau_n}(|\beta_{jk}|)=\tau_n|\beta_{jk}|,
\end{eqnarray*} 
and ALasso penalty given by 
\begin{eqnarray*}
    p_{\tau_n}(|\beta_{jk}|)=\tau_n\frac{|\beta_{jk}|}{|\widetilde{\beta}_{jk}|^\psi},
\end{eqnarray*}
where $\widetilde{\beta}_{jk}$ is a consistent estimator of $\beta_{jk}$ \citep{zou2006adaptive} and $\psi>0$ is a constant, usually, $\psi=1$ is taken.

\section{Variable Selection Based on the EM algorithm}
\label{subsubproj12}
For the estimation of the parameters in the model under the case of fixed dimension $d_n=d$, \cite{mao2017semiparametric} introduced a novel EM algorithm that extends Turnbull's self-consistency formula to regression analysis with interval-censored competing risks. 
Based on their unpenalized estimation procedure, we propose an EM-embedded method for simultaneous variable selection and parameter estimation to eliminate the computation burden.
To construct the complete-data log-likelihood in the EM algorithm, let $N_{ki}(s_{ik,j-1},s_{ikj}]$ count the number of events of $k$th type that have happened in the interval of $(s_{ik,j-1},s_{ikj}]$ for the $i$th subject, and the sub-intervals are defined by partitioning the interval $(L_i, R_i]$ into $(s_{ik0},s_{ik1}],\ldots,(s_{ik,j_{ik}-1},s_{ik,j_{ik}}]$. Here, $s_{ik0}<\ldots<s_{ik,j_{ik}}$ represent the distinct values of $t_{kj}$ in the interval $(L_i, R_i]$. Then, treating $N_{ki}$ as unobserved data, the complete-data log-likelihood can be expressed as
\begin{eqnarray}
\label{completelogproj1}
    &&\sum_{i=1}^n\Big\{\sum_{k=1}^K\sum_{j=1}^{j_{ik}}I(R_i<\infty)N_{ki}(s_{ik,j-1},s_{ikj}]\log\Delta F(s_{ikj};\boldsymbol{Z}_i,\boldsymbol{\beta}_k,\Lambda_k)\nonumber\\
    &&+I(R_i=\infty)\log S(L_i;\boldsymbol{Z}_i,\boldsymbol{\beta},\boldsymbol{\Lambda})\Big\},
\end{eqnarray}
where $F_k(t;\boldsymbol{Z}_i,\boldsymbol{\beta}_k,\Lambda_k)=1-\exp\{-\Lambda_k(t|\boldsymbol{Z}_i)\}$, $S(t;\boldsymbol{Z}_i,\boldsymbol{\beta},\boldsymbol{\Lambda})=1-\sum_{k=1}^K F_k(t;\boldsymbol{Z}_i,\boldsymbol{\beta}_k,\Lambda_k)$ is the overall survival function, 
and 
$\Delta F_k(t;\boldsymbol{Z}_i,\boldsymbol{\beta}_k,\Lambda_k)$ is the jump size of $F_k(\cdot;\boldsymbol{Z}_i,\boldsymbol{\beta}_k,\Lambda_k)$ at $t$.
Let $\widetilde{\omega}_{ikj}$ be the conditional probability that the $i$th subject experiences a failure of the $k$th cause within the interval $(s_{ik,j-1},s_{ikj}]$ given the subject's failure information.
If $\xi_i \widetilde{D}_i=k^{\prime}$, then
\begin{eqnarray*}
\widetilde{\omega}_{ikj}&=&E\left\{ N_{ki}(s_{ik,j-1},s_{ikj}]\Bigg| N_{k^{'}
i}(L_i,R_i]=1 \right\}\nonumber\\
&=&
I(k=k^{'})\frac{\Delta F_k(s_{ikj};\boldsymbol{Z}_i,\boldsymbol{\beta}_k,\Lambda_k)}{\sum_{l=1}^{j_{ik}}\Delta F_k(s_{ikl};\boldsymbol{Z}_i,\boldsymbol{\beta}_k,\Lambda_k)},
\end{eqnarray*}
and if $\xi_i=0$, then
\begin{eqnarray*}
\widetilde{\omega}_{ikj}&=&E\left\{ N_{ki}(s_{ik,j-1},s_{ikj}]\Bigg|\sum_{k^{'}=1}^{K}N_{k^{'}
i}(L_i,R_i]=1 \right\}\nonumber\\
&=&
\frac{\Delta F_k(s_{ikj};\boldsymbol{Z}_i,\boldsymbol{\beta}_k,\Lambda_k)}{\sum_{k'=1}^{K}\sum_{l=1}^{j_{ik'}}\Delta F_k(s_{ik'l};\boldsymbol{Z}_i,\boldsymbol{\beta}_{k'},\Lambda_{k'})}.
\end{eqnarray*}
Finally, if $R_i=\infty$, then $\widetilde{\omega}_{ikj}=0$.

Thus, in the second step of the EM algorithm (maximization step), we aim to maximize
\begin{eqnarray}
\label{M-step}
&&\sum_{i=1}^n\left\{\sum_{k=1}^K \sum_{j=1}^{j_{ik}} \widetilde{\omega}_{ikj} \log \Delta F_k(s_{ikj};\boldsymbol{Z}_i,\boldsymbol{\beta}_k,\Lambda_k)\right\}\nonumber\\
&&+I(R_i=\infty) \log S(L_i;\boldsymbol{Z}_i,\boldsymbol{\beta},\boldsymbol{\Lambda}).
\end{eqnarray}
By utilizing the first-order approximation of $\Delta F_k(s_{ikj};\boldsymbol{Z}_i,\boldsymbol{\beta}_k,\Lambda_k)$ as
\begin{eqnarray*}
    \widetilde{G}_k(\sum_{j'=1}^j e^{\boldsymbol{\beta}_k^\top\boldsymbol{Z}_{ikj'}}\lambda_{kj'})e^{\boldsymbol{\beta}_k^\top \boldsymbol{Z}_{ikj}}\lambda_{kj},
\end{eqnarray*} 
we rewrite the objective function in \eqref{M-step} as:
\begin{eqnarray}
\label{Q-function}
\ell_n^{\ast}(\boldsymbol{\beta},\{\lambda_{kj}\})&=&\sum_{k=1}^K\sum_{i=1}^n\sum_{j=1}^{m_k} \widetilde{\omega}_{ikj}\left\{\log 
\lambda_{kj}+\boldsymbol{\beta}_k^\top\boldsymbol{Z}_{ikj} +\log \widetilde{G}_k\left( \sum_{j'=1}^je^{\boldsymbol{\beta}_k^\top\boldsymbol{Z}_{ikj'}}\lambda_{kj'}\right) \right\}\nonumber\\
&+&\sum_{i:R_i=\infty}\log\left[\sum_{k=1}^K\exp\left\{ -G_k\left(\sum_{t_{kj}\leq L_i}\lambda_{kj} e^{\boldsymbol{\beta}_k^\top\boldsymbol{Z}_{ikj}}\right)\right\} -K+1 \right],
\end{eqnarray}
where $\widetilde{G}_k(x)=G^{(1)}_k(x)e^{-G_k(x)}$, $G_k^{(1)}(x)$ denotes the first derivative of $G_k(x)$ with respect to $x$. 

To estimate $\lambda_{kj}$, we fix $\boldsymbol{\beta}$ and set the derivative of \eqref{completelogproj1} with respect to $\lambda_{kj}$ to zero to obtain an updating formula for $\lambda_{kj}$ below.  
\begin{eqnarray}
\label{updatelamproj1}
\widetilde{\lambda}_{kj}(\boldsymbol{\beta})&=&\left(\sum_{i=1}^n \widetilde{\omega}_{ikj}\right)\Bigg[\sum_{i=1}^n\sum_{j^{\prime}=j}^{m_k} \widetilde{\omega}_{ikj^{\prime}}e^{\boldsymbol{\beta}_k^{\top}\boldsymbol{Z}_{ikj^{\prime}}}\frac{\widetilde{G}^{(1)} _k}{\widetilde{G}_k}\left( \sum_{j^{\prime \prime}=1}^{j^{\prime}}e^{\boldsymbol{\beta}_k^{\top}\boldsymbol{Z}_{ikj^{\prime \prime}}}\lambda_{kj^{\prime \prime}}\right)\nonumber\\
&+&\sum_{i:R_i=\infty,L_i\geq t_{kj}}S(L_i;\boldsymbol{Z}_i,\boldsymbol{\beta},\boldsymbol{\Lambda})^{-1}\widetilde{G}_k\left(\sum_{t_{kj^{\prime}}\leq L_i}e^{\boldsymbol{\beta}_k^{\top}\boldsymbol{Z}_{ikj^{\prime }}}\lambda_{kj^{\prime}}\right) e^{\boldsymbol{\beta}_k^{\top}\boldsymbol{Z}_{ikj}}   \Bigg]^{-1}.
\end{eqnarray}
For the estimation of $\boldsymbol{\beta}$, plug $\widetilde{\lambda}_{kj}(\boldsymbol{\beta})$ into \eqref{Q-function}, and obtain the profile log-likelihood for $\boldsymbol{\beta}$ in \eqref{finalQ-function}. Then, using a one-step Newton-Raphson algorithm can lead to the estimate of $\boldsymbol{\beta}$. The algorithm is cycled among 
$\{\widetilde{\omega}_{ikj}\}$, $\boldsymbol{\beta}$  and $\{\widetilde{\lambda}_{kj(\boldsymbol{\beta})} \}$.
Note that although we denote \eqref{Q-function} as $\ell_n$, it is not derived by taking the logarithm of the likelihood function \eqref{L_1}, but it is the objective function in the M-step of the EM algorithm. To facilitate the computation, our variable selection procedure is embedded in the EM algorithm by using this objective function.
In order to obtain a sparse estimator for $\boldsymbol{\beta}$, it is necessary to minimize the penalized objective function shown in \eqref{Qbeta}.
Given an initial value of $\boldsymbol{\beta}$, we compute $\{\widetilde{\lambda}_{kj} \}$ and $\{\widetilde{\omega}_{ikj} \}$ and fix them at the current values, then, we construct the following profile objective function that is going to be used in our penalized variable selection optimization problem,
\begin{eqnarray}
\label{finalQ-function}
\ell_{p}^{\ast}(\boldsymbol{\beta})&=&\sum_{k=1}^K\sum_{i=1}^n\sum_{j=1}^{m_k} \widetilde{\omega}_{ikj}\left\{\log 
\widetilde{\lambda}_{kj}+\boldsymbol{\beta}_k^\top\boldsymbol{Z}_{ikj} +\log \widetilde{G}_k\left( \sum_{j'=1}^je^{\boldsymbol{\beta}_k^\top\boldsymbol{Z}_{ikj'}}\widetilde{\lambda}_{kj'}\right) \right\}\nonumber\\
&+&\sum_{i:R_i=\infty}\log\left[\sum_{k=1}^K\exp\left\{ -G_k\left(\sum_{t_{kj}\leq L_i}\widetilde{\lambda}_{kj} e^{\boldsymbol{\beta}_k^\top\boldsymbol{Z}_{ikj}}\right)\right\} -K+1 \right].
\end{eqnarray}
During the penalized estimation procedure, this profile objective function is updated by pluging-in the newly estimated $\boldsymbol{\beta}$ values into $\{\widetilde{\lambda}_{kj} \}$ and $\{\widetilde{\omega}_{ikj}\}$ and keeping $\boldsymbol{\beta}$  shown in the expression of $\ell_{p}^{\ast}(\boldsymbol{\beta})$ as the argument of the function. 
By utilizing \eqref{finalQ-function}, we propose to minimize the penalized profile objective function for variable selection, which is defined as
\begin{eqnarray}
    \label{Qbeta}
    \ell_{pp}^{\ast}(\boldsymbol{\beta})&=&-\ell_{p}^{\ast}(\boldsymbol{\beta})+\sum_{k=1}^{K}\sum_{j=1}^{d_n} p_{\tau_n}(|\beta_{jk}|)\nonumber\\
    &=&-\ell_{p}^{\ast}(\boldsymbol{\beta})+\sum_{a=1}^{p_n}p_{\tau_n}(|\beta_{a}|),
\end{eqnarray}
where $\{\beta_a\}|_{1\leq a\leq p_n}=\{ \beta_{jk}\}|_{1\leq j \leq d_n,~1\leq k \leq K}$.
To obtain the penalized estimator, we propose minimizing \eqref{Qbeta}. To solve \eqref{Qbeta} with different penalty functions, we need to employ different optimization algorithms. For LASSO and ALASSO, we employ the well-known shooting algorithm  \citep{fu1998penalized},
and the modified shooting algorithm proposed by \cite{zhang2007adaptive}, respectively. For BAR, a closed-form solution as described below can be used instead of utilizing complex computational algorithms,  which significantly simplifies the computation process. 

Our strategy is to approximate the profile objective function \eqref{finalQ-function} by a second-order Taylor expansion and solve an iterative reweighted least square problem subject to penalties at each iteration. For any fixed tuning parameter $\tau_n$, we propose the following computation algorithm to minimize the objective function
\begin{itemize}
    \item [\textit{Step 1.}] Follow the EM algorithm described in \Cref{subsubproj12} by choosing the initial estimators $\boldsymbol{\beta}^{(0)}=\boldsymbol{0}$ and $\lambda_{kj}^{(0)}=1/n$ for $j=1,\ldots,m_k$ and $k=1,\ldots,K$, and obtain the estimates of $\boldsymbol{\beta}$, $\boldsymbol{\lambda}$, and $\boldsymbol{\omega}$ as $\widetilde{\boldsymbol{\beta}}$, $\widetilde{\boldsymbol{\lambda}}$, and $\widetilde{\boldsymbol{\omega}}$ without imposing a penalty or by an initial ridge regression estimator as \cite{kawaguchi2020surrogate} did in their work. 
    \item [\textit{Step 2.}] At the step 0, fix $\widetilde{\boldsymbol{\Phi}}=(\widetilde{\boldsymbol{\lambda}},\widetilde{\boldsymbol{\omega}})$ and set the initial estimator $\widehat{\boldsymbol{\beta}}^{(0)}=\widetilde{\boldsymbol{\beta}}=(\widetilde{\boldsymbol{\beta}}_1^\top,\ldots,\widetilde{\boldsymbol{\beta}}_K^\top)^\top$.
    \item [\textit{Step 3.}] At step $m+1$, compute the following four components including $\boldsymbol{u}(\boldsymbol{\beta})$, $\boldsymbol{H}(\boldsymbol{\beta})$, $\boldsymbol{X}(\boldsymbol{\beta})$, and $\boldsymbol{W}(\boldsymbol{\beta})$ based on the current value of $\widehat{\boldsymbol{\beta}}^{(m)}$. The gradient vector is presented by $\boldsymbol{u}(\boldsymbol{\beta})$:
    \begin{eqnarray*}
    \boldsymbol{u}(\boldsymbol{\beta})=(\boldsymbol{u}^\top_1(\boldsymbol{\beta}),\ldots, \boldsymbol{u}^\top_K(\boldsymbol{\beta})=(\partial \ell_p^{\ast}(\boldsymbol{\beta})/{\partial \boldsymbol{\beta}_1^\top},\ldots,\partial \ell_p^{\ast}(\boldsymbol{\beta})/{\partial \boldsymbol{\beta}_K^\top})^\top_{1\times Kd_n}
    \end{eqnarray*}
      with $Kd_n$ elements denoting the number of regression coefficient parameters for all the $K$ risks (in this work, we consider $K=2$) in total.   
      Hessian matrix $\boldsymbol{H}(\boldsymbol{\beta})$ is the second derivative of $\ell_p^{\ast}(\boldsymbol{\beta})$ given by
      \begin{eqnarray*}
    \boldsymbol{H} (\boldsymbol{\beta})=\begin{bmatrix}   {\boldsymbol{H}^{11}_{(d_n\times d_n)}(\boldsymbol{\beta})} & \ldots&{\boldsymbol{H}^{1K}_{(d_n\times d_n )}(\boldsymbol{\beta})}\\
    \vdots&\ddots&\vdots\\
    {\boldsymbol{H}^{K1}_{(d_n\times d_n )}(\boldsymbol{\beta})} & \ldots&{\boldsymbol{H}^{KK}_{(d_n\times d_n )}(\boldsymbol{\beta})}\\
    \end{bmatrix}_{(Kd_n\times Kd_n)},
      \end{eqnarray*}
    where $\boldsymbol{H}^{kk'}(\boldsymbol{\beta})=\partial^2\ell_p^{\ast}(\boldsymbol{\beta})/{{\partial\boldsymbol{\beta}_k\partial\boldsymbol{\beta}_{k'}^\top}}$ that is a square matrix with $1\le k,k'\le d_n$. 
The pseudo response vector, denoted by $\boldsymbol{W}(\boldsymbol{\beta})$, is calculated as follows
    $$\boldsymbol{W}(\boldsymbol{\beta}) =(\boldsymbol{X^\top}(\boldsymbol{\beta}))^{-1}\left\{ -\boldsymbol{H}(\boldsymbol{\beta}) \boldsymbol{\beta} +\boldsymbol{u}(\boldsymbol{\beta})\right\},$$
    where $-\boldsymbol{H}(\boldsymbol{\beta}) =\boldsymbol{X}^\top (\boldsymbol{\beta})\boldsymbol{X}(\boldsymbol{\beta})$, and $\boldsymbol{X}(\boldsymbol{\beta})$ is an upper triangular matrix that is obtained through the Cholesky decomposition of $\boldsymbol{H}(\boldsymbol{\beta})$.
    The matix $\boldsymbol{X^\top}(\boldsymbol{\beta})$ 
    may not be invertible, $(\boldsymbol{X^\top}(\boldsymbol{\beta}))^{-1}$ represents the generalized inverse. 
    \item [\textit{Step 4.}] Approximate $-\ell_p^{\ast}(\boldsymbol{\beta})$ by the second-order Taylor expansion as
    \begin{eqnarray*}
        -\ell_p^{\ast}(\boldsymbol{\beta})=\frac{1}{2} (\boldsymbol{W}(\boldsymbol{\beta}) -\boldsymbol{X}(\boldsymbol{\beta})\boldsymbol{\beta})^\top (\boldsymbol{W}(\boldsymbol{\beta}) -\boldsymbol{X}(\boldsymbol{\beta})\boldsymbol{\beta}).
    \end{eqnarray*}
    \item [\textit{Step 5.}] Minimize the objective function, \eqref{Qbeta} by substituting its approximation for $-\ell_p^{\ast}(\boldsymbol{\beta})$ in the previous step. The closed-form solution of BAR to obtain the penalized estimate at each step is
    \begin{eqnarray*}
    \widehat{\boldsymbol{\beta}}^{(m+1)}&=&
    \left\{\boldsymbol{X}(\boldsymbol{\beta})^\top\boldsymbol{X}(\boldsymbol{\beta})+\tau_n \boldsymbol{D}(\boldsymbol{\beta}) \right\}^{-1}\boldsymbol{X}^\top (\boldsymbol{\beta})\boldsymbol{W}(\boldsymbol{\beta}),
    \end{eqnarray*}
    where 
    \begin{eqnarray*}
    \boldsymbol{D}(\boldsymbol{\beta})=\text{diag}\left(
    \frac{1}{
    {{\beta}}_{11}^2},\ldots,
    \frac{1}{
    {{\beta}}_{d_n1}^2},\ldots,
    \frac{1}{{\beta}_{1K}^2},\ldots,
    \frac{1}{
    {\beta}_{d_n K}^2}\right)
    \end{eqnarray*}
    is a square matrix with $Kd_n$ rows and columns and $\boldsymbol{\beta}=\widehat{\boldsymbol{\beta}}^{(m)}$ which is the penalized estimate of $\boldsymbol{\beta}$ at the $m$th step.
    
    Note: As the successive values of $\beta_{jk}$ for $j=1,\ldots,d_n$ approach their limit, the weight matrix $\boldsymbol{D}(\boldsymbol{\beta})$ will inevitably encounter a situation where division by an extremely small non-zero value occurs, potentially leading to a so-called arithmetic overflow \citep{dai2018broken, kawaguchi2020surrogate}. To address this issue, a commonly adopted solution involves introducing a slight perturbation. Specifically, the matrix $\boldsymbol{D}(\boldsymbol{\beta})$ is replaced by 
    \begin{eqnarray*}
    \text{diag}\left(
    \frac{1}{
    {({\beta}}_{11}^2+\delta^2)},\ldots,
    \frac{1}{
    {({\beta}}_{d_n1}^2+\delta^2)},\ldots,
    \frac{1}{({\beta}_{1K}^2+\delta^2)},\ldots,
    \frac{1}{
    {(\beta}_{d_n K}^2+\delta^2)}\right),    
    \end{eqnarray*}
    where $\delta=10^{-6}$ in our study, to prevent numerical instability.
    \item [\textit{Step 6.}] Update $\lambda_{jk}$ at the $(m+1)$th step based on \eqref{updatelamproj1} as well as $\omega_{ikj}$.
    \item [\textit{Step 7.}] Return to \textit{Step 3} and repeat the procedure until the convergence criterion is satisfied. The penalized BAR estimator can then be obtained by iterating the above procedure until convergence is achieved, i.e., $\widehat{\boldsymbol{\beta}}^{*}=\lim_{m \to \infty}\widehat{\boldsymbol{\beta}}^{(m)}$.
    In our numerical studies, the convergence criterion is set to stop the iteration when $\|\widehat{\boldsymbol{\beta}}^{(m+1)}-\widehat{\boldsymbol{\beta}}^{(m)}\|<10^{-6}$.
\end{itemize}
Note that this algorithm can be readily used for many different penalty functions. The difference would be in the last two steps where one needs to utilize appropriate optimization algorithms (e.g., shooting algorithm) instead of BAR's closed-form solution. 

The selection of the tuning parameter $\tau_n$ is essential in implementing the proposed penalized  variable selection method. The performance of variable selection is highly influenced by the tuning parameter value as this parameter controls the balance between the goodness of fit and sparsity of the model. Very large values of $\tau_n$ result in all the parameters becoming zero while very small values do not provide sufficient sparsity in the model. Therefore, it is crucial to use an appropriate method to find the optimal tuning parameter. Various data-driven methods, such as the Akaike information criterion (AIC), the Bayesian information criterion (BIC), and generalized cross-validation (GCV), can be used to choose the tuning parameter. We propose to use the generalized cross-validation (GCV) method \citep{craven1978smoothing}. The GCV method was initially introduced to reduce the computational burden by weighting the ordinary leave-one-out cross-validation. Subsequently, the GCV method was adapted to perform tuning parameter selection in variable selection, as proposed by \cite{cai2005variable, fan2001variable, huang2009group}. It is defined as
\begin{equation}
    \text{GCV}(\tau_n,\widehat{\boldsymbol{\beta}})=\frac{-\ell_p^{\ast}( \widehat{\boldsymbol{\beta}})}{n\left[1-s(\tau_n,\boldsymbol{\widehat{\beta}})/n\right]^2},
\end{equation}
where $\widehat{\boldsymbol{\beta}}$ represents the vector of the penalized estimates, and $s(\tau_n,\widehat{\boldsymbol{\beta}})=\text{tr}\{(\boldsymbol{H}(\widehat{\boldsymbol{\beta}})+\eta(\tau_n,\widehat{\boldsymbol{\beta}}))^{-1}\boldsymbol{H}(\widehat{\boldsymbol{\beta}})\}$ is the number of effective parameters.  \begin{eqnarray*}
\eta(\tau_n,\widehat{\boldsymbol{\beta}})=\tau_n r(\widehat{\boldsymbol{\beta}}),
\end{eqnarray*}
and 
\begin{eqnarray*}
r(\widehat{\boldsymbol{\beta}})=\text{diag}\Bigg(\frac{\nabla_{\beta_{1,1}} p_{\tau_n}(\widehat{\boldsymbol{\beta}})}{|\widehat{\beta}_{11}\mid},\ldots,\frac{\nabla_{\beta_{d_n1}} p_{\tau_n}(\widehat{\boldsymbol{\beta}})}{\mid\widehat{\beta}_{d_n1}\mid},\ldots,\frac{\nabla_{\beta_{1K}} p_{\tau_n}(\widehat{\boldsymbol{\beta}})}{\mid\widehat{\beta}_{1K}\mid},\ldots,\frac{\nabla_{\beta_{d_nK}} p_{\lambda_n}(\widehat{\boldsymbol{\beta}})}{\mid\widehat{\beta}_{d_nK}\mid}\Bigg),
\end{eqnarray*}
$\nabla$ denotes the first derivative of the penalty function $p_{\tau_n}$ with respect to elements of $|\boldsymbol{\beta}|$.
\section{The Asymptotic Properties of the Proposed BAR Penalized Estimator}
\label{chap2asymp}
This section focuses on studying the behaviour of the proposed BAR estimator denoted as $\widehat{\boldsymbol{\beta}}^*$, as the sample size approaches infinity. The aim is to investigate the asymptotic properties of the estimator. In order to do this, we denote the true values of $\boldsymbol{\beta}$ by $$\boldsymbol{\beta}_0=(\boldsymbol{\beta}^\top_{01},\ldots,\boldsymbol{\beta}^\top_{0K})^\top=(\beta_{0,1,1},\ldots,\beta_{0,d_n,1},\ldots,\beta_{0,1,K},\ldots,\beta_{0,d_n,K})^\top.$$ 
Now, without loss of generality, assume 
\begin{eqnarray*}
\boldsymbol{\beta}_0={(\boldsymbol{\beta}_{0s1}^\top,\boldsymbol{\beta}_{0s2}^\top)}^\top,
\end{eqnarray*} 
where $\boldsymbol{\beta}_{0s1}={(\beta_{0s1,1},\ldots,\beta_{0s1,q_n})}^\top$ and $\boldsymbol{\beta}_{0s2}={(\beta_{0s2,q_n+1},\ldots,\beta_{0s2,p_n})}^\top$. $\boldsymbol{\beta}_{0s1}$ represents $q_n$ ($q_n<<p_n$) nonzero true values and $\boldsymbol{\beta}_{0s2}$ denotes the zero true values among all $K$ risks. Additionally, we define $\boldsymbol{\beta}={(\boldsymbol{\beta}_{s1}^\top,\boldsymbol{\beta}_{s2}^\top)}^\top$ and let 
${\widehat{\boldsymbol{\beta}}}^*=(\widehat{\boldsymbol{\beta}}_{s1}^{* \top},
\widehat{\boldsymbol{\beta}}_{s2}^{* \top})^\top$ denote the BAR estimator.

Define 
\begin{eqnarray*}
    \boldsymbol{\Omega}_n(\boldsymbol{\beta})=-\Ddot{\ell}_p(\boldsymbol{\beta})=\boldsymbol{X}^\top(\boldsymbol{\beta})\boldsymbol{X}(\boldsymbol{\beta})
\end{eqnarray*}
as the Cholesky decomposition, and
\begin{eqnarray*}
    \boldsymbol{v}_n(\boldsymbol{\beta})=\dot{\ell}_p(\boldsymbol{\beta})-\Ddot{\ell}_p(\boldsymbol{\beta})\boldsymbol{\beta}.
\end{eqnarray*}
Let 
\begin{eqnarray*}
\boldsymbol{\Omega}_n(\boldsymbol{\beta}_{s1})=\boldsymbol{\Omega}_n(\boldsymbol{\beta})\Big|_{\boldsymbol{\beta}=(\boldsymbol{\beta}_{s1}^\top,\boldsymbol{\beta}_{s2}^\top=\boldsymbol{0}^\top)^\top}
\end{eqnarray*}
and 
\begin{eqnarray*}
\boldsymbol{v}_n(\boldsymbol{\beta}_{s1})=\boldsymbol{v}_n(\boldsymbol{\beta})\Big|_{\boldsymbol{\beta}=(\boldsymbol{\beta}_{s1}^\top,\boldsymbol{\beta}_{s2}^\top=\boldsymbol{0}^\top)^\top}.
\end{eqnarray*}
Note that $\boldsymbol{\Omega}_n(\boldsymbol{\beta})$ and $\boldsymbol{v}_n(\boldsymbol{\beta})$ can be written as
    \begin{eqnarray*}
    \boldsymbol{\Omega}_n(\boldsymbol{\beta}) = \begin{pmatrix} \boldsymbol{\Omega}^{(1)}_n(\boldsymbol{\beta}) & \boldsymbol{\Omega}^{(12)}_n(\boldsymbol{\beta}) \\
    \boldsymbol{\Omega}^{(21)}_n(\boldsymbol{\beta}) & \boldsymbol{\Omega}^{(2)}_n(\boldsymbol{\beta}) \end{pmatrix},   
    \end{eqnarray*} 
    and
    \begin{eqnarray*} 
    \boldsymbol{v}_n(\boldsymbol{\beta}) =
    \begin{pmatrix}
    \boldsymbol{v}^{(1)}_n(\boldsymbol{\beta}) \\ \boldsymbol{v}^{(2)}_n(\boldsymbol{\beta})
    \end{pmatrix},    
    \end{eqnarray*}
    respectively, where $\boldsymbol{\Omega}_n^{(1)}(\cdot)$ is a $q_n\times q_n$ leading submatrix of $\boldsymbol{\Omega}_n(\cdot)$, and $\boldsymbol{v}_n^{(1)}(\cdot)$ contains the first $q_n$ elements of $\boldsymbol{v}_n(\cdot)$.
To establish the asymptotic properties, it is necessary to satisfy the following conditions.
\begin{itemize}
\item [C1.] (i) The set $\mathcal{B}$ is a compact subset of $\mathbb{R}^{p_n}$, and $\boldsymbol{\beta}_0$ is an inferior point of $\mathcal{B}$.

(ii) There exists $\boldsymbol{Z}_0 > 0$, such that $P(\left\lVert\boldsymbol{Z} \right\rVert \leq \boldsymbol{Z}_0) = 1$, i.e., $\boldsymbol{Z}$ is bounded. The matrix $\mathbb{E}(\boldsymbol{Z} \boldsymbol{Z}^\top)$ is non-singular. 
\item [C2.] The union of the supports of $L$ and $R$ is contained in an interval $[u,v]$ with $0 < u < v < \infty$ and there exists a positive number $\zeta$ such that $P(R - L \geq \zeta) = 1$.
\item [C3.] The functions $\Lambda_{k}(\cdot)$, $k = 1,\ldots,K$, are continuously differentiable up to order $r$ in $[u,v]$, and satisfy $1/a < \Lambda_{k}(u) < \Lambda_{k}(v) < a$ for some positive constant $a$, for every $k \in \{1,\ldots, K \}$.
\item [C4.] For $\boldsymbol{\Omega}_n(\boldsymbol{\beta})$, there exists a compact neighbourhood $\mathcal{B}_0$ of the true value of $\boldsymbol{\beta}_0$ and a $p_n \times p_n$ positive-definite matrix, $\textbf{I}(\boldsymbol{\beta})$ such that
$$
\sup_{\beta \in \mathcal{B}_0} \left\lVert n^{-1} \boldsymbol{\Omega}_n(\boldsymbol{\beta}) - \textbf{I}(\boldsymbol{\beta}) \right\rVert \xrightarrow{\text{a.s.}} 0.
$$
\item [C5.] Define $\lambda_{\min}(\boldsymbol{\beta}) = \lambda_{\min}(n^{-1} \boldsymbol{\Omega}_n(\boldsymbol{\beta}))$ and $\lambda_{\max}(\boldsymbol{\beta}) = \lambda_{\max}(n^{-1} \boldsymbol{\Omega}_n(\boldsymbol{\beta}))$, where $\lambda_{\min}(\cdot)$ and $\lambda_{\max}(\cdot)$ denote the smallest and largest eigenvalues of the matrix. There exists a constant $c_0 > 0$, for $\mathcal{B}_0$ given in (C4), such that
$$
c_0^{-1} < \inf_{\boldsymbol{\beta} \in \mathcal{B}_0} \{ \lambda_{\min}(\boldsymbol{\beta}) \} \leq \sup_{\boldsymbol{\beta} \in \mathcal{B}_0} \{ \lambda_{\max}(\boldsymbol{\beta})  \} < c_0
$$
for a sufficiently large $n$. 
\item [C6.] As $n \rightarrow \infty, p_n q_n/\sqrt{n} \rightarrow 0, \tau_n \sqrt{p_n/n} \rightarrow 0$ and $\tau^2_n / (p_n \sqrt{n}) \rightarrow \infty$. 
\item [C7.] There exist positive constants $a_0$ and $a_1$ such that $a_0 \leq |\beta_{0j}| \leq a_1, 1 \leq j \leq q_n$. 
\item [C8.] The initial estimator $\widehat{\boldsymbol{\beta}}^{(0)}$ satisfies $\left\lVert \widehat{\boldsymbol{\beta}}^{(0)} - \boldsymbol{\beta}_0 \right\rVert = O_p(\sqrt{p_n /n })$. 
\item [C9.] For every $n$, the observations $\{v_{ni}, i=1,\ldots,n\}$ are independent and identically distributed with the probability density $f_n(v_{ni}; \boldsymbol{\beta}, \boldsymbol{\Lambda})$, which has common support and the model is identifiable. The parameter space is $\boldsymbol{\Theta} = \{ \boldsymbol{\nu}: \boldsymbol{\nu} = (\boldsymbol{\beta}, \boldsymbol{\Lambda}) \in \mathcal{B} \otimes \boldsymbol{\varphi}\}$, $\boldsymbol{\beta}_0$ is an interior point of $\mathcal{B}$, then for almost all $v_{ni}$, the density $f_n$ admits all third derivatives $\partial f_n(v_{ni}; \boldsymbol{\beta}, \boldsymbol{\Lambda})/\partial \beta_j \partial \beta_k \partial \beta_h $ for all $\boldsymbol{\beta} \in \mathcal{B}$. Furthermore, there are functions $M_{njkh}$ such that
\begin{eqnarray*}
\left\vert \frac{\partial \log f_n(v_{ni}; \boldsymbol{\beta}, \boldsymbol{\Lambda})}{\partial \beta_j \partial \beta_k \partial \beta_h} \right\vert \leq M_{njkh}(v_{ni})
\end{eqnarray*}
for all $\boldsymbol{\beta} \in \mathcal{B}$ and $\boldsymbol{\Lambda} \in \boldsymbol{\varphi}$, and
\begin{eqnarray*}
E_{\boldsymbol{\beta}, \boldsymbol{\Lambda}}\{ M^2_{njkh}(v_{ni}) \} < M_d < \infty.
\end{eqnarray*}
\end{itemize} 

Condition (C8) is crucial for establishing the oracle property of BAR.  The theory of semiparametric maximum likelihood estimation  may ensure that such an  initial estimator exists with the the desired convergence rate, for example, see \cite{lian2014scad} for a similar result in a different setting. Other conditions are required for establishing consistency of the sieve maximum likelihood estimator of the nuisance parameters and the asymptotic properties of the BAR estimator. The probability density function $f_n(\cdot)$  in (C9) is the $i$th term in the observed data likelihood function. The same condition is used in \cite{fan2004nonconcave} for the complete data models including the generalized linear model. 

The theorem below establishes the oracle property of the estimator $\widehat{\boldsymbol{\beta}}^*$.
\begin{theorem}\label{theorem1} 
    \textbf{(Oracle Property)} Assuming that the regularity conditions (C1)-(C9) are satisfied, then, with probability tending to 1, the BAR estimator 
    $\widehat{\boldsymbol{\beta}}^*=(\widehat{\boldsymbol{\beta}}_{s1}^{* \top},
    \widehat{\boldsymbol{\beta}}_{s2}^{* \top})^\top$ has the following properties:
\begin{itemize}
    \item [(i).] $\widehat{\boldsymbol{\beta}}^*_{s2}=0$.
    \item [(ii).] $\widehat{\boldsymbol{\beta}}^*_{s1}$ exists and is the unique fixed point of the equation
    \begin{eqnarray*}
        \boldsymbol{\beta}_{s1}=(\boldsymbol{\Omega}_n^{(1)}(\boldsymbol{\beta}_{s1})+\tau_n \boldsymbol{D}(\boldsymbol{\beta}_{s1}))^{-1}\boldsymbol{v}_n^{(1)}(\boldsymbol{\beta}_{s1}),
    \end{eqnarray*}
    where $\boldsymbol{D}(\boldsymbol{\beta}_{s1})=diag\{\beta_{1}^{-2},\ldots,\beta_{q_n}^{-2}\}$ and $\beta_1,\ldots,\beta_{q_n}$ represent the $q_n$ non-zero elements from risk 1 ($k=1$) to risk $K$ ($k=K$).
    \item [(iii).] For any $\textbf{b}_n$ being a $q_n$-vector, assume that $\left\lVert\boldsymbol{b}_n\right\rVert =1$. Then 
    \begin{eqnarray*}
        \sqrt{n}\textbf{b}^\top_n \boldsymbol{\Sigma}^{-\frac{1}{2}}(\widehat{\boldsymbol{\beta}}^*_{s1} - \boldsymbol{\beta}_{0s1}) \xrightarrow{d} N(0,1),
    \end{eqnarray*} 
    where
    \begin{eqnarray*}
        \boldsymbol{\Sigma} = {(I^{(1)}\left(\boldsymbol{\beta}_0)\right)}^{-1}
        ={(I^{(1)}\left(\boldsymbol{\beta}_{0s1})\right)}^{-1},
        \end{eqnarray*}
        i.e., $\widehat{\boldsymbol{\beta}}^*_{s1}$ is asymptotically normal with asymptotic variance $\boldsymbol{\Sigma}/n$. $I^{(1)}(\boldsymbol{\beta}_0)$ is the leading $q_n\times q_n$ submatrix of $I(\boldsymbol{\beta}_0)$, which indicates that the semiparametric information bound for the true sparse model  is achieved by the BAR penalty and the BAR estimator possesses the oracle property. 
\end{itemize}
\end{theorem}
\section{Simulation Study}
\label{sim-project1}
To assess the theoretical results of the proposed method, we conduct simulation studies and report the oracle results comparing them with the LASSO, ALASSO, and BAR. 
We consider the logarithmic transformation functions $$G_1(x)=\frac{1}{r_1}\log (1+r_1 x)$$ and
$$G_2(x)=\frac{1}{r_2}\log (1+r_2 x),$$ with $r_1$, and $r_2$ representing the transformation parameters for each risk assuming that we have two risks in the competing risks data. The model's transformation parameters $(r_1,r_2)$ are set as $(0,0)$, $(0.5,0.5)$, and $(1,1)$. We set the cumulative hazard functions corresponding to two risks as $\Lambda_1(t)=\Lambda_2(t)=0.2(1-e^{-t})$, and employ the inverse probability method to generate a time to event variable,
\begin{eqnarray*}
T_k=-\log\Bigg(1-\frac{G_k^{-1}(-\log (1-p_kV))}{0.2e^{\boldsymbol{\beta}_k^\top \boldsymbol{Z}}} \Bigg),
\end{eqnarray*}
where $p_k=1-\exp[-G_k(0.2 e^{\boldsymbol{\beta_k}^\top\boldsymbol{Z}})],~k=1,2$ is the probability by which we generate status 1 and 2 for competing events, and $V\sim \text{Uniform}(0,1)$. The associated covariates are marginally standard normal with mean zero and a variance-covariance matrix where the value of each element at the $(i,j)$th position is $\rho^{|i-j|}$.
Two examination times are generated for interval censoring. The first examination time is generated from $U_1 \sim \text{Uniform}(0.1, 1.5)$ and the second examination time from 
$U_2=U_1+dU$, where $dU \sim \text{Uniform}(0.1, 1.6)$,  $U_1$ and $dU$ are independent.

Four different situations are considered to test the performance of our proposed method. First, we generate $n=200$ subjects using the true parameter values of $\boldsymbol{\beta}_0=(\boldsymbol{\beta}_{01}^\top,\boldsymbol{\beta}_{02}^\top)^\top$, $\boldsymbol{\beta}_{01}=(0.8,0.6,0.8,\boldsymbol{0}_{d_n-3}^\top)^\top$, $\boldsymbol{\beta}_{02}=-\boldsymbol{\beta}_{01}$ where $d_n=14$ (i.e., $p_n=28$) and $\rho=0.2$ to test the performance of our proposed method under a weak correlation among the covariates. We repeat the same experiment with $\rho=0.8$ to evaluate how well our proposed approach performs when there is a strong correlation among the covariates. We repeat the simulation $100$ times, the results of these two scenarios are reported in \autoref{tab:sim1proj1}.  In addition, we consider another setting where we increase the sample size to $400$ and increase $p_n$ to $56$ (i.e., $d_n=28$ for each risk) with both weak and strong correlation among the covariates. The results of this setting are presented in \autoref{tab:sim2proj1}.  

We evaluate the performance of the model using several criteria, including the average number of nonzero estimates of parameters with true nonzero values, referred to as true positives (TP), the average number of nonzero estimates of parameters with true zero values, known as false positives (FP), and the average number of misclassified variables (MCV). Additionally, we report the Median of mean squared errors (MMSE). $k$th MSE corresponds to the $k$th risk and is calculated as $(\widehat{\boldsymbol{\beta}}_k-\boldsymbol{\beta}_{0k})^\top\boldsymbol{\Sigma}_k(\widehat{\boldsymbol{\beta}}_k-\boldsymbol{\beta}_{0k})$, where $\boldsymbol{\Sigma}_k$ is the population covariance matrix of the covariates. Eventually, the final MSE is the sum of the MSE's for $k=1,\ldots,K$. $\widehat{\boldsymbol{\beta}}_k$ represents the penalized estimate of the regression parameters for the $k$th risk. In addition to MMSE, its standard deviation (SD) is also recorded. As shown in \autoref{tab:sim1proj1}, BAR has achieved a smaller MCV in most situations with different transformation parameters and correlation values among the covariates. It can also be observed that MMSE is smaller in BAR results although ALASSO can be considered as a competing method in terms of its general performance judged by MCV and MMSE.
\begin{table}[H]
\caption{Results of simultaneous estimation and variable selection with three sets of transformation parameters, $(0,0)$, $(0.5,0.5)$, and $(1,1)$. In this table, we assume that $n=200$, and $d_n=14$ corresponding to each risk (i.e., $p_n=28$). Data are generated from two scenarios. $\rho=0.2$ for weak correlation among covariates and $\rho=0.8$ for strong correlation among them.}
\label{tab:sim1proj1}
\resizebox{\textwidth}{!}{%
\begin{tabular}{ccccccccccccccccccc}
\hline
Penalty &  & $(r_1,r_2)$                &  & TP   &  & FP   &  & MCV  &  & MMSE (SD)      &  & TP   &  & FP   &  & MCV  &  & MMSE (SD)     \\ \hline
        &  &                            &  & \multicolumn{15}{c}{$n=200, p_n=2\times 14, q_n=6$}                                                  \\ \cline{5-19} 
        &  &                            &  & \multicolumn{7}{c}{$\rho=0.2$}               &  & \multicolumn{7}{c}{$\rho=0.8$}              \\ \cline{5-11} \cline{13-19} 
\multicolumn{19}{c}{}                                                                                                                      \\
LASSO   &  & \multirow{4}{*}{(0,0)}     &  & 5.84 &  & 1.72 &  & 1.88 &  & 0.949 (0.261)  &  &  4.21  &  &   2.10   &  &  3.89    &  & 2.541 (0.671)               \\
ALASSO  &  &                            &  & 5.54 &  & 0.90 &  & 1.36 &  & 0.858 (0.295)  &  & 3.66 &  & 0.62 &  & 2.96 &  & 1.889 (0.490) \\
BAR     &  &                            &  & 4.80 &  & 0.16 &  & 1.36 &  & 0.843 (0.306)  &  & 3.54 &  & 0.42 &  & 2.88 &  & 1.572 (0.538) \\
Oracle  &  &                            &  & 6.00 &  & 0.00 &  & 0.00 &  & 0.838 (0.232)  &  & 6.00 &  & 0.00 &  & 0.00 &  & 1.237 (0.345) \\
        &  &                            &  &      &  &      &  &      &  &                &  &      &  &      &  &      &  &               \\
LASSO   &  & \multirow{4}{*}{(0.5,0.5)} &  & 5.80 &  & 1.70 &  & 1.90 &  & 1.0556 (0.302) &  &   4.00   &  &  2.12    &  &  4.12    &  &    2.617 (0.585)           \\
ALASSO  &  &                            &  & 5.48 &  & 1.14 &  & 1.66 &  & 0.896 (0.314)  &  & 3.40 &  & 0.40 &  & 3.00 &  & 1.793 (0.431) \\
BAR     &  &                            &  & 4.72 &  & 0.18 &  & 1.46 &  & 0.859 (0.385)  &  & 3.16 &  & 0.12 &  & 2.96 &  & 1.792 (0.619) \\
Oracle  &  &                            &  & 6.00 &  & 0.00 &  & 0.00 &  & 0.849 (0.356)  &  & 6.00 &  & 0.00 &  & 0.00 &  & 1.390 (0.421) \\
        &  &                            &  &      &  &      &  &      &  &                &  &      &  &      &  &      &  &               \\
LASSO   &  & \multirow{4}{*}{(1,1)}     &  & 5.46 &  & 1.56 &  & 2.10 &  & 1.301 (0.334)  &  &    3.75  &  &  2.14    &  &   4.39   &  &     2.924 (0.479)          \\
ALASSO  &  &                            &  & 4.76 &  & 0.54 &  & 1.78 &  & 1.534 (0.403)  &  & 3.16 &  & 0.72 &  & 3.56 &  & 1.945 (0.432) \\
BAR     &  &                            &  & 4.70 &  & 0.68 &  & 1.98 &  & 1.081 (0.532)  &  & 3.16 &  & 0.62 &  & 3.46 &  & 1.748 (0.558) \\
Oracle  &  &                            &  & 6.00 &  & 0.00 &  & 0.00 &  & 0.914 (0.421)  &  & 6.00 &  & 0.00 &  & 0.00 &  & 1.470 (0.437) \\ \hline
\end{tabular}%
}
\end{table}
\noindent Consistent with expectations, LASSO performs well in terms of True Positive (TP) values, making it a practical method for detecting non-zero variables. However, when considering MCV and MMSE, BAR and ALASSO outperform LASSO. The BAR method is highly conservative, resulting in low False Positive (FP) values. This indicates that it is a reliable method for ensuring that the important variables in a model are correctly identified during the variable selection procedure.
\begin{table}[H]
\caption{Results of simultaneous estimation and variable selection with three sets of transformation parameters, $(0,0)$, $(0.5,0.5)$, and $(1,1)$. In this table, we assume that $n=400$, and $d_n=28$ corresponding to each risk (i.e., $p_n=56$). Data are generated from two scenarios. $\rho=0.2$ for weak correlation among covariates and $\rho=0.8$ for strong correlation among them.}
\label{tab:sim2proj1}
\resizebox{\textwidth}{!}{%
\begin{tabular}{ccccccccccccccccccc}
\hline
Penalty &  & $(r_1,r_2)$                &  & TP   &  & FP   &  & MCV  &  & MMSE (SD)      &  & TP   &  & FP   &  & MCV  &  & MMSE (SD)     \\ \hline
        &  &                            &  & \multicolumn{15}{c}{$n=400, p_n=2\times 28, q_n=6$}                                                  \\ \cline{5-19} 
        &  &                            &  & \multicolumn{7}{c}{$\rho=0.2$}               &  & \multicolumn{7}{c}{$\rho=0.8$}              \\ \cline{5-11} \cline{13-19} 
\multicolumn{19}{c}{}                                                                                                                      \\
LASSO   &  & \multirow{4}{*}{(0,0)}     &  & 5.80 &  & 1.53 &  & 1.73 &  & 0.894 (0.215)  &  &   4.56   &  &   1.96   &  &   3.40   &  &      2.124 (0.482)         \\
ALASSO  &  &                            &  & 5.90 &  & 0.46 &  & 0.56 &  & 0.638 (0.142)  &  & 4.44 &  & 1.10 &  & 2.66 &  & 1.340 (0.387) \\
BAR     &  &                            &  & 5.62 &  & 0.02 &  & 0.40 &  & 0.562 (0.165)  &  & 3.62 &  & 0.04 &  & 2.42 &  & 1.310 (0.533) \\
Oracle  &  &                            &  & 6.00 &  & 0.00 &  & 0.00 &  & 0.515 (0.101)  &  & 6.00 &  & 0.00 &  & 0.00 &  & 0.993 (0.326) \\
        &  &                            &  &      &  &      &  &      &  &                &  &      &  &      &  &      &  &               \\
LASSO   &  & \multirow{4}{*}{(0.5,0.5)} &  & 5.66 &  & 1.47 &  & 1.81 &  & 1.024 (0.302) &  &  4.52    &  &   2.23   &  &   3.71   &  &       2.321 (0.510)        \\
ALASSO  &  &                            &  & 5.92 &  & 0.44 &  & 0.52 &  & 0.708 (0.193)  &  & 4.38 &  & 1.04 &  & 2.66 &  & 1.524 (0.372) \\
BAR     &  &                            &  & 5.58 &  & 0.02 &  & 0.44 &  & 0.613 (0.063)  &  & 3.44 &  & 0.06 &  & 2.62 &  & 1.512 (0.579) \\
Oracle  &  &                            &  & 6.00 &  & 0.00 &  & 0.00 &  & 0.570 (0.084)  &  & 6.00 &  & 0.00 &  & 0.00 &  & 0.997 (0.398) \\
        &  &                            &  &      &  &      &  &      &  &                &  &      &  &      &  &      &  &               \\
LASSO   &  & \multirow{4}{*}{(1,1)}     &  & 5.48 &  & 1.39 &  & 1.91 &  & 1.287 (0.354)  &  &   4.41   &  &    2.71  &  &   4.30   &  &      2.547 (0.507)         \\
ALASSO  &  &                            &  & 5.82 &  & 0.78 &  & 0.96 &  & 0.865 (0.204)  &  & 4.18 &  & 1.38 &  & 3.20 &  & 1.620 (0.368) \\
BAR     &  &                            &  & 5.36 &  & 0.06 &  & 0.70 &  & 0.756 (0.251)  &  & 3.42 &  & 0.04 &  & 2.62 &  & 1.555 (0.556) \\
Oracle  &  &                            &  & 6.00 &  & 0.00 &  & 0.00 &  & 0.612 (0.114)  &  & 6.00 &  & 0.00 &  & 0.00 &  & 1.037 (0.401) \\ \hline
\end{tabular}%
}
\end{table}
\noindent A similar pattern among three penalty functions, LASSO, ALASSO, and BAR can be observed in \autoref{tab:sim2proj1}. However, it can be seen that the performance of all the methods improves with an increase in the sample size although the number of variables has also increased considerably. 
\section{Real Data Analysis}
\label{realdata-proj1} 
In this section, a sample of 1119 injecting drug users in a cohort study carried out by the Bangkok Metropolitan Administration (BMA) is used to illustrate the proposed variable selection method on competing risks data.
This study was started in 1995 aiming to investigate the feasibility of conducting phase 3 of the vaccine trial in the injecting drug user population in Bangkok, Thailand. This study originally had two goals: first, to assess the rate of complete follow-up cases in the study, and second, to find more effective HIV prevention measures by determining the important risk factors. 
All subjects in the study were HIV seronegative injecting drug users,
The subjects were followed from 1995 to 1998 at 15 BMA drug treatment clinics. Blood tests were conducted on each participant approximately every 4 months after recruitment to detect evidence of HIV-1 seroconversion (i.e., the detection of HIV-1 antibodies in the serum). 
The blood tests were examined to detect HIV antibodies and to determine if the seroconversions were of viral subtype B or subtype E. Among 117 subjects for whom  seroconversion was observed, there are 6 subjects with unknown viral subtypes or missing cause of failure, 24 subjects with viral subtype B, and 87 subjects with viral subtype E. 
\autoref{tab:chap2-dic} provides a dictionary of the covariates and other variables observed in this data set. In \autoref{tab:chap2-dic}, $Z_l$, $l=1,\ldots,d_n$, $d_n=8$ corresponds to each of the covariates (i.e., risk factors/variables). 

We treat this data set as interval-censored competing risks data and consider subtypes B and E as two competing risks while allowing for the missing cause of failure for the event of interest (HIV seroconversion).
To select the optimal model that $r_1$ and $r_2$ produce, we implement the EM algorithm and compute the log-likelihood value with different transformation parameters.
\begin{table}[H]
\caption{Description of the BMA data. Con(Cat), TI, and TV represent continuous(categorical), time-invariant, and time-varying covariates, respectively.}
\label{tab:chap2-dic}
\resizebox{\columnwidth}{!}{%
\begin{tabular}{lllll}
\hline
\begin{tabular}[c]{@{}l@{}}Variables in the \\ observed data\end{tabular} &
   &
  Type &
   &
  Description \\ \hline
$(J, U)$ &
   &
  - &
   &
  \begin{tabular}[c]{@{}l@{}}Examination times revealing the total number \\ of visit times ($J$) and date\\ of visits for each subject ($U$)\end{tabular} \\ \cline{1-1} \cline{3-3} \cline{5-5} 
$(\Delta, D, \xi)$ &
   &
  - &
   &
  \begin{tabular}[c]{@{}l@{}}Status that reveals if the subject is/is not infected\\ ($\Delta$), virus subtype of B or E \\ ($D$), and if the cause of failure is \\ missing ($\xi$)\end{tabular} \\ \cline{1-1} \cline{3-3} \cline{5-5} 
 &
   &
   &
   &
   \\
$Z_1$: Age &
   &
  Con-TI &
   &
  Age (in years) at registration \\
$Z_2$: First Age &
   &
  Con-TI &
   &
  Age at the first time using drugs (in years) \\
$Z_3$: Gender &
   &
  Cat-TI &
   &
  Gender (0: male, 1: female) \\
$Z_4$: Needle &
   &
  Cat-TV &
   &
  History of needle sharing (0: no, 1: yes) \\
$Z_5$: Jail &
   &
  Cat-TV &
   &
  \begin{tabular}[c]{@{}l@{}}Number of times imprisoned since last seen \\ (0: none, 1: one or more than one)\end{tabular} \\
$Z_6$: Income &
   &
  Cat-TI &
   &
  \begin{tabular}[c]{@{}l@{}}Monthly income (0: less than or equal to 5000 \\ baht, 1: more than 5000 baht)\end{tabular} \\
$Z_7$: Syringe &
   &
  Cat-TV &
   &
  \begin{tabular}[c]{@{}l@{}}History of injecting drug while being in prison \\ (0: no, 1: yes)\end{tabular} \\
$Z_8$: Inject Freq &
   &
  Cat-TV &
   &
  \begin{tabular}[c]{@{}l@{}}Frequency of injecting drugs (0: none, 1: at least \\ one time)\end{tabular} \\ \hline
\end{tabular}%
}
\end{table}
\begin{table}[H]
\caption{Unpenalized analysis on the BMA data with selected transformation parameters, $r_1=0.6$ and $r_2=1.8$, considering the competing events of interest, subtype B and subtype E.}
\label{selectedmodel}
\resizebox{\textwidth}{!}{%
\begin{tabular}{cclccclccc}
\hline
 &  &  & \multicolumn{3}{c}{Subtype B} &  & \multicolumn{3}{c}{Subtype E} \\ \cline{4-6} \cline{8-10} 
 & \begin{tabular}[c]{@{}c@{}}Variables\end{tabular} &  & Estimate & Std. Error & $p$-value &  & Estimate & Std. Error & $p$-value \\ \hline
 &  &  &  &  &  &  & \multicolumn{1}{l}{} & \multicolumn{1}{l}{} & \multicolumn{1}{l}{}\\
\multirow{8}{*}{$r_1=0.6, r_2=1.8$} & Age &  & 0.006 & 0.344  & 0.999 &  & 0.084 & 0.183& 0.617 \\
 & First Age &  & -0.016 & 0.260  & 0.959 &  & 0.068 & 0.128 & 0.553 \\
 & Gender &  & 1.131 & 0.519 & 0.025 &  & -1.176 & 0.706 & 0.091 \\
 & Needle &  & 0.143 & 0.566  & 0.798 &  & 0.372 & 0.307 & 0.240 \\
 & Jail &  & 0.295 & 0.505 & 0.555 &  & 0.208 & 0.290 & 0.459 \\
 & Income &  & -0.262 & 0.661 & 0.710 &  & -0.120 & 0.359 & 0.629 \\
 & Syringe &  & 0.044 & 0.436 & 0.999 &  & 0.316 & 0.247 & 0.193 \\
 & Inject Freq &  & 0.306 & 1.164 & 0.441 &  & 0.123 & 0.555 & 0.777 \\
 &  &  &  &  &  &  & \multicolumn{1}{l}{} & \multicolumn{1}{l}{} & \multicolumn{1}{l}{} \\ \hline
\multicolumn{1}{l}{} &  &  &  &  &  &  & \multicolumn{1}{l}{} & \multicolumn{1}{l}{} & \multicolumn{1}{l}{}
\end{tabular}%
}
\end{table}
\noindent We use a grid of $r_1, r_2$ over the range of $(0,3]$ with increments of $0.2$ and select the parameters that maximize the log-likelihood function, we obtain the $r_1=0.6$, $r_2=1.8$. The unpenalized results of the selected model is represented in \autoref{selectedmodel}. Then, based on the selected model, we perform variable selection using three penalty functions, LASSO, ALASSO, and BAR. We select the tuning parameter using the GCV criterion presented in \Cref{subsubproj12}. 

Based on the results in \autoref{BMAvarselresult}, it can be seen that all three penalty functions select gender to be an important variable in the model for both competing events, subtype B and subtype E. Also, LASSO, ALASSO, and BAR agree on selecting the variable syringe for  the second competing event. While LASSO and ALASSO select gender, needle, and jail for the second competing risk, subtype E, BAR shrinks these variables to zero and as it is expected, BAR produces the most sparse model among these three penalty functions. 
\begin{table}[H]
\caption{Variable selection result on the BMA data employing the selected set of transformation parameters for subtypes B and E.}
\label{BMAvarselresult}
\resizebox{\textwidth}{!}{%
\begin{tabular}{cclccclccc}
\hline
 &  &  & \multicolumn{3}{c}{Subtype B} &  & \multicolumn{3}{c}{Subtype E} \\ \cline{4-6} \cline{8-10} 
 & \begin{tabular}[c]{@{}c@{}}Variables\end{tabular} &  & LASSO & ALASSO & BAR &  & LASSO & ALASSO & BAR \\ \hline
 &  &  &  &  &  &  & \multicolumn{1}{l}{} & \multicolumn{1}{l}{} & \multicolumn{1}{l}{}\\
\multirow{8}{*}{$r_1=0.6, r_2=1.8$} & Age &  & -0.053 & 0  & 0 &  & 0 & 0& 0 \\
 & First Age &  & 0 & 0  & 0 &  & 0.078 & 0 & 0 \\
 & Gender &  & 0.310 & 0.759 & 0.957 &  & -0.238 & -0.567 & -0.770 \\
 & Needle &  & 0 & 0  & 0 &  & 0.151 & 0.307 & 0 \\
 & Jail &  & 0 & 0 & 0&  & 0.061 & 0.083 & 0 \\
 & Income &  & 0 & 0 & 0 &  & 0 & 0 & 0 \\
 & Syringe &  & 0 & 0 & 0 &  & 0.172 & 0.139 & 0.248 \\
 & Inject Freq &  & 0 & 0 & 0 &  & 0 & 0 & 0 \\
 &  &  &  &  &  &  & \multicolumn{1}{l}{} & \multicolumn{1}{l}{} & \multicolumn{1}{l}{} \\ \hline
\multicolumn{1}{l}{} &  &  &  &  &  &  & \multicolumn{1}{l}{} & \multicolumn{1}{l}{} & \multicolumn{1}{l}{}
\end{tabular}%
}
\end{table} 
\section{Discussion and Concluding Remarks}\label{discussionproj1}
We have considered interval-censored competing risks data and proposed a penalized variable selection technique to estimate and select variables, simultaneously under a semiparametric transformation model. The semiparametric transformation model is a general term referring to a class of models that encompasses some special types including the proportional hazards and proportional odds models. Employing this model makes our model more flexible to be able to take different forms. We employed the broken adaptive ridge regression method for variable selection and proposed an iteratively reweighted least square algorithm to approximate the likelihood function as a least square problem along with an optimization procedure to solve the variable selection problem using LASSO, Adaptive LASSO, and BAR. Unlike the works published in the literature, we took all the risks in a competing risks data set into consideration. Despite the Fine-Gray model, our approach allows for the assessment of the variables corresponding to all the risks or submodels and therefore, there is no need for determining the censoring distribution like in the Fine-Gray model. We established the oracle properties of the BAR estimator for interval-censored competing risks data. We improved the existing techniques in the proofs and obtained a semiparametric information bound for the sparse estimator of the true model parameters. Our numerical results demonstrate that the proposed methods outperform existing competitors.



\section*{Author contributions}
All the authors contributed equally.

\section*{Acknowledgments}
Omitted for peer review.

\section*{Financial disclosure}
Omitted for peer review.

\section*{Conflict of interest}

The authors declare that they have no conflict of interest.

\bibliographystyle{apalike}    

\pagebreak
\section*{Appendix: Proofs of the Asymptotic Properties in Theorem~\ref{theorem1}}
\label{chap3sec9}
\medskip
\renewcommand{\theequation}{A.\arabic{equation}}
\setcounter{equation}{0}
Suppose that the log-likelihood of our model is $\ell_n(\boldsymbol{\beta},\boldsymbol{\Lambda}) = \log \mathcal{L}_n(\boldsymbol{\beta}, \boldsymbol{\Lambda})$ defined in \eqref{loglikforproof}. Let $(\widetilde{\boldsymbol{\beta}}, \widetilde{\boldsymbol{\Lambda}})$ be the unpenalized estimates of $(\boldsymbol{\beta}, \boldsymbol{\Lambda})$ obtained by using the semiparametric or parametric methods.

The total number of variables denoted by $p_n$ is considered as diverging, i.e., $p_n \longrightarrow \infty$ and $q_n \longrightarrow \infty$, when $n \longrightarrow \infty$, but $p_n$ and $q_n$ satisfy condition (C6). 

We assume $\boldsymbol{\beta} = (\boldsymbol{\beta}^\top_{s1}, \boldsymbol{\beta}^\top_{s2})^\top$, and the corresponding true value 
$\boldsymbol{\beta}_0= (\boldsymbol{\beta}^\top_{0s1}, \boldsymbol{\beta}^\top_{0s2})^\top$, 
where \begin{eqnarray*}
\boldsymbol{\beta}_{s1}=(\beta_{s1,1},\ldots,\beta_{s1,q_n})^\top 
\end{eqnarray*} 
consists of all nonzero coefficients (which is a $q_n$-vector of parameters) across all the three transitions and \begin{eqnarray*}
\boldsymbol{\beta}_{s2}=(\beta_{s2,q_n+1},\ldots,\beta_{s2,p_n})^\top   
\end{eqnarray*} 
is a $(p_n - q_n)$-vector of parameters, consisting of all zero coefficients. Vector of $\boldsymbol{\beta}=(\beta_{1},\ldots,\beta_{p_n})^\top$ represents all the parameters in the model (containing both non-zero and zero ones). For the simultaneous estimation and variable selection, we consider the penalized function
\begin{eqnarray*}
    \ell_{pp}(\boldsymbol{\beta})
=-\ell_p(\boldsymbol{\beta})+\sum_{k=1}^{K}\sum_{j=1}^{d_k} p_{\lambda_n}(\boldsymbol{\beta}_{j,k}),
\end{eqnarray*}
where $\ell_p(\boldsymbol{\beta}) =\max_{\boldsymbol{\Lambda}}\ell_n(\boldsymbol{\beta},\boldsymbol{\Lambda})$.

Utilizing BAR penalty function, we have 
\begin{eqnarray}
\label{obj-proofs}
\ell_{pp}(\boldsymbol{\beta} | \check{\boldsymbol{\beta}}) =-\ell_p(\boldsymbol{\beta})
+\lambda_n\sum_{k=1}^{3}\sum_{j=1}^{d_k} \frac{ \beta_{j,k}^2}{\check{\beta}_{j,k}^2}= - \ell_p(\boldsymbol{\beta}) + \lambda_n \sum^{p_n}_{j=1} \frac{\beta^2_j}{\check{\beta}^2_j}.
\end{eqnarray}

To establish the asymptotic properties, first, we show that minimizing \eqref{obj-proofs} is asymptotically equivalent to minimizing the following penalized least-squared function
\begin{eqnarray*}
\frac{1}{2}\left\lVert \textbf{W}(\check{\boldsymbol{\beta}}) - \textbf{X} (\check{\boldsymbol{\beta}})\boldsymbol{\beta} \right\rVert^2 + \lambda_n \sum^{p_n}_{j=1} \frac{\beta^2_j}{\check{\beta}^2_j}, 
\end{eqnarray*}
using Cholesky decomposition. 

Since $(\widetilde{\boldsymbol{\beta}},\widetilde{\boldsymbol{\Lambda}}) = \max_{(\boldsymbol{\beta},\boldsymbol{\Lambda})} \ell_n(\boldsymbol{\beta},\boldsymbol{\Lambda})$, 
\begin{eqnarray*}
\widetilde{\boldsymbol{\beta}} = \max_{\boldsymbol{\beta}} \ell_n(\boldsymbol{\beta}, \widetilde{\boldsymbol{\Lambda}}) = \max_{\boldsymbol{\beta}} \ell_n(\boldsymbol{\beta}\vert \widetilde{\boldsymbol{\Lambda}}),   
\end{eqnarray*}
where   
$\ell_n(\boldsymbol{\beta} | \widetilde{\boldsymbol{\Lambda}}) = \log\{ \mathcal{L}_n (\boldsymbol{\beta},\widetilde{\boldsymbol{\Lambda}})\}$ and $\ell_n(\boldsymbol{\beta}|\boldsymbol{\Lambda}) = \ell_n(\boldsymbol{\beta},\boldsymbol{\Lambda})$.

Define $\dot{\ell}_n(\boldsymbol{\beta}|\boldsymbol{\Lambda}) = \partial \ell_n(\boldsymbol{\beta}|\boldsymbol{\Lambda}) / \partial \boldsymbol{\beta}$, and $\Ddot{\ell}_n(\boldsymbol{\beta}|\boldsymbol{\Lambda}) = \partial^2 \ell_n(\boldsymbol{\beta}|\boldsymbol{\Lambda})/\partial \boldsymbol{\beta} \partial \boldsymbol{\beta}^\top$. Then, $(\widetilde{\boldsymbol{\beta}}, \widetilde{\boldsymbol{\Lambda}})$ satisfies $\dot{\ell}_n(\widetilde{\boldsymbol{\beta}}|\widetilde{\boldsymbol{\Lambda}}) = 0$. 

By the first-order Taylor expansion, we have
\begin{eqnarray*}
0 = \dot{\ell}_n(\widetilde{\boldsymbol{\beta}}|\widetilde{\boldsymbol{\Lambda}}) \approx \dot{\ell}_n(\boldsymbol{\beta}|\widetilde{\boldsymbol{\Lambda}}) + \Ddot{\ell}_n(\boldsymbol{\beta} | \widetilde{\boldsymbol{\Lambda}})
(\widetilde{\boldsymbol{\beta}} - \boldsymbol{\beta}),  
\end{eqnarray*}
which yields
\begin{eqnarray*}
\widetilde{\boldsymbol{\beta}} - \boldsymbol{\beta} \approx [- \Ddot{\ell}_n(\boldsymbol{\beta} | \widetilde{\boldsymbol{\Lambda}})]^{-1}\dot{\ell}_n(\boldsymbol{\beta}|\widetilde{\boldsymbol{\Lambda}}).   
\end{eqnarray*}
On the other hand, by the second-order Taylor expansion,
\begin{eqnarray*}
\ell_n(\widetilde{\boldsymbol{\beta}}|\widetilde{\boldsymbol{\Lambda}}) \approx \ell_n(\boldsymbol{\beta}|\widetilde{\boldsymbol{\Lambda}}) + [\Dot{\ell}_n(\boldsymbol{\beta}|\widetilde{\boldsymbol{\Lambda}})]^\top (\widetilde{\boldsymbol{\beta}} - \boldsymbol{\beta}) + (\widetilde{\boldsymbol{\beta}} - \boldsymbol{\beta})^\top \frac{\Ddot{\ell}_n(\boldsymbol{\beta} | \widetilde{\boldsymbol{\Lambda}})}{2} (\widetilde{\boldsymbol{\beta}} - \boldsymbol{\beta}).
\end{eqnarray*}
Thus we have
\begin{equation*}
    \begin{split}
        \ell_p(\boldsymbol{\beta}) = \ell_n(\boldsymbol{\beta} | \widetilde{\boldsymbol{\Lambda}})
        & \approx \ell_n(\widetilde{\boldsymbol{\beta}} | \widetilde{\boldsymbol{\Lambda}}) - [\Dot{\ell}_n(\boldsymbol{\beta} | \widetilde{\boldsymbol{\Lambda}})]^\top [- \Ddot{\ell}_n (\boldsymbol{\beta} | \widetilde{\boldsymbol{\Lambda}})]^{-1} \Dot{\ell}_n(\boldsymbol{\beta} | \widetilde{\boldsymbol{\Lambda}}) \\ 
        & + \frac{1}{2} [\Dot{\ell}_n(\boldsymbol{\beta} | \widetilde{\boldsymbol{\Lambda}})]^\top [- \Ddot{\ell}_n (\boldsymbol{\beta} | \widetilde{\boldsymbol{\Lambda}})]^{-1} [- \Ddot{\ell}_n (\boldsymbol{\beta} | \widetilde{\boldsymbol{\Lambda}})] [- \Ddot{\ell}_n (\boldsymbol{\beta} | \widetilde{\boldsymbol{\Lambda}})]^{-1} \Dot{\ell}_n(\boldsymbol{\beta} | \widetilde{\boldsymbol{\Lambda}}).
    \end{split}
\end{equation*}
Hence
\begin{eqnarray*}
    \ell_p(\boldsymbol{\beta})= -\frac{1}{2} [\Dot{\ell}_n(\boldsymbol{\beta} | \widetilde{\boldsymbol{\Lambda}})]^\top [- \Ddot{\ell}_n (\boldsymbol{\beta} | \widetilde{\boldsymbol{\Lambda}})]^{-1} \Dot{\ell}_n(\boldsymbol{\beta} | \widetilde{\boldsymbol{\Lambda}}) + C,
\end{eqnarray*}
where $C = \ell_n(\widetilde{\boldsymbol{\beta}} | \widetilde{\boldsymbol{\Lambda}})$ is a constant independent of $\boldsymbol{\beta}$. Therefore, maximizing $\ell_p(\boldsymbol{\beta})$ is equivalent to minimizing
\begin{eqnarray*}
\ell_p(\boldsymbol{\beta}) =\frac{1}{2} [\Dot{\ell}_n(\boldsymbol{\beta} | \widetilde{\boldsymbol{\Lambda}})]^\top [- \Ddot{\ell}_n (\boldsymbol{\beta} | \widetilde{\boldsymbol{\Lambda}})]^{-1} \Dot{\ell}_n(\boldsymbol{\beta} | \widetilde{\boldsymbol{\Lambda}}).
\end{eqnarray*}
Next, we show that $-\ell_p(\boldsymbol{\beta}) = \frac{1}{2}\left\lVert \textbf{W}(\boldsymbol{\beta}) - \textbf{X}(\boldsymbol{\beta}) \boldsymbol{\beta} \right\rVert^2$ by the Cholesky decomposition.

Let $\textbf{X}$ be defined by the Cholesky decomposition of $-\Ddot{\ell}_n(\boldsymbol{\beta} | \widetilde{\boldsymbol{\Lambda}})$ as $- \Ddot{\ell}_n (\boldsymbol{\beta} | \widetilde{\boldsymbol{\Lambda}}) = \textbf{X}^\top(\boldsymbol{\beta}) \textbf{X}(\boldsymbol{\beta})$ and define the pseudo-response vector $\textbf{W}(\boldsymbol{\beta}) = (\textbf{X}^\top(\boldsymbol{\beta}))^{-1} [\Dot{\ell}_n(\boldsymbol{\beta} | \widetilde{\boldsymbol{\Lambda}}) - \Ddot{\ell}_n(\boldsymbol{\beta} | \widetilde{\boldsymbol{\Lambda}}) \boldsymbol{\beta}] $. Then we have
$$
\frac{1}{2}\left\lVert \textbf{W}(\boldsymbol{\beta}) - \textbf{X}(\boldsymbol{\beta}) \boldsymbol{\beta} \right\rVert^2 = -\frac{1}{2} [\Dot{\ell}_n(\boldsymbol{\beta} | \widetilde{\boldsymbol{\Lambda}})]^\top [\Ddot{\ell}_n(\boldsymbol{\beta} | \widetilde{\boldsymbol{\Lambda}})]^{-1}[\Dot{\ell}_n(\boldsymbol{\beta} | \widetilde{\boldsymbol{\Lambda}})],
$$
unlike \cite{zhao2019simultaneous}, here we write $\textbf{W}(\boldsymbol{\beta})$ and $\textbf{X}(\boldsymbol{\beta})$ to emphasize the dependence of $\textbf{X}$ and $\textbf{W}$ on $\boldsymbol{\beta}$. Note that in terms of notation, we consider $\boldsymbol{X}(\boldsymbol{\beta})=\boldsymbol{X}(\boldsymbol{\beta}|\widetilde{\boldsymbol{\Lambda}})$, and $\boldsymbol{W}(\boldsymbol{\beta})=\boldsymbol{W}(\boldsymbol{\beta}|\widetilde{\boldsymbol{\Lambda}})$. 

This implies that minimizing \eqref{obj-proofs} is asymptotically equivalent to minimizing the following penalized least square function iteratively
\begin{eqnarray*} 
\frac{1}{2}\left\lVert \textbf{W}(\check{\boldsymbol{\beta}}) - \textbf{X}(\check{\boldsymbol{\beta}}) \boldsymbol{\beta} \right\rVert^2 + \lambda_n\sum^{p_n}_{j=1} \frac{\beta^2_j}{\check{\beta}^2_j}.
\end{eqnarray*} 

To prove Theorem~\ref{theorem1}, we first introduce the following notations. Define
\begin{equation}\label{A1}
\begin{pmatrix}
    \boldsymbol{\alpha}^*(\boldsymbol{\beta}) \\
    \boldsymbol{\gamma}^*(\boldsymbol{\beta})
\end{pmatrix} 
= J(\boldsymbol{\beta}) = \{\boldsymbol{\Omega}_n(\boldsymbol{\beta}) + \lambda_n \boldsymbol{D}(\boldsymbol{\beta})\}^{-1} \boldsymbol{v}_n(\boldsymbol{\beta}),
\end{equation}
where $\boldsymbol{\Omega}_n(\boldsymbol{\beta})=\boldsymbol{X}^\top(\boldsymbol{\beta})\boldsymbol{X}(\boldsymbol{\beta})$ and $\boldsymbol{v}(\boldsymbol{\beta})=\boldsymbol{X}^\top(\boldsymbol{\beta})\boldsymbol{W}(\boldsymbol{\beta})$. Now we partition the matrix $\{ n^{-1}\boldsymbol{\Omega}_n(\boldsymbol{\beta}) \}^{-1}$ into
$$
\{ n^{-1}\boldsymbol{\Omega}_n(\boldsymbol{\beta}) \}^{-1} = \begin{pmatrix}
\textbf{A}(\boldsymbol{\beta}) & \textbf{B}(\boldsymbol{\beta}) \\
\textbf{B}^\top(\boldsymbol{\beta}) & \textbf{G}(\boldsymbol{\beta}) \\
\end{pmatrix},
$$
where $\textbf{A}(\boldsymbol{\beta}), \textbf{B}(\boldsymbol{\beta})$ and $\textbf{G}(\boldsymbol{\beta})$ are $q_n \times q_n$, $q_n \times (p_n - q_n)$ and $(p_n - q_n) \times (p_n - q_n)$ matrices, respectively. Here we use $\boldsymbol{\Omega}_n(\boldsymbol{\beta})$ and $\boldsymbol{v}_n(\boldsymbol{\beta})$ instead of $\boldsymbol{\Omega}_n$ and $\boldsymbol{v}_n$ to emphasize the dependence of $\boldsymbol{\Omega}_n$ and $\boldsymbol{v}_n$ on $\boldsymbol{\beta}$. This is important in the subsequent proofs, particularly in Lemma~\ref{lemma2}.

Multiplying $\boldsymbol{\Omega}^{-1}_n(\boldsymbol{\beta})\left(\boldsymbol{\Omega}_n(\boldsymbol{\beta})+\lambda_n \boldsymbol{D}(\boldsymbol{\beta}) \right)$ and substituting $\boldsymbol{\beta}_{0} = (\boldsymbol{\beta}^\top_{0s1},\boldsymbol{\beta}^\top_{0s2})^\top$ on both sides of (\ref{A1}), we have
\begin{equation}\label{A2}
\begin{pmatrix}
    \boldsymbol{\alpha}^*(\boldsymbol{\beta}) - \boldsymbol{\beta}_{0s1} \\
    \boldsymbol{\gamma}^*(\boldsymbol{\beta})
\end{pmatrix}
+ \frac{\lambda_n}{n} 
\begin{pmatrix}
    \textbf{A}(\boldsymbol{\beta})\textbf{D}_1(\boldsymbol{\beta}_{s1})\boldsymbol{\alpha}^*(\boldsymbol{\beta}) + \textbf{B}(\boldsymbol{\beta})\textbf{D}_2(\boldsymbol{\beta}_{s2})\boldsymbol{\gamma}^*(\boldsymbol{\beta}) \\
    \textbf{B}^\top(\boldsymbol{\beta}) \textbf{D}_1(\boldsymbol{\beta}_{s1})\boldsymbol{\alpha}^*(\boldsymbol{\beta}) + \textbf{G}(\boldsymbol{\beta})\textbf{D}_2(\boldsymbol{\beta}_{s2})\boldsymbol{\gamma}^*(\boldsymbol{\beta}) 
\end{pmatrix}
= \widehat{\textbf{b}}(\boldsymbol{\beta}) - \boldsymbol{\beta}_0,
\end{equation}
where $\widehat{\textbf{b}}(\boldsymbol{\beta}) = \boldsymbol{\Omega}^{-1}_n(\boldsymbol{\beta}) \boldsymbol{v}_n(\boldsymbol{\beta})$, $\textbf{D}_1(\boldsymbol{\beta}_{s1}) = \text{diag}(\beta^{-2}_{s1,1},\ldots,\beta^{-2}_{s1,q_n})$ and \begin{eqnarray*}
\textbf{D}_2(\boldsymbol{\beta}_{s2}) = \text{diag}(\beta^{-2}_{s2,q_n+1},\ldots,\beta^{-2}_{s2,p_n}).   
\end{eqnarray*}
We need the following three Lemmas, Lemma~\ref{lemma1}, Lemma~\ref{lemma2}, and Lemma~\ref{lemma3} to prove Theorem~\ref{theorem1}. Note: Although our proofs follow that in \cite{zhao2019simultaneous}, we have made modifications in several places, for example, our Lemma~\ref{lemma3} is different from theirs in the following three aspects:
\begin{itemize}
    \item [1.] In \cite{zhao2019simultaneous}, $\boldsymbol{\Omega}_n^{(1)}$ and $\boldsymbol{v}_n^{(1)}$ are treated as constants while here, we have $\boldsymbol{\Omega}_n^{(1)}=\boldsymbol{\Omega}_n^{(1)}(\boldsymbol{\alpha})$ and $\boldsymbol{v}_n^{(1)}=\boldsymbol{v}_n^{(1)}(\boldsymbol{\alpha})$.
    \item [2.] The domain $H_{n1}$ is defined to have a different form from ${[1/K,K_0]}^{q_n}$.
    \item [3.] The proofs in the following are different due to the $\boldsymbol{\alpha}$ dependence of $\boldsymbol{\Omega}_n^{(1)}=\boldsymbol{\Omega}_n^{(1)}(\boldsymbol{\alpha})$ and $\boldsymbol{v}_n^{(1)}=\boldsymbol{v}_n^{(1)}(\boldsymbol{\alpha})$.
\end{itemize}
Other differences will be discussed in the course of our proofs. 

\begin{lemma}
\label{lemma1}
    Let $\delta$ be a large positive constant. Define $H_{n1} = \{\boldsymbol{\beta}_{s1}: \left\lVert \boldsymbol{\beta}_{s1} - \boldsymbol{\beta}_{0s1} \right\rVert \leq \delta \sqrt{p_n /n} \}$ and $H_{n2} = \{\boldsymbol{\beta}_{s2}: \left\lVert \boldsymbol{\beta}_{s2} - \boldsymbol{\beta}_{0s2}  || = ||\boldsymbol{\beta}_{s2} \right\rVert \leq \delta \sqrt{p_n /n} \}$, $H_n = H_{n1} \otimes H_{n2}$. Then, under conditions (C1)-(C8), with probability tending to 1, we have
\begin{enumerate}
    \item [(i).] $\sup_{\boldsymbol{\beta} \in H_n} \left\lVert \widehat{\textbf{b}}(\boldsymbol{\beta}) - \boldsymbol{\beta}_0 \right\rVert = O_p(\sqrt{p_n /n})$.
    \item [(ii).]$\sup_{\boldsymbol{\beta} \in H_n} \frac{\boldsymbol{\gamma}^*(\boldsymbol{\beta})}{\left\lVert \boldsymbol{\beta}_{s2} \right\rVert} < \frac{1}{c_1}$ for some constant $c_1 > 1$.
    \item [(iii).] $J(\cdot)$ is a mapping from $H_n$ to itself.
\end{enumerate}
\end{lemma}
\begin{proof}[Proof of Lemma~\ref{lemma1}] 
We want to show 
\begin{eqnarray*}
\sup_{\boldsymbol{\beta} \in H_n} \left\lVert \widehat{\textbf{b}}(\boldsymbol{\beta}) - \boldsymbol{\beta}_0 \right\rVert = O_p(\sqrt{p_n /n}).  
\end{eqnarray*}

Since $\boldsymbol{\Omega}_n(\boldsymbol{\beta}) = - \Ddot{\ell}_n(\boldsymbol{\beta} | \widetilde{\boldsymbol{\Lambda}})$ and $ \boldsymbol{v}_n(\boldsymbol{\beta}) = \Dot{\ell}_n(\boldsymbol{\beta} | \widetilde{\boldsymbol{\Lambda}}) - \Ddot{\ell}_n(\boldsymbol{\beta} | \widetilde{\boldsymbol{\Lambda}})\boldsymbol{\beta}$, we have
\begin{eqnarray*}
\widehat{\textbf{b}}(\boldsymbol{\beta})  = \boldsymbol{\Omega}^{-1}_n(\boldsymbol{\beta}) \boldsymbol{v}_n(\boldsymbol{\beta}) &=& [-\Ddot{\ell}_n(\boldsymbol{\beta} | \widetilde{\boldsymbol{\Lambda}})]^{-1}[\Dot{\ell}_n(\boldsymbol{\beta} | \widetilde{\boldsymbol{\Lambda}}) - \Ddot{\ell}_n(\boldsymbol{\beta} | \widetilde{\boldsymbol{\Lambda}})\boldsymbol{\beta}] \\
&=& \boldsymbol{\beta} - [\Ddot{\ell}_n(\boldsymbol{\beta} | \widetilde{\boldsymbol{\Lambda}})]^{-1}[\Dot{\ell}_n(\boldsymbol{\beta} | \widetilde{\boldsymbol{\Lambda}})].
\end{eqnarray*}
By Taylor expansion at $\widetilde{\boldsymbol{\beta}}$, we obtain
\begingroup
\allowdisplaybreaks
\begin{eqnarray*} 
\Dot{\ell}_n(\boldsymbol{\beta} | \widetilde{\boldsymbol{\Lambda}}) = \Dot{\ell}_n(\widetilde{\boldsymbol{\beta}} | \widetilde{\boldsymbol{\Lambda}}) + \Ddot{\ell}_n(\widetilde{\boldsymbol{\beta}}^* | \widetilde{\boldsymbol{\Lambda}})(\boldsymbol{\beta} - \widetilde{\boldsymbol{\beta}}) =\Ddot{\ell}_n(\widetilde{\boldsymbol{\beta}}^* | \widetilde{\boldsymbol{\Lambda}})(\boldsymbol{\beta} - \widetilde{\boldsymbol{\beta}}),
\end{eqnarray*}
\endgroup
where $\widetilde{\boldsymbol{\beta}}^*$ is between $\widetilde{\boldsymbol{\beta}}$ and $\boldsymbol{\beta}$. Then
\begin{equation*}
    \begin{split}
        \widehat{\textbf{b}}(\boldsymbol{\beta})  &= \boldsymbol{\beta} - [\Ddot{\ell}_n(\boldsymbol{\beta} | \widetilde{\boldsymbol{\Lambda}})]^{-1}[\Ddot{\ell}_n(\widetilde{\boldsymbol{\beta}}^* | \widetilde{\boldsymbol{\Lambda}})](\boldsymbol{\beta} - \widetilde{\boldsymbol{\beta}}) \\
        & = \boldsymbol{\beta} - [n^{-1} \Omega_n(\boldsymbol{\beta})]^{-1} [n^{-1} \Omega_n(\widetilde{\boldsymbol{\beta}}^*)](\boldsymbol{\beta} - \widetilde{\boldsymbol{\beta}}).
    \end{split}
\end{equation*}
Since $\left\lVert \widetilde{\boldsymbol{\beta}} - \boldsymbol{\beta}_0 \right\rVert = O_p(\sqrt{p_n / n}) = o_p(1)$, if $\boldsymbol{\beta} \in H_n$, then
\begin{eqnarray*}
\sup_{\boldsymbol{\beta} \in H_n} \left\lVert \boldsymbol{\beta} - \boldsymbol{\beta}_0 \right\rVert \leq \sqrt{2}\delta \sqrt{p_n / n} = O_p(1)    
\end{eqnarray*} 
and 
\begin{eqnarray*}
\left\lVert \widetilde{\boldsymbol{\beta}}^* - \boldsymbol{\beta}_0 \right\rVert = o_p(1).    
\end{eqnarray*} 
By Condition (C4), we have
\begin{eqnarray*}
n^{-1} \boldsymbol{\Omega}_n(\boldsymbol{\beta}) = I(\boldsymbol{\beta}_0) + o_p(1)
\end{eqnarray*}
and
\begin{eqnarray*}
n^{-1} \boldsymbol{\Omega}_n(\widetilde{\boldsymbol{\beta}}^*) = I(\boldsymbol{\beta}_0) + o_p(1)  
\end{eqnarray*}
uniformly for $\boldsymbol{\beta} \in H_n$.
Therefore, 
$$
[n^{-1} \boldsymbol{\Omega}_n(\boldsymbol{\beta})]^{-1} 
=I^{-1}(\boldsymbol{\beta}_0)+o_p(1), 
$$
$$
[n^{-1} \boldsymbol{\Omega}_n(\boldsymbol{\beta})]^{-1} 
[n^{-1} \boldsymbol{\Omega}_n(\widetilde{\boldsymbol{\beta}}^*)]
=\textbf{I}_{p_n}+o_p(1) 
$$
and 
\begin{equation*}
\begin{split}
\widehat{\textbf{b}}(\boldsymbol{\beta}) & = \boldsymbol{\beta} - (\textbf{I}_{p_n} + o_p(1))(\boldsymbol{\beta} - \widetilde{\boldsymbol{\beta}}) \\
& = \boldsymbol{\beta} - (\boldsymbol{\beta} - \widetilde{\boldsymbol{\beta}}) + o_p(1)(\boldsymbol{\beta} - \widetilde{\boldsymbol{\beta}}) \\
& = \widetilde{\boldsymbol{\beta}} + o_p(1)(\boldsymbol{\beta} - \widetilde{\boldsymbol{\beta}}),
\end{split}
\end{equation*}
where $\textbf{I}_{p_n}$ is an $p_n \times p_n$ identity matrix. 
Hence, we have 
\begin{equation*}
\begin{split}
\widehat{\textbf{b}}(\boldsymbol{\beta}) - \boldsymbol{\beta}_0 & = \widetilde{\boldsymbol{\beta}} - \boldsymbol{\beta}_0 + o_p(1)(\boldsymbol{\beta} - \boldsymbol{\beta}_0 - (\widetilde{\boldsymbol{\beta}} - \boldsymbol{\beta}_0)) \\
& = \widetilde{\boldsymbol{\beta}} - \boldsymbol{\beta}_0 + o_p(1)(\boldsymbol{\beta} - \boldsymbol{\beta}_0) + o_p(1)(\widetilde{\boldsymbol{\beta}} - \boldsymbol{\beta}_0).
\end{split}
\end{equation*}
As a result, 
\begin{eqnarray*}
\left\lVert \widehat{\textbf{b}}(\boldsymbol{\beta}) - \boldsymbol{\beta}_0 \right\rVert \leq \left\lVert \widetilde{\boldsymbol{\beta}} - \boldsymbol{\beta}_0 \right\rVert + o_p(1) \left\lVert \boldsymbol{\beta} - \boldsymbol{\beta}_0 \right\rVert +o_p(1) \left\lVert \widetilde{\boldsymbol{\beta}} - \boldsymbol{\beta}_0 \right\rVert,\end{eqnarray*}
and subsequently,
\begin{equation*}
    \begin{split}
        \sup_{\boldsymbol{\beta} \in H_n} \left\lVert \widehat{\textbf{b}}(\boldsymbol{\beta}) - \boldsymbol{\beta}_0 \right\rVert & \leq \left\lVert \widetilde{\boldsymbol{\beta}} - \boldsymbol{\beta}_0 \right\rVert + o_p(1) \sup_{\boldsymbol{\beta} \in H_n} \left\lVert \boldsymbol{\beta} - \boldsymbol{\beta}_0 \right\rVert + o_p(1) \left\lVert \widetilde{\boldsymbol{\beta}} - \boldsymbol{\beta}_0 \right\rVert \\
        & = O_p(\sqrt{p_n/n}) + o_p(1) (\delta \sqrt{p_n/n}) + o_p(1) O_p(\sqrt{p_n/n}) \\
        & = O_p(\sqrt{p_n/n}).
    \end{split}
\end{equation*}
Since we have proved for part (i) that
$$\sup_{\boldsymbol{\beta} \in H_n} \left\lVert \widehat{\textbf{b}}(\boldsymbol{\beta}) - \boldsymbol{\beta}_0 \right\rVert = O_p(\sqrt{p_n /n}),$$ 
it follows from (\ref{A2}) that 
$$
\sup_{\boldsymbol{\beta} \in H_n} \left\lVert \boldsymbol{\gamma}^*(\boldsymbol{\beta}) + \frac{\lambda_n}{n} \textbf{B}^\top(\boldsymbol{\beta})\textbf{D}_1(\boldsymbol{\beta}_{s1})\boldsymbol{\alpha}^*(\boldsymbol{\beta}) + \frac{\lambda_n}{n} \textbf{G}(\boldsymbol{\beta})\textbf{D}_2(\boldsymbol{\beta}_{s2}) \boldsymbol{\gamma}^*(\boldsymbol{\beta})  \right\rVert = O_p(\sqrt{p_n/n}).
$$
In sequel, we assume that for a matrix $\textbf{A}$, $\left\lVert \textbf{A} \right\rVert$ represents the induced 2-norm. Then, using the properties of the matrix 2-norm, we have
\begin{equation*}
\begin{split}
\left\lVert \textbf{B}(\boldsymbol{\beta}) \right\rVert &\leq \left\lVert ((n^{-1} \boldsymbol{\Omega}_n(\boldsymbol{\beta}))^{-1} \right\rVert  = \lambda_{\max}\{ (n^{-1}\boldsymbol{\Omega}_n(\boldsymbol{\beta}))^{-1} \} = [\lambda_{\min}(n^{-1}\boldsymbol{\Omega}_n(\boldsymbol{\beta}))]^{-1} \\
& \leq (1/c_1)^{-1} = c_1, \; \;\text{where $c_1$ is given in Condition (C5).}
\end{split}
\end{equation*}
Put it another way, this is $\sup_{\boldsymbol{\beta} \in H_n} \left\lVert \textbf{B}(\boldsymbol{\beta}) \right\rVert\leq c_1$. Similarly, we have 
\begin{eqnarray*}
\sup_{\boldsymbol{\beta} \in H_n} \left\lVert \textbf{B}^\top(\boldsymbol{\beta}) \right\rVert \leq c_1.   
\end{eqnarray*} 


Next, we want to prove (\ref{A3}):
\begin{equation}\label{A3}
\sup_{\boldsymbol{\beta} \in H_n} \left\lVert \frac{\lambda_n}{n} \textbf{B}^\top (\boldsymbol{\beta}) \textbf{D}_1(\boldsymbol{\beta}_{s1}) \boldsymbol{\alpha}^*(\boldsymbol{\beta}) \right\rVert \leq \bigg(\frac{\lambda_n}{\sqrt{n}} \bigg) O_p(\sqrt{q_n /n}) = o_p(\sqrt{p_n/n}).
\end{equation}
By (\ref{A1}), we have 
\begin{equation*}
    \begin{split}
        J(\boldsymbol{\beta}) &=\begin{pmatrix} \boldsymbol{\alpha}^*(\boldsymbol{\beta}) \\ \boldsymbol{\gamma}^*(\boldsymbol{\beta}) \end{pmatrix}\\
        &=\{ \boldsymbol{\Omega}_n(\boldsymbol{\beta}) + \lambda_n \textbf{D}_n(\boldsymbol{\beta})\}^{-1} \boldsymbol{v}_1(\boldsymbol{\beta}) \\
        &= [(\boldsymbol{\Omega}_n(\boldsymbol{\beta}) + \lambda_n \textbf{D}_n(\boldsymbol{\beta}))^{-1} \boldsymbol{\Omega}_n(\boldsymbol{\beta})]\left[\boldsymbol{\Omega}^{-1}_n(\boldsymbol{\beta})\boldsymbol{v}_1(\boldsymbol{\beta})\right] \\
        &= \left[(\boldsymbol{\Omega}_n(\boldsymbol{\beta}) + \lambda_n \textbf{D}_n(\boldsymbol{\beta}))^{-1} \boldsymbol{\Omega}_n(\boldsymbol{\beta})\right] \widehat{\textbf{b}}(\boldsymbol{\beta}).
    \end{split}
\end{equation*}
By (C5), there exists a constant $M > 0$, such that
$$
\sup_{\boldsymbol{\beta} \in H_n} \left\lVert (\boldsymbol{\Omega}_n(\boldsymbol{\beta}) + \lambda_n \textbf{D}_n(\boldsymbol{\beta}))^{-1} \boldsymbol{\Omega}_n(\boldsymbol{\beta}) \right\rVert \leq M.
$$
Then, we have
$$
\left\lVert J(\boldsymbol{\beta}) \right\rVert \leq \left\lVert (\boldsymbol{\Omega}_n(\boldsymbol{\beta}) + \lambda_n \textbf{D}_n(\boldsymbol{\beta}))^{-1} \right\rVert \left\lVert \widehat{\textbf{b}}(\boldsymbol{\beta}) \right\rVert \leq M \left\lVert \widehat{\textbf{b}}(\boldsymbol{\beta})\right\rVert.
$$
On the other hand, $\left\lVert J(\boldsymbol{\beta}) \right\rVert^2 = \left\lVert \boldsymbol{\alpha}^*(\boldsymbol{\beta}) \right\rVert^2 + \left\lVert \boldsymbol{\gamma}^*(\boldsymbol{\beta}) \right\rVert^2$, and
\begingroup
\allowdisplaybreaks
\begin{eqnarray*}
    \left\lVert \boldsymbol{\widehat{b}}(\boldsymbol{\beta}) \right\rVert=\left\lVert \boldsymbol{\widehat{b}}(\boldsymbol{\beta})-\boldsymbol{\beta}_0 \right\rVert&\leq&o_p\left(\sqrt{p_n/n} \right)+a_1\sqrt{q_n}\\
    &=&O_p(q_n),
\end{eqnarray*}
\endgroup
i.e.,
\begin{eqnarray}
    \label{B1}
    \sup_{\boldsymbol{\beta}\in H_n}\left\lVert \widehat{\boldsymbol{b}}(\boldsymbol{\beta})\right\rVert=O_p(\sqrt{q_n}).
\end{eqnarray}
By Lemma~\ref{lemma1} (i) and condition (C7), we also have
$$
\left\lVert \boldsymbol{\alpha}^*(\boldsymbol{\beta}) \right\rVert \leq \left\lVert J(\boldsymbol{\beta}) \right\rVert \leq M \left\lVert \widehat{\textbf{b}}(\boldsymbol{\beta}) \right\rVert.
$$
Then
\begin{eqnarray}
\label{B2}
\sup_{\boldsymbol{\beta} \in H_n} \left\lVert \boldsymbol{\alpha}^*(\boldsymbol{\beta})\right\rVert \leq M \sup_{\boldsymbol{\beta} \in H_n} \left\lVert \widehat{\textbf{b}}(\boldsymbol{\beta})\right\rVert = O_p(\sqrt{q_n}).
\end{eqnarray}
Now, we consider $\left\lVert \textbf{D}_1(\boldsymbol{\beta}_{s1})\right\rVert$. Since
$$
\left\lVert \textbf{D}_1(\boldsymbol{\beta}_{s1})\right\rVert = \lambda_{\max}\{ \textbf{D}_1(\boldsymbol{\beta}_{s1})\} = \max_{1 \leq j \leq q_n} \bigg\{\frac{1}{\beta^2_{s1j}} \bigg\} = \frac{1}{\min_{1 \leq j \leq q_n} \{1/\beta^2_{s1j} \}},
$$
when $\boldsymbol{\beta} \in H_n$,  we have $\left\lVert \boldsymbol{\beta} - \boldsymbol{\beta}_0 \right\rVert \leq \sqrt{2} \delta \sqrt{p_n/n}$, then
\begin{eqnarray*}
\left\lVert \beta_{s1j} - \beta_{0j}\right\rVert = | \beta_{s1j} - \beta_{0j}| \leq \delta \sqrt{p_n /n},   
\end{eqnarray*} 
i.e.,
$$
|\beta_{0j}| - \delta\sqrt{\frac{p_n}{n}} \leq |\beta_{s1j}| \leq |\beta_{0j}| + \delta \sqrt{\frac{p_n}{n}},~1 \leq j \leq q_n.
$$
By Condition (C7), when $n$ is sufficiently large, we have
$$
\frac{a_0}{2} \leq |\beta_{s1j}| \leq 2a_1,
$$
because $\delta \sqrt{p_n/n} \rightarrow 0, \text{ as } n \rightarrow \infty$. Then, 
$$
\big( \frac{a_0}{2} \big)^2 \leq \min_{1 \leq j \leq q_n} \{\beta^2_{s1j}\} \leq \max_{1 \leq j \leq q_n} \{\beta^2_{s1j}\} \leq (2a_1)^2
$$
and
$$
\left\lVert\textbf{D}_1(\boldsymbol{\beta}_{s1})\right\rVert = \frac{1}{\min_{1 \leq j \leq q_n} \{\beta^2_{s1j}\}} \leq \frac{1}{(a_0/2)^2} = \frac{4}{a^2_0}.
$$
This implies
\begin{eqnarray}
\label{B3}
\sup_{\boldsymbol{\beta} \in H_n} \left\lVert \textbf{D}_1(\boldsymbol{\beta}_{s1})\right\rVert \leq 4/a^2_0.
\end{eqnarray}
Therefore, by (\ref{B1}), (\ref{B2}), and (\ref{B3}), we have
\begin{equation*}
    \begin{split}
        \sup_{\boldsymbol{\beta} \in H_n}  \left\lVert \frac{\lambda_n}{n} \textbf{B}^\top (\boldsymbol{\beta}) \textbf{D}_1(\boldsymbol{\beta}_{s1}) \boldsymbol{\alpha}^*(\boldsymbol{\beta}) \right\rVert & \leq \frac{\lambda_n}{n} \cdot \sup_{\boldsymbol{\beta} \in H_n} \left\lVert \textbf{B}^\top (\boldsymbol{\beta})\right\rVert \cdot \sup_{\boldsymbol{\beta} \in H_n} \left\lVert\textbf{D}_1(\boldsymbol{\beta}_{s1}) \right\rVert \cdot \sup_{\boldsymbol{\beta} \in H_n} \left\lVert \boldsymbol{\alpha}^*(\boldsymbol{\beta})\right\rVert \\
        & \leq \frac{\lambda_n}{n} \cdot c_1 \cdot \frac{4}{a^2_0} \cdot O_p(\sqrt{q_n}) \\
        & = \frac{\lambda_n}{\sqrt{n}} \frac{4c_1}{a^2_0} O_p(\sqrt{q_n \ n}). 
    \end{split}
\end{equation*}
Since Condition (C6) states that $\lambda_n / \sqrt{n} \rightarrow 0$, then
$$
 \sup_{\boldsymbol{\beta} \in H_n} \left\lVert \frac{\lambda_n}{n} \textbf{B}^\top (\boldsymbol{\beta}) \textbf{D}_1(\boldsymbol{\beta}_{s1}) \boldsymbol{\alpha}^*(\boldsymbol{\beta}) \right\rVert = o(1) \cdot O_p(\sqrt{q_n / n}) = o_p(\sqrt{q_n / n}) = o_p(\sqrt{p_n / n}).
$$
The proof of (\ref{A3}) is completed. 

Next, we prove (\ref{A4}):
\begin{equation} \label{A4}
    c_0^{-1} \left\lVert\frac{\lambda_n}{n} \textbf{D}_2(\boldsymbol{\beta}_{s2}) \boldsymbol{\gamma}^*(\boldsymbol{\beta}) \right\rVert - \left\lVert \boldsymbol{\gamma}^*(\boldsymbol{\beta}) \right\rVert \leq o_p(\delta \sqrt{p_n / n}).
\end{equation}
From (\ref{A2}), we obtain
$$
\left\lVert\boldsymbol{\gamma}^*(\boldsymbol{\beta}) + \frac{\lambda_n}{n} \{\textbf{B}^\top (\boldsymbol{\beta}) \textbf{D}_1(\boldsymbol{\beta}_{s1}) \boldsymbol{\alpha}^*(\boldsymbol{\beta}) + \textbf{G}(\boldsymbol{\beta})\textbf{D}_2(\boldsymbol{\beta}_{s2})\boldsymbol{\gamma}^*(\boldsymbol{\beta}) \} \right\rVert \leq \left\lVert \widehat{\textbf{b}}(\boldsymbol{\beta}) - \boldsymbol{\beta}_0 \right\rVert.
$$
It implies 
\begin{eqnarray}
\label{E1}
&&\left\lVert\frac{\lambda_n}{n} \textbf{G}(\boldsymbol{\beta})\textbf{D}_2(\boldsymbol{\beta}_{s2})\boldsymbol{\gamma}^*(\boldsymbol{\beta})\right\rVert - \left\lVert\boldsymbol{\gamma}^*(\boldsymbol{\beta})\right\rVert - \left\Vert\frac{\lambda_n}{n}\textbf{B}^\top (\boldsymbol{\beta}) \textbf{D}_1(\boldsymbol{\beta}_{s1}) \boldsymbol{\alpha}^*(\boldsymbol{\beta})\right\rVert\nonumber\\
&&\leq  \left\lVert \widehat{\textbf{b}}(\boldsymbol{\beta}) - \boldsymbol{\beta}_0 \right\rVert.
\end{eqnarray}
Now, consider 
\begin{equation*}
  \begin{split}
    \left\lVert \frac{\lambda_n}{n} \textbf{D}_2(\boldsymbol{\beta}_{s2}) \boldsymbol{\gamma}^*(\boldsymbol{\beta}) \right\rVert & = \left\lVert \frac{\lambda_n}{n} \textbf{G}^{-1}(\boldsymbol{\beta})\textbf{G}(\boldsymbol{\beta})\textbf{D}_2(\boldsymbol{\beta}_{s2})\boldsymbol{\gamma}^*(\boldsymbol{\beta}) \right\rVert \\
    & \leq \left\lVert \textbf{G}^{-1}(\boldsymbol{\beta}) \right\rVert \cdot \left\lVert\frac{\lambda_n}{n} \textbf{G}(\boldsymbol{\beta})\textbf{D}_2(\boldsymbol{\beta}_{s2})\boldsymbol{\gamma}^*(\boldsymbol{\beta}) \right\rVert,
  \end{split}
\end{equation*}
which yields
$$
\left\lVert \frac{\lambda_n}{n} \textbf{G}(\boldsymbol{\beta})\textbf{D}_2(\boldsymbol{\beta}_{s2})\boldsymbol{\gamma}^*(\boldsymbol{\beta})\right\rVert \geq \frac{1}{\left\lVert\textbf{G}^{-1}(\boldsymbol{\beta})\right\rVert} \left\lVert\frac{\lambda_n}{n} \textbf{D}_2(\boldsymbol{\beta}_{s2})\boldsymbol{\gamma}^*(\boldsymbol{\beta}) \right\rVert.
$$
Since by Condition (C5) and the proofs given below, we have 
\begin{equation*}
    \begin{split}
        \left\lVert\textbf{G}^{-1}(\boldsymbol{\beta})\right\rVert & = \lambda_{\max}\{\textbf{G}^{-1}(\boldsymbol{\beta}) \} = \frac{1}{\lambda_{\min}\{\textbf{G}(\boldsymbol{\beta})\}} \leq \frac{1}{\inf_{\boldsymbol{\beta \in H_n}} \lambda_{\min} \{ \textbf{G}(\boldsymbol{\beta})\}} \\
        & \leq 1/(1/c_0) = c_0,
    \end{split}
\end{equation*}
then $1/||\textbf{G}^{-1}(\boldsymbol{\beta})|| \geq 1/c_0$,  $\inf_{\boldsymbol{\beta \in H_n}} \{1 /||\textbf{G}^{-1}(\boldsymbol{\beta})|| \} \geq 1/c_0$, and
\begin{eqnarray}
\label{E2}
\left\lVert \frac{\lambda_n}{n} \textbf{G}(\boldsymbol{\beta}) \textbf{D}_2(\boldsymbol{\beta}_{s2}) \boldsymbol{\gamma}^*(\boldsymbol{\beta}) \right\rVert \geq \frac{1}{c_0} \left\lVert \frac{\lambda_n}{n} \textbf{D}_2(\boldsymbol{\beta}_{s2}) \boldsymbol{\gamma}^*(\boldsymbol{\beta}) \right\rVert.
\end{eqnarray}
Finally,  (\ref{E1}), (\ref{E2}), (\ref{A3}) and Lemma~\ref{lemma1} (ii) together imply (\ref{A4}).
Here we explain why $\inf_{\boldsymbol{\beta} \in H_n} \{\lambda_{\min} \{\textbf{G}(\boldsymbol{\beta}) \} \} \geq 1/c_0$. This is due to
\begin{equation*}
    \begin{split}
 \lambda_{\min} \{\textbf{G}(\boldsymbol{\beta}) \} & \geq \lambda_{\min} \{ (n^{-1} \boldsymbol{\Omega}_n(\boldsymbol{\beta}))^{-1} \} \\
        & = \lambda_{\max} \{n^{-1} \boldsymbol{\Omega}_n(\boldsymbol{\beta}) \} \\
        & \geq \lambda_{\min} \{n^{-1} \boldsymbol{\Omega}_n(\boldsymbol{\beta}) \} \\
        & \geq \inf_{\boldsymbol{\beta} \in H_n} \{ \lambda_{\min} \{n^{-1} \boldsymbol{\Omega}_n(\boldsymbol{\beta}) \} \}\\
        & \geq  1/c_0, \;\; \text{(by Condition (C5)).}
    \end{split}
\end{equation*}
Let
\begin{equation*}
\frac{m_{\boldsymbol{\gamma}^*(\boldsymbol{\beta})}}{\boldsymbol{\beta}_{s2}} = (\boldsymbol{\gamma}_1^*(\boldsymbol{\beta})/\beta_{s2,q_n + 1},\ldots,\boldsymbol{\gamma}_{p_n-q_n}^*(\boldsymbol{\beta}) /\beta_{s2,p_n})^\top.
\end{equation*}
Then, $m_{\boldsymbol{\gamma}^*(\boldsymbol{\beta})}/\boldsymbol{\beta}_{s2} = \text{diag}(\boldsymbol{\beta}_{s2}) \textbf{D}_2(\boldsymbol{\beta}_{s2}) \boldsymbol{\gamma}^*(\boldsymbol{\beta}) $, and therefore
\begin{eqnarray}
\label{E3}
\left\lVert \frac{m_{\boldsymbol{\gamma}^*(\boldsymbol{\beta})}}{\boldsymbol{\beta}_{s2}} \right\rVert &\leq& \left\lVert \text{diag}(\boldsymbol{\beta}_{s2}) \right\rVert \cdot \left\lVert \textbf{D}_2(\boldsymbol{\beta}_{s2}) \boldsymbol{\gamma}^*(\boldsymbol{\beta}) \right\rVert \nonumber\\
&=& \sqrt{\left\lVert \text{diag}(\boldsymbol{\beta}_{s2}) \right\rVert^2} \cdot \left\lVert\textbf{D}_2(\boldsymbol{\beta}_{s2}) \boldsymbol{\gamma}^*(\boldsymbol{\beta})\right\rVert \nonumber\\
&=&\sqrt{\max_{q_n + 1 \leq j \leq p_n} \beta^2_{s2j}} \cdot \left\lVert\textbf{D}_2(\boldsymbol{\beta}_{s2}) \boldsymbol{\gamma}^*(\boldsymbol{\beta})\right\rVert\nonumber\\
&\leq& \left\lVert\boldsymbol{\beta}_{s2} \right\rVert \cdot \left\lVert\textbf{D}_2(\boldsymbol{\beta}_{s2}) \boldsymbol{\gamma}^*(\boldsymbol{\beta})\right\rVert\nonumber\\
&\leq& \delta \sqrt{p_n/n} \left\lVert\textbf{D}_2(\boldsymbol{\beta}_{s2}) \boldsymbol{\gamma}^*(\boldsymbol{\beta})\right\rVert. 
\end{eqnarray}
Write $\boldsymbol{\gamma}^*(\boldsymbol{\beta})=\text{diag}(\boldsymbol{\beta}_{s2})\frac{\boldsymbol{m}_{\boldsymbol{\gamma}^*(\boldsymbol{\beta})}}{\boldsymbol{\beta}_{s2}}$, then 
\begin{eqnarray}
\label{E4}
    \left\lVert \boldsymbol{\gamma}^*(\boldsymbol{\beta})\right\rVert\leq \left\lVert \text{diag}(\boldsymbol{\beta}_{s2})\right\rVert\left\lVert\frac{\boldsymbol{m}_{\boldsymbol{\gamma}^*(\boldsymbol{\beta})}}{\boldsymbol{\beta}_{s2}}\right\rVert&\leq& \left\lVert \boldsymbol{\beta}_{s2}\right\rVert\left\lVert\frac{\boldsymbol{m}_{\boldsymbol{\gamma}^*(\boldsymbol{\beta})}}{\boldsymbol{\beta}_{s2}}\right\rVert \nonumber\\
&\leq&\delta\sqrt{p_n/n}\left\lVert\frac{\boldsymbol{m}_{\boldsymbol{\gamma}^*(\boldsymbol{\beta})}}{\boldsymbol{\beta}_{s2}}\right\rVert.
\end{eqnarray}
(\ref{E2}) and (\ref{E3}) imply 
\begin{eqnarray}
    \label{E5}
    \left\lVert \frac{\lambda_n}{n}\boldsymbol{G}(\boldsymbol{\beta})\boldsymbol{D}_2(\boldsymbol{\beta}_{s2})\boldsymbol{\gamma}^*(\boldsymbol{\beta})) \right\rVert &\geq& (1/c_0)(\lambda_n/n)\left\lVert \boldsymbol{D}_2(\boldsymbol{\beta}_{s2})\boldsymbol{\gamma}^*(\boldsymbol{\beta})) \right\rVert \nonumber\\
    &\geq& (1/c_0)(\lambda_n/n)\left(\frac{\sqrt{n}}{\delta\sqrt{p_n}} \right)\left\lVert\frac{\boldsymbol{m}_{\boldsymbol{\gamma}^*(\boldsymbol{\beta})}}{\boldsymbol{\beta}_{s2}}\right\rVert.
\end{eqnarray}
By (\ref{E1}), (\ref{E4}), (\ref{E5}), and Lemma~\ref{lemma1}(i), we can conclude that
\begin{eqnarray*}
    (1/c_0)(\lambda_n/n)\left(\frac{\sqrt{n}}{\delta \sqrt{p_n}} \right)\left\lVert\frac{\boldsymbol{m}_{\boldsymbol{\gamma}^*(\boldsymbol{\beta})}}{\boldsymbol{\beta}_{s2}}\right\rVert-\delta\sqrt{p_n/n}\left\lVert\frac{\boldsymbol{m}_{\boldsymbol{\gamma}^*(\boldsymbol{\beta})}}{\boldsymbol{\beta}_{s2}}\right\rVert\leq o_p\left(\delta\sqrt{p_n/n} \right).
\end{eqnarray*}
Therefore, 
\begin{eqnarray*}
    \left[\frac{\lambda_n}{c_0n}\left(\frac{\sqrt{n}}{\delta\sqrt{p_n}} \right)^2-1 \right]\left\lVert\frac{\boldsymbol{m}_{\boldsymbol{\gamma}^*(\boldsymbol{\beta})}}{\boldsymbol{\beta}_{s2}}\right\rVert\leq o_p(1),
\end{eqnarray*}
and since $\lambda_n/(p_n\delta^2)\longrightarrow 0$, we obtain
\begin{eqnarray*}
    \left\lVert\frac{\boldsymbol{m}_{\boldsymbol{\gamma}^*(\boldsymbol{\beta})}}{\boldsymbol{\beta}_{s2}}\right\rVert\leq \frac{1}{{\frac{\lambda_n}{c_0  p_n\delta^2 }-1}}o_p(1)=o_p(1),
\end{eqnarray*}
which implies \begin{eqnarray}
  \label{A6}\sup_{\boldsymbol{\beta}\in H_{n}} \left\lVert\frac{\boldsymbol{m}_{\boldsymbol{\gamma}^*(\boldsymbol{\beta})}}{\boldsymbol{\beta}_{s2}}\right\rVert=o_p(1).\end{eqnarray}
It follows from (\ref{E4}) and (\ref{A6}) that
\begin{eqnarray}
    \label{A7}
    \left\lVert 
    \boldsymbol{\gamma}^*(\boldsymbol{\beta})\right\rVert
    \leq \left\lVert 
\boldsymbol{\beta}_{s2}\right\rVert\left\lVert\frac{\boldsymbol{m}_{\boldsymbol{\gamma}^*(\boldsymbol{\beta})}}{\boldsymbol{\beta}_{s2}}\right\rVert
\leq (\delta\sqrt{p_n/n} ) o_p(1).
\end{eqnarray}
Hence, 
\begin{eqnarray*}
    \sup_{\boldsymbol{\beta}\in H_{n}} \left\{\frac{\left\lVert 
    \boldsymbol{\gamma}^*(\boldsymbol{\beta})\right\rVert}{\left\lVert 
    \boldsymbol{\beta}_{s2}\right\rVert} \right\}\leq \sup_{\boldsymbol{\beta}\in H_{n}} \left\lVert\frac{\boldsymbol{m}_{\boldsymbol{\gamma}^*(\boldsymbol{\beta})}}{\boldsymbol{\beta}_{s2}}\right\rVert=o_p(1),
\end{eqnarray*}
which implies that Lemma~\ref{lemma1} (ii) holds.  

To prove Lemma~\ref{lemma1} (iii), from (\ref{E4}) and (\ref{A6}), we have already shown that, with probability tending to 1, 
\begin{eqnarray*}
\left\lVert\boldsymbol{\gamma}^*(\boldsymbol{\beta})\right\rVert \leq o_p(1)\delta\sqrt{p_n/n}\leq \delta\sqrt{p_n/n}.  
\end{eqnarray*}
Therefore, we are left to show that
\begin{eqnarray*}
    \left\lVert \boldsymbol{\alpha}^*(\boldsymbol{\beta})-\boldsymbol{\beta}_{0s_1}\right\rVert\leq \delta \sqrt{p_n/n}
\end{eqnarray*}
with probability tending to 1.

Similar to the proof of (\ref{A3}), we have 
\begin{eqnarray*}
    \sup_{\boldsymbol{\beta}\in H_{n}}\left\lVert \frac{\lambda_n}{n}\boldsymbol{A}(\boldsymbol{\beta})\boldsymbol{D}_1(\boldsymbol{\beta}_{s1}) \boldsymbol{\alpha}^*(\boldsymbol{\beta})\right\rVert=o_p(\sqrt{p_n/n})=o_p(\delta\sqrt{p_n/n}).
\end{eqnarray*}
Subsequently, from (\ref{A2}), we have 
\begin{eqnarray}\label{A8}
\sup_{\boldsymbol{\beta}\in H_{n}}\left\lVert \boldsymbol{\alpha}^*(\boldsymbol{\beta})-\boldsymbol{\beta}_{0s1}+\frac{\lambda_n}{n} \boldsymbol{B}(\boldsymbol{\beta})\boldsymbol{D}_2(\boldsymbol{\beta}_{s2})\boldsymbol{\gamma}^*(\boldsymbol{\beta})\right\rVert=o_p\left( \delta \sqrt{p_n/n}\right).
\end{eqnarray}
According to (\ref{A4}) and (\ref{A7}), we have 
\begingroup
\allowdisplaybreaks
\begin{eqnarray*}
    c_1^{-1}\left\lVert \frac{\lambda_n}{n}\boldsymbol{D}_2(\boldsymbol{\beta_{s2}})\boldsymbol{\gamma}^*(\boldsymbol{\beta}) \right\rVert &\leq&\left\lVert \boldsymbol{\gamma}^*(\boldsymbol{\beta}) \right\rVert +o_p\left(\delta\sqrt{p_n/n} \right)\\
    &\leq& o_p\left(\delta\sqrt{p_n/n} \right)+o_p\left(\delta\sqrt{p_n/n} \right)\\
    &=&o_p\left(\delta\sqrt{p_n/n} \right).
\end{eqnarray*}
\endgroup
Then $\left\lVert (\lambda_n/n) \boldsymbol{D}_2(\boldsymbol{\beta_{s2}})\boldsymbol{\gamma}^*(\boldsymbol{\beta}) \right\rVert \leq c_1\cdot o_p\left(\delta\sqrt{p_n/n}\right)=o_p\left(\delta\sqrt{p_n/n} \right)$, and therefore, 
\begin{equation}
\begin{split}
    \label{A9}
    \sup_{\boldsymbol{\beta}\in H_n}\left\lVert \frac{\lambda_n}{n}\boldsymbol{B}(\boldsymbol{\beta})\boldsymbol{D}_2(\boldsymbol{\beta_{s2}})\boldsymbol{\gamma}^*(\boldsymbol{\beta}) \right\rVert& \leq \left\lVert \boldsymbol{B}(\boldsymbol{\beta}) \right\rVert \left\lVert \frac{\lambda_n}{n}\boldsymbol{D}_2(\boldsymbol{\beta_{s2}})\boldsymbol{\gamma}^*(\boldsymbol{\beta}) \right\rVert\nonumber\\
    & \leq c_1 \cdot o_p\left(\delta\sqrt{p_n/n} \right)\nonumber\\
    & =o_p\left( \delta\sqrt{p_n/n}\right).
    \end{split}
\end{equation}


Thus, (\ref{A8}) and (\ref{A9}) yield
\begin{eqnarray}
    \label{extra2}
    \sup_{\boldsymbol{\beta}\in H_n}\left\lVert \boldsymbol{\alpha}^*(\boldsymbol{\beta})-\boldsymbol{\beta}_{0s1}\right\rVert\leq o_p\left(\delta\sqrt{p_n/n}\right).
\end{eqnarray}
The inequality of (\ref{extra2}) implies that, with probability tending to 1, $\forall \boldsymbol{\beta} \in H_n$, we have
\begin{eqnarray*}
    \left\lVert \boldsymbol{\alpha}^*(\boldsymbol{\beta})-\boldsymbol{\beta}_{0s1}\right\rVert\leq \delta\sqrt{p_n/n}
\end{eqnarray*}
for large $n$, and hence, Lemma~\ref{lemma1} (iii) holds.
\end{proof}

Let $\boldsymbol{\beta}_{s1}=\boldsymbol{\alpha}$ and $\boldsymbol{\beta}_{s2}=\boldsymbol{0}$ in $\boldsymbol{\Omega}_n(\boldsymbol{\beta})$ and $\boldsymbol{v}_n(\boldsymbol{\beta})$,  we define $\boldsymbol{\Omega}_n(\boldsymbol{\alpha})=\boldsymbol{\Omega}_n(\boldsymbol{\beta})$ when $\boldsymbol{\beta}_{s1}=\boldsymbol{\alpha}$ and $\boldsymbol{\beta}_{s2}=\boldsymbol{0}$. Similarly, define $\boldsymbol{v}_n(\boldsymbol{\alpha})=\boldsymbol{v}_n(\boldsymbol{\beta})$ when $\boldsymbol{\beta}_{s1}=\boldsymbol{\alpha}$ and $\boldsymbol{\beta}_{s2}=\boldsymbol{0}$. The same applies to $\boldsymbol{\Omega}_n^{(1)}(\boldsymbol{\alpha})$ and $\boldsymbol{v}_n^{(1)}(\boldsymbol{\alpha})$. We have the following lemma. 

\begin{lemma}\label{lemma2}
(A matrix calculus identity): Assume a vector $\boldsymbol{\alpha} \in \mathbb{R}^{q_n}$, $q_n\geq 1$, $f$ is a mapping from $\mathbb{R}^{q_n}$ to $\mathbb{R}^{q_n}$ defined by $f(\boldsymbol{\alpha})=(f_1(\boldsymbol{\alpha}),\ldots,f_{q_n}(\boldsymbol{\alpha}))^\top$, and $f$ is differentiable. Also, $\boldsymbol{\omega}(\boldsymbol{\alpha})$ is a $q_n\times q_n$ matrix and a mapping from $\mathbb{R}^{q_n}$ to $\mathbb{R}^{q_n\times q_n}$ and differentiable. Then 
\begin{eqnarray*}
    \frac{\partial \left[\boldsymbol{\omega}(\boldsymbol{\alpha})f(\boldsymbol{\alpha}))\right]}{\partial \boldsymbol{\alpha}^\top}=\boldsymbol{\omega}(\boldsymbol{\alpha})\frac{\partial f(\boldsymbol{\alpha})}{\partial \boldsymbol{\alpha}^\top}+\begin{pmatrix}
        f^\top(\boldsymbol{\alpha})&\ldots&0\\
        \vdots&\ddots&\vdots\\
        0&\ldots&f^\top(\boldsymbol{\alpha})
    \end{pmatrix}
    \begin{pmatrix}
        \left( \frac{\partial \boldsymbol{\omega}_1^\top (\boldsymbol{\alpha})}{\partial \boldsymbol{\alpha}}\right)^\top\\
        \vdots\\
        \left( \frac{\partial \boldsymbol{\omega}_{q_n}^\top (\boldsymbol{\alpha})}{\partial \boldsymbol{\alpha}}\right)^\top
    \end{pmatrix},
\end{eqnarray*}
where the two matrices in the last term of the above equation are block matrices, $\boldsymbol{\omega}_j^\top (\boldsymbol{\alpha})$ is the $j$th row of $\boldsymbol{\omega}(\boldsymbol{\alpha})$ and $\partial \boldsymbol{\omega}_j^\top (\boldsymbol{\alpha})/\partial \boldsymbol{\alpha}$ is a $q_n\times q_n$ matrix, $1\leq j \leq q_n$.
\end{lemma} 
Since Lemma \ref{lemma2} can be proved easily, we omit the proof. 
\begin{lemma}
\label{lemma3}
Under the conditions (C1)-(C9), with probability tending to 1, the equation $\boldsymbol{\alpha}=(\boldsymbol{\Omega}_n^{(1)}(\boldsymbol{\alpha})+\lambda_n \boldsymbol{D}_1(\boldsymbol{\alpha}))^{-1}\boldsymbol{v}_n^{(1)}(\alpha)$ has a unique fixed-point $\widehat{\alpha}^{*}$ in the domain $H_{n1}$.  
\end{lemma} 
Before we prove this lemma, we want to mention that our proofs are different from those in the literature for the BAR estimator under other settings of models and data, such as \cite{zhao2019simultaneous}. 
The following points merit consideration here:
\begin{itemize}
     \item [1.] Two expressions $\boldsymbol{\Omega}_n^{(1)}$ and $\boldsymbol{v}_n^{(1)}$
     are written as functions of $\alpha$ to emphasize that they depend on $\alpha$ and cannot be treated as constants in Lemma~\ref{lemma3}.
    \item [2.] In order to prove Lemma~\ref{lemma3}, we need Lemma~\ref{lemma2} and a new condition (C9) to deal with more complicated and high-order terms in the proofs. 
    \item [3.] In Theorem~\ref{theorem1}, we showed that the limiting variance is not necessarily a sandwich form, it implies that the BAR estimator is semiparametricaly efficient. 
\end{itemize}
\begin{proof}[Proof of Lemma~\ref{lemma3}] 
Define 
\begin{eqnarray}
\label{A10}
    f(\boldsymbol{\alpha})=(f_1(\boldsymbol{\alpha}),\ldots,f_{q_n}(\boldsymbol{\alpha}))^\top\equiv (\boldsymbol{\Omega}_n^{(1)}(\boldsymbol{\alpha})+\lambda_n \boldsymbol{D}_1(\boldsymbol{\alpha}))^{-1}\boldsymbol{v}_n^{(1)}(\boldsymbol{\alpha}),
\end{eqnarray}
where $\boldsymbol{\alpha}=(\boldsymbol{\alpha}_1,\ldots,\boldsymbol{\alpha}_{q_n})^\top$. By multiplying $(\boldsymbol{\Omega}_n^{(1)}(\boldsymbol{\alpha}))^{-1}(\boldsymbol{\Omega}_n^{(1)}(\boldsymbol{\alpha})+\lambda_n \boldsymbol{D}_1(\boldsymbol{\alpha}))$ and subtracting $\boldsymbol{\beta}_{0s1}$ on both sides of (\ref{A10}), we have 
\begin{eqnarray}
\label{F1}
    f(\boldsymbol{\alpha})-\boldsymbol{\beta}_{0s1}+\lambda_n(\boldsymbol{\Omega}_{n}^{(1)}(\boldsymbol{\alpha})\boldsymbol{D}_1(\boldsymbol{\alpha}))f(\boldsymbol{\alpha})=(\boldsymbol{\Omega}_n^\top(\boldsymbol{\alpha}))^{-1}\boldsymbol{v}_n^{(1)}(\boldsymbol{\alpha})-\boldsymbol{\beta}_{0s1},
\end{eqnarray}
where $\boldsymbol{\Omega}_n(\boldsymbol{\alpha})=\boldsymbol{X}^\top(\boldsymbol{\alpha})\boldsymbol{X}(\boldsymbol{\alpha})$, $\boldsymbol{v}_n(\boldsymbol{\alpha})=\boldsymbol{X}^\top (\boldsymbol{\alpha})\boldsymbol{W}(\boldsymbol{\alpha})$  by Cholesky decomposition, and $\boldsymbol{W}(\boldsymbol{\alpha})$ is the pseudo response vector. 
Let $\boldsymbol{X}(\boldsymbol{\alpha})=(\boldsymbol{X}_1(\boldsymbol{\alpha}),\boldsymbol{X}_2(\boldsymbol{\alpha}))$, $\boldsymbol{X}_1(\boldsymbol{\alpha})$ is a $p_n \times q_n$ matrix and $\boldsymbol{X}_2(\boldsymbol{\alpha})$ is a $p_n\times (p_n - q_n)$ matrix. Then 
\begin{eqnarray*}
    \boldsymbol{X}^\top(\boldsymbol{\alpha})=\begin{pmatrix}
      \boldsymbol{X}_1^\top(\boldsymbol{\alpha}) \\
      \boldsymbol{X}_2^\top(\boldsymbol{\alpha})
    \end{pmatrix},
\end{eqnarray*}
\begin{equation*}
\begin{split}
\boldsymbol{\Omega}_n(\boldsymbol{\alpha})=\boldsymbol{X}^\top(\boldsymbol{\alpha})\boldsymbol{X}(\boldsymbol{\alpha})& =\begin{pmatrix}
      \boldsymbol{X}_1^\top(\boldsymbol{\alpha}) \\
      \boldsymbol{X}_2^\top(\boldsymbol{\alpha})
    \end{pmatrix}(\boldsymbol{X}_1(\boldsymbol{\alpha}), \boldsymbol{X}_2(\boldsymbol{\alpha}))\\
    & =\begin{pmatrix}
     \boldsymbol{X}_1^\top(\boldsymbol{\alpha}) \boldsymbol{X}_1(\boldsymbol{\alpha}) & \boldsymbol{X}_1^\top(\boldsymbol{\alpha})\boldsymbol{X}_2(\boldsymbol{\alpha}) \\
     \boldsymbol{X}_2^\top(\boldsymbol{\alpha}) \boldsymbol{X}_1(\boldsymbol{\alpha}) & \boldsymbol{X}_2^\top(\boldsymbol{\alpha}) \boldsymbol{X}_2(\boldsymbol{\alpha})
    \end{pmatrix}.
    \end{split}
\end{equation*}
We obtain $\boldsymbol{\Omega}_n^{(1)}(\boldsymbol{\alpha})=\boldsymbol{X}_1^\top(\boldsymbol{\alpha})\boldsymbol{X}_1(\boldsymbol{\alpha})$, $\boldsymbol{v}_n^{(1)}(\boldsymbol{\alpha})=\boldsymbol{X}_1^\top(\boldsymbol{\alpha})\boldsymbol{W}(\boldsymbol{\alpha})$, and 
\begin{eqnarray*}
    \boldsymbol{v}_n(\boldsymbol{\alpha})=\boldsymbol{X}^\top(\boldsymbol{\alpha})\boldsymbol{W}(\boldsymbol{\alpha})=\begin{pmatrix}
        \boldsymbol{X}_1^\top(\boldsymbol{\alpha})\boldsymbol{W}(\boldsymbol{\alpha}) \\
        \boldsymbol{X}_2^\top(\boldsymbol{\alpha})\boldsymbol{W}(\boldsymbol{\alpha})
    \end{pmatrix}.
\end{eqnarray*}
Thus, in (\ref{F1}), we have 
\begingroup
\allowdisplaybreaks
\begin{eqnarray}
    \label{F2}
    (\boldsymbol{\Omega}_n^{(1)}(\boldsymbol{\alpha}))^{-1}\boldsymbol{v}_n^{(1)}(\boldsymbol{\alpha})- \boldsymbol{\beta}_{0s1}&=&(\boldsymbol{X}_1^\top(\boldsymbol{\alpha})\boldsymbol{X}_1(\boldsymbol{\alpha}))^{-1}\boldsymbol{X}_1^\top(\boldsymbol{\alpha})\boldsymbol{W}(\boldsymbol{\alpha})-\boldsymbol{\beta}_{0s1}\nonumber\\
    &=&(\boldsymbol{X}_1^\top(\boldsymbol{\alpha})\boldsymbol{X}_1(\boldsymbol{\alpha}))^{-1}\boldsymbol{X}_1^\top(\boldsymbol{\alpha})\boldsymbol{W}(\boldsymbol{\alpha})\nonumber\\
    &-&(\boldsymbol{X}_1^\top(\boldsymbol{\alpha})\boldsymbol{X}_1(\boldsymbol{\alpha}))^{-1}\boldsymbol{X}_1^\top(\boldsymbol{\alpha})\boldsymbol{X}_1(\boldsymbol{\alpha})\boldsymbol{\beta}_{0s1}\nonumber\\
    &=&(\boldsymbol{X}_1^\top(\boldsymbol{\alpha})\boldsymbol{X}_1(\boldsymbol{\alpha}))^{-1}\boldsymbol{X}_1^\top(\boldsymbol{\alpha})\nonumber\\
    &&[\boldsymbol{W}(\alpha)-\boldsymbol{X}_1(\boldsymbol{\alpha})\boldsymbol{\beta}_{0s1}].
\end{eqnarray}
\endgroup
Since $\boldsymbol{\beta}_{0s2}=0$,  we obtain
\begin{eqnarray*}
\boldsymbol{X}(\boldsymbol{\alpha})\boldsymbol{\beta}_0=(\boldsymbol{X}_1(\boldsymbol{\alpha})\boldsymbol{X}_2(\boldsymbol{\alpha}))
\begin{pmatrix}
\boldsymbol{\beta}_{0s1}\\\boldsymbol{\beta}_{0s2}
\end{pmatrix}=\boldsymbol{X}_1(\boldsymbol{\alpha})\boldsymbol{\beta}_{0s1}
\end{eqnarray*}
and 
\begin{eqnarray*}
    \widehat{\boldsymbol{b}}(\boldsymbol{\alpha})&=&\boldsymbol{\Omega}_n^{-1}(\boldsymbol{\alpha})\boldsymbol{v}_n^{(1)}(\boldsymbol{\alpha})\nonumber\\
    &=&(\boldsymbol{X}^\top(\boldsymbol{\alpha})\boldsymbol{X}(\boldsymbol{\alpha}))^{-1}\boldsymbol{X}^\top(\boldsymbol{\alpha})\boldsymbol{W}(\boldsymbol{\alpha})\nonumber\\
    &=&\boldsymbol{X}^{-1}(\boldsymbol{\alpha})\boldsymbol{W}(\boldsymbol{\alpha}).
\end{eqnarray*}
Then, from (\ref{F2}), we have 
\begin{equation}
\begin{split}
\label{F3}
    (\boldsymbol{\Omega}_n^{(1)}(\boldsymbol{\alpha}))^{-1}\boldsymbol{v}_n^{(1)}(\boldsymbol{\alpha})-\boldsymbol{\beta}_{0s1} & =(\boldsymbol{X}_1^\top(\boldsymbol{\alpha})\boldsymbol{X}_1(\boldsymbol{\alpha}))^{-1}\boldsymbol{X}_1^\top(\boldsymbol{\alpha})\boldsymbol{X}(\boldsymbol{\alpha})[\boldsymbol{X}^{-1}(\boldsymbol{\alpha})\boldsymbol{W}(\boldsymbol{\alpha})-\boldsymbol{\beta}_0]\\
    & =(\boldsymbol{X}_1^\top(\boldsymbol{\alpha})\boldsymbol{X}_1(\boldsymbol{\alpha}))^{-1}\boldsymbol{X}_1^\top(\boldsymbol{\alpha})\boldsymbol{X}(\boldsymbol{\alpha})(\widehat{\boldsymbol{b}}(\boldsymbol{\alpha})-\boldsymbol{\beta}_0).
\end{split}
\end{equation}
From (\ref{F3}), we obtain
\begin{eqnarray}
    \label{F4}
    \left\lVert(\boldsymbol{\Omega}^{(1)}(\boldsymbol{\alpha}))^{-1}\boldsymbol{v}_n^{(1)}(\boldsymbol{\alpha})-\boldsymbol{\beta}_{0s1}\right\lVert&\leq& \left\lVert(\boldsymbol{X}_1^\top(\boldsymbol{\alpha}
    )\boldsymbol{X}_1(\boldsymbol{\alpha}))^{-1}\right\lVert.\left\lVert(\boldsymbol{X}_1^\top(\boldsymbol{\alpha}
    )\boldsymbol{X}(\boldsymbol{\alpha})\right\rVert\nonumber\\
    &&\left\lVert\widehat{\boldsymbol{b}}(\boldsymbol{\alpha})-\boldsymbol{\beta}_0\right\rVert.   
\end{eqnarray}
Since $\boldsymbol{X}^\top(\boldsymbol{\alpha}
    )\boldsymbol{X}(\boldsymbol{\alpha})=(\boldsymbol{X}_1^\top(\boldsymbol{\alpha}
    )\boldsymbol{X}_2(\boldsymbol{\alpha}))^\top\boldsymbol{X}(\boldsymbol{\alpha})=\begin{pmatrix}
        \boldsymbol{X}_1^\top(\boldsymbol{\alpha}
    )\boldsymbol{X}(\boldsymbol{\alpha})\\
    \boldsymbol{X}_2^\top(\boldsymbol{\alpha}
    )\boldsymbol{X}(\boldsymbol{\alpha})
    \end{pmatrix}$,
    we have 
    \begin{eqnarray*}
        \left\lVert\boldsymbol{X}_1^\top(\boldsymbol{\alpha}
    )\boldsymbol{X}(\boldsymbol{\alpha}) \right\rVert\leq \left\lVert\boldsymbol{X}^\top(\boldsymbol{\alpha}
    )\boldsymbol{X}(\boldsymbol{\alpha})\right\rVert=\left\lVert\boldsymbol{\Omega}_n(\boldsymbol{\alpha})\right\rVert.
    \end{eqnarray*}
Noticing $\boldsymbol{\Omega}_n^{(1)}(\boldsymbol{\alpha})=\boldsymbol{X}_1^\top(\boldsymbol{\alpha}
    )\boldsymbol{X}_1(\boldsymbol{\alpha})$, from (\ref{F4}), we have 
\begin{equation*}
    \begin{split}
        \sup_{\boldsymbol{\alpha}\in H_{n1}}\left\lVert\boldsymbol{\Omega}^{(1)}(\boldsymbol{\alpha})\boldsymbol{v}^{(1)}_n(\boldsymbol{\alpha})
        -\boldsymbol{\beta}_{0s1}
        \right\rVert & \leq\sup_{\boldsymbol{\alpha}\in H_{n1}}\left[{\left\lVert\left(\frac{\boldsymbol{X}_1^\top(\boldsymbol{\alpha})
    \boldsymbol{X}_1(\boldsymbol{\alpha})}{n}\right)^{-1}\right\rVert\left\lVert \frac{\boldsymbol{X}^\top(\boldsymbol{\alpha}
    )\boldsymbol{X}(\boldsymbol{\alpha})}{n}\right\rVert  \left\lVert \widehat{\boldsymbol{b}}(\boldsymbol{\alpha})-\boldsymbol{\beta}_{0s1}\right\rVert}\right]\\
    & \leq \sup_{\boldsymbol{\alpha}\in H_{n1}}\left\lVert \left(\frac{\boldsymbol{\Omega}_n^{(1)}(\boldsymbol{\alpha})}{n}\right)^{-1} \right\rVert \sup_{\boldsymbol{\alpha}\in H_{n1}}\left\lVert \frac{\boldsymbol{\Omega}_n(\boldsymbol{\alpha})}{n} \right\rVert \sup_{\boldsymbol{\alpha}\in H_{n1}}\left\lVert \widehat{\boldsymbol{b}}(\boldsymbol{\alpha})-\boldsymbol{
    \beta
    }_{0}\right\rVert\\
    & = \sup_{\boldsymbol{\alpha}\in H_{n1}}\left[ \lambda_{\text{max}} \left\{ \left(\frac{\boldsymbol{\Omega}_n^{(1)}(\boldsymbol{\alpha})}{n} \right)^{-1}\right\} \right]\sup_{\boldsymbol{\alpha}\in H_{n1}}\left[ \lambda_{\text{max}}\left\{\frac{\boldsymbol{\Omega}_n(\boldsymbol{\alpha})}{n} \right\} \right]\\
    & \sup_{\boldsymbol{\alpha}\in H_{n1}}\left[ \left\lVert \widehat{\boldsymbol{b}}(\boldsymbol{\alpha})-\boldsymbol{\beta}_{0}\right\rVert \right]\\
    & =\sup_{\boldsymbol{\alpha}\in H_{n1}}\left[ \left\{\lambda_{\text{min}}  \left(\frac{\boldsymbol{\Omega}_n^{(1)}(\boldsymbol{\alpha})}{n} \right)\right\}^{-1} \right]\sup_{\boldsymbol{\alpha}\in H_{n1}}\left[ \lambda_{\text{max}}\left\{\frac{\boldsymbol{\Omega}_n(\boldsymbol{\alpha})}{n} \right\} \right]\\
    & \sup_{\boldsymbol{\alpha}\in H_{n1}}\left[ \left\lVert \widehat{\boldsymbol{b}}(\boldsymbol{\alpha})-\boldsymbol{\beta}_{0}\right\rVert \right].
    \end{split}
\end{equation*}
Then, by Condition (C5), we have
\begin{equation*}
\begin{split}
    & \sup_{\boldsymbol{\alpha}\in H_{n1}}\left[ \left\{\lambda_{\text{min}}  \left(\frac{\boldsymbol{\Omega}_n^{(1)}(\boldsymbol{\alpha})}{n} \right)\right\}^{-1} \right]\sup_{\boldsymbol{\alpha}\in H_{n1}}\left[ \lambda_{\text{max}}\left\{\frac{\boldsymbol{\Omega}_n(\boldsymbol{\alpha})}{n} \right\} \right]\cdot \sup_{\boldsymbol{\alpha}\in H_{n1}}\left[ \left\lVert \widehat{\boldsymbol{b}}(\boldsymbol{\alpha})-\boldsymbol{\beta}_{0}\right\rVert \right]\\
    & \leq \left[\frac{1}{c_0}\right]^{-1}\cdot c_o \cdot \sup_{\boldsymbol{\alpha}\in H_{n1}}\left\lVert\widehat{\boldsymbol{b}}(\boldsymbol{\alpha})-\boldsymbol{\beta}_{0}
    \right\rVert\\
    & =c_0^2 \cdot\sup_{\boldsymbol{\alpha}\in H_{n1}}\left\lVert 
    \widehat{\boldsymbol{b}}(\boldsymbol{\alpha})-\boldsymbol{\beta}_{0} \right\rVert.
    \end{split}
\end{equation*}
By Lemma~\ref{lemma1} (i), i.e.,  $\sup_{\boldsymbol{\alpha}\in H_{n}}\left\lVert\widehat{\boldsymbol{b}}(\boldsymbol{\alpha})-\boldsymbol{\beta}_{0} \right\rVert=O_p(\sqrt{p_n/n})$,  we have
    \begin{eqnarray*}
          \sup_{\boldsymbol{\alpha}\in H_{n1}}\left\lVert (\boldsymbol{\Omega}_n^{(1)}(\boldsymbol{\alpha}))^{-1}\boldsymbol{v}_n^{(1)}(\boldsymbol{\alpha}) -\boldsymbol{\beta}_{0s1} \right\rVert=O_p(\sqrt{p_n/n}).
    \end{eqnarray*}
Therefore, from (\ref{F1}), we obtain
\begin{eqnarray}
    \label{F5}
    \sup_{\boldsymbol{\alpha}\in H_{n1}}\left\lVert f(\boldsymbol{\alpha})-\boldsymbol{\beta}_{0s1}+\lambda_n(\boldsymbol{\Omega}_n^{(1)}(\boldsymbol{\alpha}))^{-1}\boldsymbol{D}_1(\boldsymbol{\alpha})f(\boldsymbol{\alpha})\right\rVert=O_p(\sqrt{p_n/n}).
\end{eqnarray}
Next, we want to show 
\begin{eqnarray}
    \label{F6}
    \sup_{\boldsymbol{\alpha}\in H_{n1}}\left\lVert \lambda_n(\boldsymbol{\Omega}_n^{(1)}(\boldsymbol{\alpha}))^{-1}\boldsymbol{D}_1(\boldsymbol{\alpha})f(\boldsymbol{\alpha})\right\rVert=o_p(\sqrt{q_n/n}).
\end{eqnarray}
Then, from (\ref{F5}) and (\ref{F6}), it follows that
\begin{eqnarray*}
    \sup_{\boldsymbol{\alpha}\in H_{n1}}\left\lVert f(\boldsymbol{\alpha})-\boldsymbol{\beta}_{0s1}\right\rVert=O_p(\sqrt{p_n/n})\longrightarrow 0,
\end{eqnarray*}
    which implies, with probability tending to 1, that $f(\boldsymbol{\alpha})\in H_{n1}$, i.e., $f(\boldsymbol{\alpha})$ is a mapping from $H_{n1}$ to itself.

In order to prove (\ref{F6}), first, we rewrite it as
\begin{eqnarray*}
\sup_{\boldsymbol{\alpha}\in H_{n1}}\left\lVert \frac{\lambda_n}{n}(n^{-1}\boldsymbol{\Omega}_n^{(1)}(\boldsymbol{\alpha}))^{-1}\boldsymbol{D}_1(\boldsymbol{\alpha})f(\boldsymbol{\alpha})\right\rVert=o_p(\sqrt{q_n/n}).
\end{eqnarray*}
Since $\widehat{\boldsymbol{b}}(\boldsymbol{\alpha})=\boldsymbol{X}^{-1}(\boldsymbol{\alpha})\boldsymbol{W}(\boldsymbol{\alpha})$, $\boldsymbol{D}_1(\boldsymbol{\alpha})=\text{diag}(\alpha_1^{-2},\ldots,\alpha_{q_n}^{-2})$,
\begin{equation*}
\begin{split}
    \boldsymbol{v}_n^{(1)}(\boldsymbol{\alpha})& =\boldsymbol{X}_1^\top(\boldsymbol{\alpha})\boldsymbol{W}(\boldsymbol{\alpha})=\boldsymbol{X}_1^\top(\boldsymbol{\alpha})\boldsymbol{X}(\boldsymbol{\alpha})\left[ \boldsymbol{X}^{-1}(\boldsymbol{\alpha})\boldsymbol{W}(\boldsymbol{\alpha}) \right]\\
    & =\boldsymbol{X}_1^\top(\boldsymbol{\alpha})\boldsymbol{X}(\boldsymbol{\alpha})\widehat{\boldsymbol{b}}(\boldsymbol{\alpha}).
    \end{split}
\end{equation*}
As shown before, we have
\begin{equation*}
\begin{split}
    \left\lVert \widehat{\boldsymbol{b}}(\boldsymbol{\alpha})\right\rVert& = \left\lVert \widehat{\boldsymbol{b}}(\boldsymbol{\alpha})-\boldsymbol{\beta}_0+ \boldsymbol{\beta}_0\right\rVert\leq \left\lVert \widehat{\boldsymbol{b}}(\boldsymbol{\alpha})-\boldsymbol{\beta}_0 \right\rVert + \left\lVert \boldsymbol{\beta}_0 \right\rVert\\
    & = o_p\left(\sqrt{p_n/n}\right)+O_p(q_n)\\
    & =O_p(q_n)
    \end{split}
\end{equation*}
and 
\begin{eqnarray}
\label{F7}
    \left\lVert \boldsymbol{v}_n^{(1)}(\boldsymbol{\alpha})\right\rVert & \leq& \left\lVert \boldsymbol{X}_1^\top(\boldsymbol{\alpha})\boldsymbol{X}(\boldsymbol{\alpha})\right\rVert \left\lVert \widehat{\boldsymbol{b}}(\boldsymbol{\alpha})\right\rVert\nonumber \\
    &\leq& \left\lVert \boldsymbol{X}^\top(\boldsymbol{\alpha})\boldsymbol{X}(\boldsymbol{\alpha}) \right\rVert\left\lVert\widehat{\boldsymbol{b}}(\boldsymbol{\alpha})\right\rVert \nonumber\\
    &=&n\left\lVert \frac{\boldsymbol{\Omega}_n(\boldsymbol{\alpha})}{n}\right\rVert
    \left\lVert \widehat{\boldsymbol{b}}(\boldsymbol{\alpha})\right\rVert\nonumber \\
    &\leq& c_0 \cdot n \left\lVert \widehat{\boldsymbol{b}}(\boldsymbol{\alpha})\right\rVert\nonumber \\
    &\leq& c_0\cdot n\cdot O_p(q_n)\cdot \qquad \text{(by (C5))}
\end{eqnarray}
Then
\begin{eqnarray}
\label{F8}
    \left\lVert f(\boldsymbol{\alpha})\right\rVert &=& \left\lVert \left(\boldsymbol{\Omega}_n^{(1)}(\boldsymbol{\alpha})+\lambda_n\boldsymbol{D}_1(\boldsymbol{\alpha})\right)^{-1}\boldsymbol{v}_n^{(1)}(\boldsymbol{\alpha})   \right\rVert\nonumber\\
    &\leq& \frac{1}{n}\left\lVert \left(\frac{\boldsymbol{\Omega}_n^{(1)}(\boldsymbol{\alpha})}{n} +\frac{\lambda_n}{n}\boldsymbol{D}_1(\boldsymbol{\alpha})\right)^{-1}\right\rVert\left\lVert \boldsymbol{v}_n^{(1)}(\boldsymbol{\alpha}) \right\rVert\nonumber\\
    &=&\frac{1}{n}\lambda_{\text{max}}\left[ \left(\frac{\boldsymbol{\Omega}_n^{(1)}(\boldsymbol{\alpha})}{n} +\frac{\lambda_n}{n}\boldsymbol{D}_1(\boldsymbol{\alpha})\right)^{-1}\right]\left\lVert \boldsymbol{v}_n^{(1)}(\boldsymbol{\alpha}) \right\rVert\nonumber\\
    &=&\frac{1}{n}\left[\lambda_{\text{min}} \left(\frac{\boldsymbol{\Omega}_n^{(1)}(\boldsymbol{\alpha})}{n} +\frac{\lambda_n}{n}\boldsymbol{D}_1(\boldsymbol{\alpha})\right)\right]^{-1}\left\lVert \boldsymbol{v}_n^{(1)}(\boldsymbol{\alpha}) \right\rVert\nonumber\\
    &\leq& \frac{1}{n}\left[\lambda_{\text{min}}\left(\frac{\boldsymbol{\Omega}_n^{(1)}(\boldsymbol{\alpha})}{n}\right)\right]^{-1}\left\lVert \boldsymbol{v}_n^{(1)}(\boldsymbol{\alpha}) \right\rVert \qquad \text{(since $\frac{\lambda_n}{n}\boldsymbol{D}_1(\boldsymbol{\alpha})$ is positive definite)}\nonumber\\
    &\leq& \frac{1}{n(1/c_0)}\left\lVert \boldsymbol{v}_n^{(1)}(\boldsymbol{\alpha}) \right\rVert\qquad\text{(by (C5))}\nonumber\\
    &\leq& \frac{1}{n}\cdot c_0^2\cdot n\cdot O_p(q_n)\;\; \qquad\text{(by (\ref{F7}))} \nonumber\\
    &=&c_0^2 \cdot O_p(q_n).
\end{eqnarray}
Since $\boldsymbol{\alpha}\in H_{n1}$, by (C7), when $n$ is large enough, $|\alpha_j|\geq a_0/2,~ 1\leq j\leq q_n$, then
\begin{eqnarray}
    \label{F9}
    \left \lVert \boldsymbol{D}_1(\boldsymbol{\alpha})\right\rVert =\lambda_{\text{max}}(\boldsymbol{D}_1(\boldsymbol{\alpha}))=\max_{1\leq j\leq q_n}(\alpha_j^{-2})\leq (a_0/2)^{-2}=4a_0^{-2}.
\end{eqnarray}
Thus, by (\ref{F6}), (\ref{F7}), and (\ref{F8}), we have 
\begingroup
\allowdisplaybreaks
\begin{eqnarray*}
    \left\lVert \frac{\lambda_n}{n}\left( n^{-1}\boldsymbol{\Omega}_n^{(1)}(\boldsymbol{\alpha})\right)^{-1}\boldsymbol{D}_1(\boldsymbol{\alpha})f(\boldsymbol{\alpha}) \right \rVert&\leq&\frac{\lambda_n}{n}\left\lVert \left(n^{-1}\boldsymbol{\Omega}_n^{(1)}(\boldsymbol{\alpha})\right)^{-1} \right\rVert \left\lVert \boldsymbol{D}_1(\boldsymbol{\alpha})\right\rVert
    \left\lVert f(\boldsymbol{\alpha})\right\rVert\\
    &\leq& \frac{\lambda_n}{n}\left(\frac{1}{1/c_0}\right)(4a_0^{-2})\cdot c_0^2\cdot O_p(q_n)\\
    &=&(4c_0^3a_0^{-2})\cdot O_p\left(\frac{\lambda_n\sqrt{q_n}}{\sqrt{n}}\sqrt{\frac{q_n}{n}}\right)\\
    &=&\left(4c_0^3a_0^{-2} \right)\\
    && \times o_p\left( \sqrt{\frac{q_n}{n}}\right)\qquad\text{(since by (C6), $\frac{\lambda_n\sqrt{q_n}}{n}\longrightarrow0$)}\\
    &=&o_p\left(\sqrt{\frac{q_n}{n}} \right).
\end{eqnarray*}
\endgroup
Thus, 
\begin{eqnarray*}
    \sup_{\boldsymbol{\alpha}\in H_{n1}}\left \lVert \lambda_n(\boldsymbol{\Omega}_n^{(1)}(\boldsymbol{\alpha}))^{-1}\boldsymbol{D}_1(\boldsymbol{\alpha})f(\boldsymbol{\alpha})\right\rVert=o_p\left(\sqrt{\frac{q_n}{n}}\right),
\end{eqnarray*}
i.e.,  (\ref{F6}) holds.

Recall that $$\boldsymbol{\Omega}_n(\boldsymbol{\alpha})=\boldsymbol{\Omega}_n(\boldsymbol{\beta})\Big|_{\boldsymbol{\beta}_{s1}=\boldsymbol{\alpha},~\boldsymbol{\beta}_{s2}=0},\;\;
\boldsymbol{v}_n(\boldsymbol{\alpha})=\boldsymbol{v}_n(\boldsymbol{\beta})\Big|_{\boldsymbol{\beta}_{s1}=\boldsymbol{\alpha},~\boldsymbol{\beta}_{s2}=0},$$
$$\boldsymbol{\Omega}_n^{(1)}(\boldsymbol{\alpha})=\boldsymbol{\Omega}^{(1)}_n(\boldsymbol{\beta})\Big|_{\boldsymbol{\beta}_{s1}=\boldsymbol{\alpha},~\boldsymbol{\beta}_{s2}=0},\;\; \boldsymbol{v}_n^{(1)}(\boldsymbol{\alpha})=\boldsymbol{v}^{(1)}_n(\boldsymbol{\beta})\Big|_{\boldsymbol{\beta}_{s1}=\boldsymbol{\alpha},~\boldsymbol{\beta}_{s2}=0},
$$ 
and 
$$\sup_{\boldsymbol{\alpha}\in H_{n1}}\left\lVert f(\boldsymbol{\alpha})-\boldsymbol{\beta}_{0s1}\right\rVert=O_p\left(\sqrt{\frac{p_n}{n}} \right)
 ,
 $$ 
 which implies that with probability tending to 1, $f(\boldsymbol{\alpha})$ is a mapping from $H_{n1}$ to itself.

Multiplying $\boldsymbol{\Omega}_n^{(1)}(\boldsymbol{\alpha})+\lambda_n\boldsymbol{D}_1(\boldsymbol{\alpha})$ on both sides of (\ref{A10}), we obtain
\begin{eqnarray}
    \label{G1}\left(\boldsymbol{\Omega}_n^{(1)}(\boldsymbol{\alpha}) +\lambda_n\boldsymbol{D}_1(\boldsymbol{\alpha})\right)f(\boldsymbol{\alpha})=\boldsymbol{v}_n^{(1)}(\boldsymbol{\alpha}).
\end{eqnarray}
Denote the $j$th row of $\boldsymbol{\Omega}_n^{(1)}(\boldsymbol{\alpha})$ by $\boldsymbol{\omega}_j^\top(\boldsymbol{\alpha})$ and the $j$th row of $\boldsymbol{D}_1(\boldsymbol{\alpha})$ by $\boldsymbol{d}_j^\top(\boldsymbol{\alpha})$. Then, 
\begin{eqnarray*}
    \boldsymbol{m}_j^\top(\boldsymbol{\alpha})=\left( \frac{\partial^2 [\sum_{i=1}^n \log f_n(v_{ni},(\boldsymbol{\alpha}^\top, \boldsymbol{0}^\top),\boldsymbol{\Lambda})]}{\partial \alpha_j \partial\alpha_1},\ldots,\frac{\partial^2 [\sum_{i=1}^n \log f_n(v_{ni},(\boldsymbol{\alpha}^\top, \boldsymbol{0}^\top),\boldsymbol{\Lambda})]}{\partial \alpha_j \partial\alpha_{q_n}}\right),
\end{eqnarray*}
where $\boldsymbol{d}_j^\top=(0,\ldots,0,\alpha_j^{-2},\ldots,0)$.
We take derivatives on both sides of (\ref{G1}) and have 
\begin{eqnarray}
    \label{G2}
    \frac{\partial}{\partial\boldsymbol{\alpha}^\top}\left[ (\boldsymbol{\Omega}_n^{(1)}(\boldsymbol{\alpha})+\lambda_n \boldsymbol{D}_1(\boldsymbol{\alpha}))f(\boldsymbol{\alpha}) \right]=\frac{\partial}{\partial\boldsymbol{\alpha}^\top}[\boldsymbol{v}_n^{(1)}(\boldsymbol{\alpha})].
\end{eqnarray}
Since 
\begingroup
\allowdisplaybreaks
\begin{eqnarray*}
\boldsymbol{v}_n(\boldsymbol{\alpha})&=&\dot{\ell}_n(\boldsymbol{\alpha}|\widetilde{\boldsymbol{\Lambda}})+\boldsymbol{\Omega}_n(\boldsymbol{\alpha})\begin{pmatrix}
        \boldsymbol{\alpha}\\
        0
    \end{pmatrix}\\
    &=&\dot{\ell}_n(\boldsymbol{\alpha}|\widetilde{\boldsymbol{\Lambda}})+\begin{pmatrix}
        \boldsymbol{\Omega}_n^{(1)}(\boldsymbol{\alpha})&\boldsymbol{\Omega}_n^{(12)}(\boldsymbol{\alpha})\\
        \boldsymbol{\Omega}_n^{(21)}(\boldsymbol{\alpha})&
        \boldsymbol{\Omega}_n^{(2)}(\boldsymbol{\alpha})
    \end{pmatrix}\begin{pmatrix}
        \boldsymbol{\alpha}\\
        0
    \end{pmatrix}\\
    &=&\dot{\ell}_n(\boldsymbol{\alpha}|\widetilde{\boldsymbol{\Lambda}})+\begin{pmatrix}
        \boldsymbol{\Omega}_n^{(1)}(\boldsymbol{\alpha})\boldsymbol{\alpha}\\
        \boldsymbol{\Omega}_n^{(21)}(\boldsymbol{\alpha})\boldsymbol{\alpha}
    \end{pmatrix},
\end{eqnarray*}
\endgroup
then, $\boldsymbol{v}_n^{(1)}(\boldsymbol{\alpha})=\dot{\ell}_n^{(1)}(\boldsymbol{\alpha}|\widetilde{\boldsymbol{\Lambda}})+\boldsymbol{\Omega}_n^{(1)}(\boldsymbol{\alpha})\boldsymbol{\alpha}$, and by Lemma~\ref{lemma2}, we have
\begin{eqnarray}
\label{G3}
    \frac{\partial \boldsymbol{v}_n^{(1)}(\boldsymbol{\alpha})}{\partial \boldsymbol{\alpha}^\top}&=&-\boldsymbol{\Omega}_n^{(1)}(\boldsymbol{\alpha})+\boldsymbol{\Omega}_n^{(1)}(\boldsymbol{\alpha})\textbf{I}_{q_n}+
    \begin{pmatrix}
        \boldsymbol{\alpha}^\top &\ldots&0\\
        \vdots&\ddots&\vdots\\
        0&\ldots&\boldsymbol{\alpha}^\top
    \end{pmatrix}
    \begin{pmatrix}
        \left( \frac{\partial \omega_{1}^\top(\boldsymbol{\alpha})}{\partial \boldsymbol{\alpha}} \right)^\top\\
        \vdots\\
        \left( \frac{\partial \omega_{q_{n}}^\top(\boldsymbol{\alpha})}{\partial \boldsymbol{\alpha}} \right)^\top
    \end{pmatrix}\nonumber\\
    &=&\begin{pmatrix}
        \boldsymbol{\alpha}^\top &\ldots&0\\
        \vdots&\ddots&\vdots\\
        0&\ldots&\boldsymbol{\alpha}^\top
    \end{pmatrix}
    \begin{pmatrix}
        \left( \frac{\partial \omega_{1}^\top(\boldsymbol{\alpha})}{\partial \boldsymbol{\alpha}} \right)^\top\\
        \vdots\\
        \left( \frac{\partial \omega_{q_{n}}^\top(\boldsymbol{\alpha})}{\partial \boldsymbol{\alpha}} \right)^\top
    \end{pmatrix}.
\end{eqnarray}
By applying Lemma~\ref{lemma2} to the left-hand-side of (\ref{G2}), we obtain
\begingroup
\allowdisplaybreaks
\begin{eqnarray}
    \label{G4}
    \frac{\partial}{\partial\boldsymbol{\alpha}^\top}\left[ (\boldsymbol{\Omega}_n^{(1)}(\boldsymbol{\alpha})+\lambda_n \boldsymbol{D}_1(\boldsymbol{\alpha}))f(\boldsymbol{\alpha}) \right]&=&(\boldsymbol{\Omega}_n^{(1)}(\boldsymbol{\alpha})+\lambda_n \boldsymbol{D}_1(\boldsymbol{\alpha}))\frac{\partial}{\partial\boldsymbol{\alpha}^\top}f(\boldsymbol{\alpha})\nonumber\\
    &+&\begin{pmatrix}
        f^\top(\boldsymbol{\alpha})&\ldots&0\\
        \vdots&\ddots&\vdots\\
        0&\ldots&f^\top(\boldsymbol{\alpha})
    \end{pmatrix}\\
    &&\left[\begin{pmatrix}
        \left( \frac{\partial \boldsymbol{\omega}_1^\top (\boldsymbol{\alpha})}{\partial \boldsymbol{\alpha}}\right)^\top\\
        \vdots\\
        \left( \frac{\partial \boldsymbol{\omega}_{q_n}^\top (\boldsymbol{\alpha})}{\partial \boldsymbol{\alpha}}\right)^\top
    \end{pmatrix}+\lambda_n\begin{pmatrix}
        \left( \frac{\partial \boldsymbol{d}_1^\top (\boldsymbol{\alpha})}{\partial \boldsymbol{\alpha}}\right)^\top\\
        \vdots\\
        \left( \frac{\partial \boldsymbol{d}_{q_n}^\top (\boldsymbol{\alpha})}{\partial \boldsymbol{\alpha}}\right)^\top
    \end{pmatrix}\right]\nonumber.
\end{eqnarray}
\endgroup
Since 
\begin{eqnarray*}
    \frac{\partial \boldsymbol{d}_j^\top}{\partial \boldsymbol{\alpha}}=\begin{pmatrix}
        0&\ldots&0&0&0&\ldots&0\\
        &\vdots&&\vdots&&\vdots&\\
        0&\ldots&0&-2\alpha_j^{-3}&0&\ldots&0\\
        &\vdots&&\vdots&&\vdots&\\0&\ldots&0&0&0&\ldots&0
    \end{pmatrix},
\end{eqnarray*}
\begin{eqnarray*}
    f^\top(\boldsymbol{\alpha})\left( \frac{\partial \boldsymbol{d}_j^\top}{\partial\boldsymbol{\alpha}} \right)^\top=(0,\ldots,0,-2f_j(\boldsymbol{\alpha})\boldsymbol{\alpha}_j^{-3},0,\ldots,0),
\end{eqnarray*}
then we have
\begin{eqnarray*}
    \begin{pmatrix}
        f^\top(\boldsymbol{\alpha})&\ldots&0\\
        \vdots&\ddots&\vdots\\
        0&\ldots&f^\top(\boldsymbol{\alpha})
    \end{pmatrix}\begin{pmatrix}
        \left( \frac{\partial \boldsymbol{d}_1^\top (\boldsymbol{\alpha})}{\partial \boldsymbol{\alpha}}\right)^\top\\
        \vdots\\
        \left( \frac{\partial \boldsymbol{d}_{q_n}^\top (\boldsymbol{\alpha})}{\partial \boldsymbol{\alpha}}\right)^\top
    \end{pmatrix}=\text{diag}(-2f_1(\boldsymbol{\alpha})\boldsymbol{\alpha}_1^{-3},\ldots,-2f_{q_n}(\boldsymbol{\alpha})\boldsymbol{\alpha}_{q_n}^{-3}).
\end{eqnarray*}
By (\ref{G3}) and (\ref{G4}), (\ref{G2}) becomes
\begin{eqnarray*}
\label{G5}
    &&\left(\boldsymbol{\Omega}_n^{(1)}(\boldsymbol{\alpha})+\lambda_n \boldsymbol{D}_1(\boldsymbol{\alpha}) \right)\frac{\partial}{\partial\boldsymbol{\alpha}^\top}f(\boldsymbol{\alpha})+\lambda_n\text{diag}(-2f_1(\boldsymbol{\alpha})\boldsymbol{\alpha}_1^{-3},\ldots,-2f_{q_n}(\boldsymbol{\alpha})\boldsymbol{\alpha}_{q_n}^{-3})\nonumber\\
    &&+\begin{pmatrix}
        (f(\boldsymbol{\alpha})-\boldsymbol{\alpha})^\top&\ldots&0\\
        \vdots&\ddots&\vdots\\
        0&\ldots&(f(\boldsymbol{\alpha})-\boldsymbol{\alpha})^\top
    \end{pmatrix}\begin{pmatrix}
        \left( \frac{\partial \boldsymbol{\omega}_1^\top (\boldsymbol{\alpha})}{\partial \boldsymbol{\alpha}}\right)^\top\\
        \vdots\\
        \left( \frac{\partial \boldsymbol{\omega}_{q_n}^\top (\boldsymbol{\alpha})}{\partial \boldsymbol{\alpha}}\right)^\top
    \end{pmatrix}=0.
\end{eqnarray*}
Denote $\dot{f}(\boldsymbol{\alpha})=\frac{\partial f(\boldsymbol{\alpha})}{\partial \boldsymbol{\alpha}^\top}$ (which is a $q_n\times q_n$ matrix) and
\begin{eqnarray*}
    \begin{pmatrix}
        (f(\boldsymbol{\alpha})-\boldsymbol{\alpha})^\top&\ldots&0\\
        \vdots&\ddots&\vdots\\
        0&\ldots&(f(\boldsymbol{\alpha})-\boldsymbol{\alpha})^\top
    \end{pmatrix}\left( \frac{\partial\boldsymbol{\Omega}_n^{(1)}(\boldsymbol{\alpha})}{\partial \boldsymbol{\alpha}}\right)^\top=\boldsymbol{F}_n(\boldsymbol{\alpha})\boldsymbol{P}_n(\boldsymbol{\alpha}).
\end{eqnarray*}
Then, we have
\begin{eqnarray*}
    &&\left(\boldsymbol{\Omega}_n^{(1)}(\boldsymbol{\alpha})+\lambda_n \boldsymbol{D}_1(\boldsymbol{\alpha}) \right)\dot{f}(\boldsymbol{\alpha})+\lambda_n\text{diag}(-2f_1(\boldsymbol{\alpha})\alpha_1^{{-3}},\ldots,-2f_{q_n}(\boldsymbol{\alpha})\alpha_{q_n}^{{-3}})\nonumber\\
    &&+\boldsymbol{F}_n(\boldsymbol{\alpha})\boldsymbol{P}_n(\boldsymbol{\alpha})=0,
\end{eqnarray*}
or 
\begin{eqnarray}
\label{G6}
    \left(\boldsymbol{\Omega}_n^{(1)}(\boldsymbol{\alpha})+\lambda_n \boldsymbol{D}_1(\boldsymbol{\alpha}) \right)\dot{f}(\boldsymbol{\alpha})&=&2\lambda_n\text{diag}(f_1(\boldsymbol{\alpha})\alpha_1^{-3},\ldots,f_{q_n}(\boldsymbol{\alpha})\alpha_{q_n}^{-3})\nonumber\\
    &-&\boldsymbol{F}_n(\boldsymbol{\alpha})\boldsymbol{P}_n(\boldsymbol{\alpha}).
\end{eqnarray}
Dividing both sides of (\ref{G6}) by $n$, we have
\begin{eqnarray*}
\left(\frac{\boldsymbol{\Omega}_n^{(1)}(\boldsymbol{\alpha})}{n}+\frac{\lambda_n}{n}\boldsymbol{D}_1(\boldsymbol{\alpha})\right)\dot{f}(\boldsymbol{\alpha})&=&2(\lambda_n/n)\text{diag}(f_1(\boldsymbol{\alpha})\alpha_1^{-3},\ldots,f_{q_n}(\boldsymbol{\alpha})\alpha_{q_n}^{-3})\nonumber\\
&-&\boldsymbol{F}_n(\boldsymbol{\alpha})\boldsymbol{P}_n(\boldsymbol{\alpha})/n,
\end{eqnarray*}
and therefore
\begin{eqnarray}
\label{G7}
    &&\sup_{\boldsymbol{\alpha}\in H_{n1}}\left\lVert \left( \frac{\boldsymbol{\Omega}_n^{(1)}(\boldsymbol{\alpha})}{n}+\frac{\lambda_n}{n}\boldsymbol{D}_1(\boldsymbol{\alpha})\right)\dot{f}(\boldsymbol{\alpha})\right\rVert\nonumber\\
    &&=\sup_{\boldsymbol{\alpha}\in H_{n1}}\left[\frac{2 \lambda_n}{n}\left\lVert \text{diag}(f_1(\boldsymbol{\alpha})\alpha_1^{-3},\ldots,f_{q_n}(\boldsymbol{\alpha})\alpha_{q_n}^{-3})-\frac{\boldsymbol{F}_n(\boldsymbol{\alpha})\boldsymbol{P}_n(\boldsymbol{\alpha})}{n} \right\rVert\right].
\end{eqnarray}
First, we show that the right-hand-side of (\ref{G7}) is $o_p(1)$, which is equivalent to showing 
\begin{eqnarray}
    \label{G8}
    \sup_{\boldsymbol{\alpha}\in H_{n1}}\left[(2\lambda_n/n)\left\lVert \text{diag}(f_1(\boldsymbol{\alpha})\boldsymbol\alpha_1^{-3},\ldots,f_{q_n}(\boldsymbol{\alpha})\alpha_{q_n}^{-3})\right\rVert\right]=o_p(1)
\end{eqnarray}
and 
\begin{eqnarray}
    \label{G9}
    \sup_{\boldsymbol{\alpha}\in H_{n1}} \left\lVert  \frac{
    F_n(\boldsymbol{\alpha})
    P_n(\boldsymbol{\alpha})}{n}\right\rVert=o_p(1).
\end{eqnarray}
To show (\ref{G8}), since 
\begin{eqnarray*}
    \left\lVert \text{diag}(f_1(\boldsymbol{\alpha})\boldsymbol{\alpha}_1^{-3},\ldots,f_{q_n}(\boldsymbol{\alpha})\boldsymbol{\alpha}_{q_n}^{-3})\right\rVert=\max_{1\leq j\leq q_n}\left\{ |f_j(\boldsymbol{\alpha})\boldsymbol{\alpha}_j^{-3}|\right\},
\end{eqnarray*}
by (C7), $a_0\leq |\beta_{0s1,j}|\leq a_1$, $1\leq j \leq q_n$, then, when $\boldsymbol{\alpha}\in H_{n1}$, we have $|\alpha_j-\beta_{0s1,j}|\leq \delta \sqrt{p_n/n}$. 
Thus, when $n$ is large enough, we have
\begin{eqnarray*}
    |\alpha_j|\geq |\beta_{0s1,j}|-\delta\sqrt{p_n/n}\geq |\beta_{0s1,j}|-\frac{1}{2}|\beta_{0s1,j}|=\frac{1}{2}|\beta_{0s1,j}|\geq a_0/2,
\end{eqnarray*}
we obtain $|\alpha_j^{-3}|\leq (a_0/2)^{-3}$.

By 
\begin{eqnarray*}
    \sup_{\boldsymbol{\alpha}\in H_{n1}}\left\lVert f(\boldsymbol{\alpha})-\boldsymbol{\beta}_{0s1} \right\rVert\leq O_p\left(\sqrt{p_n/n} \right),
\end{eqnarray*}
which has been shown before, we have
\begin{eqnarray*}
    \sup_{\boldsymbol{\alpha}\in H_{n1}}\left\lVert f_j(\boldsymbol{\alpha})-\boldsymbol{\beta}_{0s1,j} \right\rVert \leq O_p\left(\sqrt{p_n/n} \right)=o_p(1).
\end{eqnarray*}
Thus, we obtain 
\begin{eqnarray*}
    \sup_{\boldsymbol{\alpha}\in H_{n1}}|f_j(\boldsymbol{\alpha})|\leq |\beta_{0s1,j}|+o_p(1)\leq a_1+o_p(1).
\end{eqnarray*}
Hence
\begin{eqnarray*}
    \sup_{\boldsymbol{\alpha}\in H_{n1}}|f_j(\boldsymbol{\alpha})\alpha_j^{-3}|\leq (a_1+o_p(1))(a_0/2)=O_p(1),
\end{eqnarray*}
and
\begin{eqnarray*}
    \max_{1\leq j\leq q_n}\left\{|f_j(\boldsymbol{\alpha})\alpha_j^{-3}| \right\}=o_p(1).
\end{eqnarray*}
Since $\lambda_n/n\longrightarrow 0$, then
\begin{eqnarray}
\label{extra}
    \sup_{\boldsymbol{\alpha}\in H_{n1}}\left[(2\lambda_n/n)\left\lVert \text{diag}(f_1(\boldsymbol{\alpha})\boldsymbol\alpha_1^{-3},\ldots,f_{q_n}(\boldsymbol{\alpha})\alpha_{q_n}^{-3})\right\rVert\right]\leq (\lambda_n/n)\cdot O_p(1)=o_p(1),
\end{eqnarray}
which implies that (\ref{G8}) holds.

Now, we prove (\ref{G9}). Since $\left\lVert \boldsymbol{F}_n(\boldsymbol{\alpha})\boldsymbol{P}_n(\boldsymbol{\alpha}) \right\rVert \leq  \left\lVert \boldsymbol{F}_n(\boldsymbol{\alpha}) \right\rVert \left\lVert \boldsymbol{P}_n(\boldsymbol{\alpha}) \right\rVert$, one can write
\begin{eqnarray*}
    \boldsymbol{F}_n(\boldsymbol{\alpha})\boldsymbol{F}_n^\top(\boldsymbol{\alpha})=\begin{pmatrix}
        \left\lVert f_n(\boldsymbol{\alpha})-\boldsymbol{\alpha} \right\rVert^2 & \ldots & 0\\
        \vdots&\ddots&\vdots\\
        0&\ldots& \left\lVert f_n(\boldsymbol{\alpha})-\boldsymbol{\alpha} \right\rVert^2
    \end{pmatrix},
\end{eqnarray*}
then $\left\lVert \boldsymbol{F}_n(\boldsymbol{\alpha})\boldsymbol{F}_n^\top(\boldsymbol{\alpha})\right\rVert=\lambda_{\text{max}}(\boldsymbol{F}_n(\boldsymbol{\alpha})\boldsymbol{F}_n^\top(\boldsymbol{\alpha}))=\left\lVert f(\boldsymbol{\alpha})-\boldsymbol{\alpha} \right\rVert^{2}$, thus
\begin{eqnarray*}
    \left\lVert \boldsymbol{F}_n(\boldsymbol{\alpha}) \right\rVert&=&\sqrt{ \left\lVert \boldsymbol{F}_n(\boldsymbol{\alpha}) \boldsymbol{F}_n^\top(\boldsymbol{\alpha})\right\rVert}=\sqrt{\left\lVert \boldsymbol{f}(\boldsymbol{\alpha})-\boldsymbol{\alpha}\right\rVert^2}=\left\lVert \boldsymbol{f}(\boldsymbol{\alpha})-\boldsymbol{\alpha}\right\rVert\nonumber\\
    &\leq& \left\lVert \boldsymbol{f}(\boldsymbol{\alpha})-\boldsymbol{\beta}_{0s1}\right\rVert+\left\lVert \boldsymbol{\alpha}-\boldsymbol{\beta}_{0s1}\right\rVert.
\end{eqnarray*}
Since
\begin{eqnarray*}
    \sup_{\boldsymbol{\alpha}\in H_{n1}}\left\lVert f(\boldsymbol{\alpha})-\boldsymbol{\alpha} \right\rVert&\leq& \sup_{\boldsymbol{\alpha}\in H_{n1}}  \left\lVert f_(\boldsymbol{\alpha})-\boldsymbol{\beta}_{0s1}\right\rVert+\sup_{\boldsymbol{\alpha}\in H_{n1}}\left\lVert \boldsymbol{\alpha}-\boldsymbol{\beta}_{0s1}\right\rVert\\
    &=&O_p\left(\sqrt{p_n/n} \right)+\delta\left(\sqrt{p_n/n} \right)\\
    &=&O_p\left(\sqrt{p_n/n} \right),
\end{eqnarray*}
therefore
\begin{eqnarray}
    \label{G10}
    \sup_{\boldsymbol{\alpha}\in H_{n1}}\left\lVert \boldsymbol{F}_n(\boldsymbol{\alpha}) \right\rVert=O_p\left(\sqrt{p_n/n} \right).
\end{eqnarray}
On the other hand, we have
\begin{eqnarray*}
    \frac{\boldsymbol{P}_n^\top(\boldsymbol{\alpha})\boldsymbol{P}_n(\boldsymbol{\alpha})}{n^2}=\sum_{j=1}^{q_n}\left(\frac{1}{n}\frac{\partial \omega_j^\top(\boldsymbol{\alpha})}{\partial \boldsymbol{\alpha}}\right)\left(\frac{1}{n}\frac{\partial \omega_j^\top(\boldsymbol{\alpha})}{\partial \boldsymbol{\alpha}}\right)^\top.
\end{eqnarray*}
Therefore, we obtain
\begin{eqnarray*}
    \left\lVert \frac{\boldsymbol{P}_n^\top(\boldsymbol{\alpha})\boldsymbol{P}_n(\boldsymbol{\alpha})}{n^2} \right\rVert&\leq& \sum_{j=1}^{q_n}\left\lVert\left(\frac{1}{n}\frac{\partial \omega_j^\top(\boldsymbol{\alpha})}{\partial \boldsymbol{\alpha}}\right)\left(\frac{1}{n}\frac{\partial \omega_j^\top(\boldsymbol{\alpha})}{\partial \boldsymbol{\alpha}}\right)^\top\right\rVert\\
    &=&\sum_{j=1}^{q_n}\lambda_{\text{max}}\left[\left(\frac{1}{n}\frac{\partial \omega_j^\top(\boldsymbol{\alpha})}{\partial \boldsymbol{\alpha}}\right)\left(\frac{1}{n}\frac{\partial \omega_j^\top(\boldsymbol{\alpha})}{\partial \boldsymbol{\alpha}}\right)^\top\right].
\end{eqnarray*}
Since the trace of a symmetric matrix is equal to the sum of its eigenvalues, we obtain 
\begin{eqnarray*}
    \left\lVert \frac{\boldsymbol{P}_n^\top(\boldsymbol{\alpha})\boldsymbol{P}_n(\boldsymbol{\alpha})}{n^2} \right\rVert&\leq&\sum_{j=1}^{q_n}\text{trace}\left[\left(\frac{1}{n}\frac{\partial \omega_j^\top(\boldsymbol{\alpha})}{\partial \boldsymbol{\alpha}}\right)\left(\frac{1}{n}\frac{\partial \omega_j^\top(\boldsymbol{\alpha})}{\partial \boldsymbol{\alpha}}\right)^\top\right]\\
    &=&\sum_{j=1}^{q_n}\sum_{k=1}^{q_n}\sum_{h=1}^{q_n}\left(\frac{1}{n}\frac{\partial \omega _{jk}(\boldsymbol{\alpha})}{\partial\boldsymbol{\alpha}_h} \right)^2.
\end{eqnarray*}
Noticing  
\begin{eqnarray*}
    \omega_{j}^\top(\boldsymbol{\alpha})=\left( \frac{\partial^2 [\sum_{i=1}^n \log f_n(v_{ni},(\boldsymbol{\alpha}^\top, \boldsymbol{0}^\top),\widetilde{\boldsymbol{\Lambda}})]}{\partial \alpha_j \partial\alpha_1},\ldots,\frac{\partial^2 [\sum_{i=1}^n \log f_n(v_{ni},(\boldsymbol{\alpha}^\top, \boldsymbol{0}^\top),\widetilde{\boldsymbol{\Lambda}})]}{\partial \alpha_j \partial\alpha_{q_n}}\right),
\end{eqnarray*}
by Cauchy-Schwarz inequality and condition (C9), we have 
\begin{eqnarray*}
    \left[\frac{1}{n}\frac{\partial \omega_{jk}(\boldsymbol{\alpha})}{\partial \boldsymbol{\alpha}_h} \right]^2&=&\left[\frac{1}{n}\frac{\partial^3[\sum_{i=1}^n \log f_n(v_{ni},(\boldsymbol{\alpha}^\top, \boldsymbol{0}^\top),\widetilde{\boldsymbol{\Lambda}})]}{\partial \alpha_j \partial\alpha_k\partial \alpha_h} \right]^2\\
    &=&\frac{1}{n^2}\left[\sum_{i=1}^n\frac{\partial^3 [ \log f_n(v_{ni},(\boldsymbol{\alpha}^\top, \boldsymbol{0}^\top),\widetilde{\boldsymbol{\Lambda}})]}{\partial \alpha_j \partial\alpha_k\partial \alpha_h} \right]^2\\
    &\leq& \frac{n}{n^2}\sum_{i=1}^n\left[\frac{\partial^3 [\log f_n(v_{ni},(\boldsymbol{\alpha}^\top, \boldsymbol{0}^\top),\widetilde{\boldsymbol{\Lambda}})]}{\partial \alpha_j \partial\alpha_k\partial \alpha_h} \right]^2\\
    &\leq& \frac{1}{n}\sum_{i=1}^n {M^2_n}_{jkh}(v_{ni}).
\end{eqnarray*}
Hence
\begin{eqnarray*}
    \sup_{\boldsymbol{\alpha}\in H_{n1}}\left\lVert \frac{\boldsymbol{P}_n^\top(\boldsymbol{\alpha})\boldsymbol{P}_n(\boldsymbol{\alpha})}{n^2} \right\rVert\leq \frac{1}{n}\sum_{j=1}^{q_n}\sum_{k=1}^{q_n}\sum_{h=1}^{q_n}\sum_{i=1}^{n}{M^2_n}_{jkh}(v_{ni}).
\end{eqnarray*}
Since (C9) indicates $E_{(\boldsymbol{\beta},\boldsymbol{\Lambda})}\left\{{M^2_n}_{jkh}(v_{ni})\right\}<M_d<\infty$, we have
\begin{eqnarray*}
    E_{(\boldsymbol{\beta},\boldsymbol{\Lambda})}\left[\frac{1}{n}\sum_{j=1}^{q_n}\sum_{k=1}^{q_n}\sum_{h=1}^{q_n}{M^2_n}_{jkh}(v_{ni})\right]\leq M_{d}q_{n}^3,
\end{eqnarray*}
which implies $\sum_{j=1}^{q_n}\sum_{k=1}^{q_n}\sum_{h=1}^{q_n}{M^2_n}_{jkh}(v_{ni})/n=O_p(q_n^3)$. 
As a result, we obtain that
\begin{eqnarray}
\label{G11}
    \sup_{\boldsymbol{\alpha}\in H_{n1}}\left\lVert \frac{\boldsymbol{P}_n^\top(\boldsymbol{\alpha})\boldsymbol{P}_n(\boldsymbol{\alpha})}{n^2} \right\rVert=O_p(q_n^3).
\end{eqnarray}
Finally, by (\ref{G10}) and (\ref{G11}), we obtain
\begin{eqnarray*}
    \sup_{\boldsymbol{\alpha}\in H_{n1}}\left\lVert \boldsymbol{F}_n(\boldsymbol{\alpha})\boldsymbol{P}_n(\boldsymbol{\alpha})/n\right\rVert&\leq& O_p\left(\sqrt{p_n/n}q_n^{3/2} \right)=O_p\left(\sqrt{p_n q_n^{3}/n}\right)\\
    &&\leq O_p\left(\sqrt{p_n^2 q_n^{2}/n}\right)=O_p\left(p_n q_n/\sqrt{n}\right).
\end{eqnarray*}
Consequently, by (C6), $p_nq_n/\sqrt{n}\longrightarrow 0$, we have 
\begin{eqnarray*}
    \sup_{\boldsymbol{\alpha}\in H_{n1}}\left\lVert \boldsymbol{F}_n(\boldsymbol{\alpha})\boldsymbol{P}_n(\boldsymbol{\alpha})/n\right\rVert=o_p(1),
\end{eqnarray*}
which means that (\ref{G9}) holds. 

By (\ref{G7}), we have
\begin{eqnarray}
    \label{G12}
    \sup_{\boldsymbol{\alpha}\in H_{n1}}\left\lVert \left( \frac{\boldsymbol{\Omega}_n^{(1)}(\boldsymbol{\alpha})}{n}+\frac{\lambda_n}{n}\boldsymbol{D}_1(\boldsymbol{\alpha})\right)\dot{f}(\boldsymbol{\alpha})\right\rVert=o_p(1).
\end{eqnarray}
Subsequently, we aim to demonstrate that with probability tending to 1,
\begin{eqnarray*}
    \sup_{\boldsymbol{\alpha}\in H_{n1}}\left\lVert \dot{f}(\boldsymbol{\alpha})\right\rVert\longrightarrow 0.
\end{eqnarray*}
Since for any two matrices $\boldsymbol{A}$ and $\boldsymbol{B}$, by the 2-norm properties, we have
\begin{eqnarray*}
    \lambda_{\text{min}}(\boldsymbol{A})\left\lVert \boldsymbol{B}\right\rVert\leq \left\lVert \boldsymbol{AB}\right\rVert\leq \lambda_{\text{max}}(\boldsymbol{A})\left\lVert \boldsymbol{B}\right\rVert.
\end{eqnarray*}
According to (C5), we can conclude that
\begin{eqnarray*}
    \left\lVert\frac{\boldsymbol{\Omega}_n^{(1)}(\boldsymbol{\alpha})}{n}\dot{f}(\boldsymbol{\alpha}) \right\rVert\geq \frac{1}{c_0}\left\lVert \dot{f}(\boldsymbol{\alpha}) \right\rVert.
\end{eqnarray*}

Then by (C7), when $n$ is large enough, $\forall j \in \{ 1,\ldots,q_n\}$,
\begin{eqnarray*}
    |\alpha_j|\geq |\beta_{0s1,j}|-|\alpha_j-\beta_{0s1,j}|\geq |\beta_{0s1,j}|-\frac{a_0}{2}\geq \frac{a_0}{2}>0.
\end{eqnarray*}
Then
\begin{eqnarray*}
    \left\lVert \boldsymbol{D}_1(\boldsymbol{\alpha})\right\rVert=\lambda_{\text{max}}(\boldsymbol{D}_1(\boldsymbol{\alpha}))=\max_{1\leq j\leq q_n}(\alpha_j^{-2})\leq (a_0/2)^{-2},
\end{eqnarray*}
and
\begin{eqnarray*}
    \frac{\lambda_n}{n}\left\lVert \boldsymbol{D}_1(\boldsymbol{\alpha})\dot{f}(\boldsymbol{\alpha})\right\rVert\leq\frac{\lambda_n}{n}\lambda_{\text{max}}(\boldsymbol{D}_1(\boldsymbol{\alpha}))\left\lVert \dot{f}(\boldsymbol{\alpha})\right\rVert \leq  \frac{\lambda_n}{n}(a_0/2)^{-2}\left\lVert \dot{f}(\boldsymbol{\alpha})\right\rVert.
\end{eqnarray*}
Therefore, we have
\begin{eqnarray}
   \label{G13}\left\lVert\left(\frac{\boldsymbol{\Omega}_n^{(1)}(\boldsymbol{\alpha})}{n}+\frac{\lambda_n}{n}\boldsymbol{D}_1(\boldsymbol{\alpha})\right)\dot{f}(\boldsymbol{\alpha}) \right\rVert &\geq& \left\lVert\left(\frac{\boldsymbol{\Omega}_n^{(1)}(\boldsymbol{\alpha})}{n}\right)\dot{f}(\boldsymbol{\alpha}) \right\rVert -\frac{\lambda_n}{n}\left\lVert\boldsymbol{D}_1(\boldsymbol{\alpha})\dot{f}(\boldsymbol{\alpha}) \right\rVert \nonumber\\
   &\geq& \frac{1}{c_0}\left\lVert \dot{f}(\boldsymbol{\alpha})  \right\rVert-\frac{\lambda_n}{n}(a_0/2)^{-2}\left\lVert \dot{f}(\boldsymbol{\alpha})  \right\rVert\nonumber\\
   &=&\left[\frac{1}{c_0}-\frac{\lambda_n}{n}(a_0/2)^{-2} \right]\left\lVert \dot{f}(\boldsymbol{\alpha})\right\rVert. 
\end{eqnarray}
By (\ref{G12}) and (\ref{G13}) we obtain
\begin{eqnarray*}
    o_p(1)\geq \left[\frac{1}{c_0}-\frac{\lambda_n}{n}(\alpha_0/2)^{-2} \right]\sup_{\boldsymbol{\alpha}\in H_{n1}}\left\lVert \dot{f}(\boldsymbol{\alpha})\right\rVert,
\end{eqnarray*}
and $\sup_{\boldsymbol{\alpha}\in H_{n1}}\left\lVert \dot{f}(\boldsymbol{\alpha})\right\rVert=o_p(1)$, which implies that $f(\cdot)$ is a contraction mapping from $H_{n1}$ to itself with probability tending to 1. Hence, according to the contraction mapping theorem, there exists one unique fixed-point $\widehat{\boldsymbol{\alpha}}^*\in H_{n1}$ such that 
\begin{eqnarray}
    \label{A11}\widehat{\boldsymbol{\alpha}}^*=(\boldsymbol{\Omega}_n^{(1)}(\widehat{\boldsymbol{\alpha}}^*)+\lambda_n \boldsymbol{D}_1(\widehat{\boldsymbol{\alpha}}^*))^{-1}\boldsymbol{v}_n^{(1)}(\widehat{\boldsymbol{\alpha}}^*).
\end{eqnarray}
This completes the proof of Lemma~\ref{lemma3}.
\end{proof}

\begin{proof}[Proof of Theorem~\ref{theorem1}] 
(i)
By definition of $\widehat{\boldsymbol{\beta}}^*$ and $\widehat{\boldsymbol{\beta}}^{(m)}$, we know that ${\widehat{\boldsymbol{\beta}}}^*=\lim_{m\rightarrow \infty}\widehat{\boldsymbol{\beta}}^{(m)}$, and ${\widehat{\boldsymbol{\beta}}}_{s2}^*=\lim_{m\rightarrow\infty}\widehat{\boldsymbol{\beta}}_{s2}^{(m)}$. Since $\widehat{\boldsymbol{\beta}}^{(m)}\in H_n$, by Lemma~\ref{lemma1} (ii), 
\begin{eqnarray*}
    \widehat{\boldsymbol{\beta}}_{s2}^{(m)}=\boldsymbol{\gamma}^*(\widehat{\boldsymbol{\beta}}^{(m-1)})< \frac{1}{c_0}\left\lVert \widehat{\boldsymbol{\beta}}_{s2}^{(m-1)} \right\rVert< \ldots< \left(\frac{1}{c_0}\right)^{m}\left\lVert \widehat{\boldsymbol{\beta}}_{s2}^{(0)}\right\rVert.
\end{eqnarray*}
Since $\left(1/c_0\right)^m\rightarrow 0$, $m\rightarrow \infty$, then, $\lim_{m\rightarrow \infty}\widehat{\boldsymbol{\beta}}_{s2}^{(m)}=0$, which implies that $\widehat{\boldsymbol{\beta}}_{s2}^{*}=0$ with probability tending to 1. 
\end{proof}

\begin{proof}[Proof of Theorem~\ref{theorem1}] 
(ii) In Lemma~\ref{lemma3}, we have shown that the following equation
\begin{eqnarray}
    \label{R1}
    \boldsymbol{\alpha}=(\boldsymbol{\Omega}_n^{(1)}(\boldsymbol{\alpha})+\lambda_n\boldsymbol{D}_1(\boldsymbol{\alpha}))^{-1}\boldsymbol{v}_n^{(1)}(\boldsymbol{\alpha})
\end{eqnarray}
has a unique fixed-point $\widehat{\boldsymbol{\alpha}}^*$ in the domain $H_{n1}$
 such that 
 \begin{eqnarray}
     \label{R11}
    \widehat{\boldsymbol{\alpha}}^*=(\boldsymbol{\Omega}_n^{(1)}(\widehat{\boldsymbol{\alpha}}^*)+\lambda_n\boldsymbol{D}_1(\widehat{\boldsymbol{\alpha}}^*))^{-1}\boldsymbol{v}_n^{(1)}(\widehat{\boldsymbol{\alpha}}^*),
 \end{eqnarray}
where 
\begin{eqnarray*}
    \boldsymbol{\Omega}_n^{(1)}(\widehat{\boldsymbol{\alpha}}^*)=\boldsymbol{\Omega}_n^{(1)}(\boldsymbol{\beta})\Big|_{\boldsymbol{\beta}_{s1}=\widehat{\boldsymbol{\alpha}}^*,\boldsymbol{\beta}_{s2}=0},
\end{eqnarray*}
\begin{eqnarray*}
    \boldsymbol{v}_n^{(1)}(\widehat{\boldsymbol{\alpha}}^*)=\boldsymbol{v}_n^{(1)}(\boldsymbol{\beta})\Big|_{\boldsymbol{\beta}_{s1}=\widehat{\boldsymbol{\alpha}}^*,\boldsymbol{\beta}_{s2}=0}.
\end{eqnarray*}
The next part is to show that with probability tending to 1, $\widehat{\boldsymbol{\beta}}_{s1}^*=\widehat{\boldsymbol{\alpha}}^*$, or $P(\widehat{\boldsymbol{\beta}}_{s1}^*=\widehat{\boldsymbol{\alpha}}^*)=1$, i.e., with probability tending to 1, $\widehat{\boldsymbol{\beta}}_{s1}^*$ is the unique fixed-point of (\ref{R1}).

First, by (\ref{A2}),  that is 
\begin{eqnarray*}
\begin{pmatrix}
    \boldsymbol{\alpha}^*(\boldsymbol{\beta}) - \boldsymbol{\beta}_{0s1} \\
    \boldsymbol{\gamma}^*(\boldsymbol{\beta})
\end{pmatrix}
+ \frac{\lambda_n}{n} 
\begin{pmatrix}
    \textbf{A}(\boldsymbol{\beta})\textbf{D}_1(\boldsymbol{\beta}_{s1})\boldsymbol{\alpha}^*(\boldsymbol{\beta}) + \textbf{B}(\boldsymbol{\beta})\textbf{D}_2(\boldsymbol{\beta}_{s2})\boldsymbol{\gamma}^*(\boldsymbol{\beta}) \\
    \textbf{B}^\top(\boldsymbol{\beta}) \textbf{D}_1(\boldsymbol{\beta}_{s1})\boldsymbol{\alpha}^*(\boldsymbol{\beta}) + \textbf{G}(\boldsymbol{\beta})\textbf{D}_2(\boldsymbol{\beta}_{s2})\boldsymbol{\gamma}^*(\boldsymbol{\beta}) 
\end{pmatrix}
= \widehat{\boldsymbol{b}}(\boldsymbol{\beta}) - \boldsymbol{\beta}_0,
\end{eqnarray*}
we obtain 
\begin{eqnarray*}
    \boldsymbol{\gamma}^*(\boldsymbol{\beta})+\frac{\lambda_n}{n}(\textbf{B}^\top(\boldsymbol{\beta}) \textbf{D}_1(\boldsymbol{\beta}_{s1})\boldsymbol{\alpha}^*(\boldsymbol{\beta}) + \textbf{G}(\boldsymbol{\beta})\textbf{D}_2(\boldsymbol{\beta}_{s2})\boldsymbol{\gamma}^*(\boldsymbol{\beta}) )=(\widehat{\boldsymbol{b}}(\boldsymbol{\beta}) - \boldsymbol{\beta}_0)^{(2)}.
\end{eqnarray*}
We want to show that $\lim_{\boldsymbol{\beta}_{s2}\rightarrow 0}\boldsymbol{\gamma}^*(\boldsymbol{\beta})=0$. By Lemma~\ref{lemma1} (ii) when $\boldsymbol{\beta}\in H_n$,
\begin{eqnarray*}
    \left\lVert \boldsymbol{\gamma}^*(\boldsymbol{\beta}) \right\rVert\leq \left\lVert \boldsymbol{\beta}_{s2} \right\rVert. 
\end{eqnarray*}
Therefore, $\lim_{\boldsymbol{\beta}_{s2}\rightarrow 0}\boldsymbol{\gamma}^*(\boldsymbol{\beta})=0$.
By multiplying $(\boldsymbol{\Omega}_n(\boldsymbol{\beta})+\lambda_n \boldsymbol{D}(\boldsymbol{\beta}))$ on both sides of (\ref{A1}), one can get
\begin{eqnarray}
\{\boldsymbol{\Omega}_n(\boldsymbol{\beta}) + \lambda_n \boldsymbol{D}(\boldsymbol{\beta})\}
\begin{pmatrix}
    \boldsymbol{\alpha}^*(\boldsymbol{\beta}) \\
    \boldsymbol{\gamma}^*(\boldsymbol{\beta})
\end{pmatrix} 
= \boldsymbol{v}_n(\boldsymbol{\beta}),
\end{eqnarray}
which can be rewritten as 
\begin{eqnarray*}
\begin{bmatrix}
    \begin{pmatrix}
        \boldsymbol{\Omega}_n^{(1)}(\boldsymbol{\beta})&\boldsymbol{\Omega}_n^{(12)}(\boldsymbol{\beta})\\
        \boldsymbol{\Omega}_n^{(21)}(\boldsymbol{\beta})&\boldsymbol{\Omega}_n^{(2)}(\boldsymbol{\beta})
    \end{pmatrix}+
    \begin{pmatrix}
        \lambda_n\boldsymbol{D}_1(\boldsymbol{\beta}_{s1})&0\\
        0&\lambda_n\boldsymbol{D}_2(\boldsymbol{\beta}_{s2})
    \end{pmatrix}
\end{bmatrix}\begin{pmatrix}\boldsymbol{\alpha}^*(\boldsymbol{\beta})\\
\boldsymbol{\gamma}^*(\boldsymbol{\beta})
    
\end{pmatrix}=\begin{pmatrix}
    \boldsymbol{v}_n^{(1)}(\boldsymbol{\beta})\\
    \boldsymbol{v}_n^{(2)}(\boldsymbol{\beta})
\end{pmatrix}.
\end{eqnarray*}
Consequently,
\begin{eqnarray*}
    \left(\boldsymbol{\Omega}_n^{(1)}(\boldsymbol{\beta})+\lambda_n\boldsymbol{D}_1(\boldsymbol{\beta})\right)\boldsymbol{\alpha}^*(\boldsymbol{\beta})+\boldsymbol{\Omega}_n^{(12)}(\boldsymbol{\beta})\boldsymbol{\gamma}^*(\boldsymbol{\beta})=\boldsymbol{v}_n^{(1)}(\boldsymbol{\beta}).
\end{eqnarray*}
Then, we have
\begin{eqnarray*}
    \boldsymbol{\alpha}^*(\boldsymbol{\beta})=\left( \boldsymbol{\Omega}_n^{(1)}(\boldsymbol{\beta})+\lambda_n\boldsymbol{D}_1(\boldsymbol{\beta})\right)^{-1}\left[\boldsymbol{v}_n^{(1)}(\boldsymbol{\beta})-\boldsymbol{\Omega}_n^{(12)}(\boldsymbol{\beta})\boldsymbol{\gamma}^*(\boldsymbol{\beta}) \right].
\end{eqnarray*}
Since $\lim_{\boldsymbol{\beta}_{s2}\rightarrow 0}\boldsymbol{\gamma}^*(\boldsymbol{\beta})=0$, we have 
\begin{eqnarray*}
    \lim_{\boldsymbol{\beta}_{s2}\rightarrow 0}\left[\boldsymbol{\Omega}_n^{(12)}(\boldsymbol{\beta})\boldsymbol{\gamma}^*(\boldsymbol{\beta})\right]=0,
\end{eqnarray*}
and  
\begin{eqnarray*}
    \lim_{\boldsymbol{\beta}_{s2}\rightarrow 0}\boldsymbol{\alpha}^*(\boldsymbol{\beta})=\left(\boldsymbol{\Omega}_n^{(1)}(\boldsymbol{\beta}_{s1})+\lambda_n\boldsymbol{D}_1(\boldsymbol{\beta}_{s1}) \right)^{-1}\boldsymbol{v}_n^{(1)}(\boldsymbol{\beta}_{s1})=f(\boldsymbol{\beta}_{s1}).
\end{eqnarray*}
Since $\boldsymbol{\alpha}^*(\boldsymbol{\beta})$ is continuous and thus continuous on the compact set $\boldsymbol{\beta}\in H_n$, as $m\rightarrow\infty$, $\widehat{\boldsymbol{\beta}}_{s2}^{(m)}\rightarrow 0$, we obtain
\begin{eqnarray}
\label{R20}
    \eta_m\equiv \sup_{\boldsymbol{\beta}\in H_{n1}}\left\lVert \boldsymbol{\alpha}^*(\boldsymbol{\beta}_{s1},\widehat{\boldsymbol{\beta}}_{s2}^{(m)})-f(\boldsymbol{\beta}_{s1}) \right\rVert \longrightarrow 0.
\end{eqnarray}
Since $f(\cdot)$ is a contract mapping, and $\sup_{\boldsymbol{\alpha}\in H_{n1}}\left\lVert \dot{f}(\boldsymbol{\alpha})\right\rVert \longrightarrow 0$, $n\rightarrow \infty$, then, with probability tending to 1, we have 
\begin{eqnarray*}
    \sup_{\boldsymbol{\alpha}\in H_{n1}}\left\lVert \dot{f}(\boldsymbol{\alpha})\right\rVert \leq \frac{1}{c_3},
\end{eqnarray*}
for some $c_3>1$, and 
\begin{eqnarray*}
    \left\lVert f(\widehat{\boldsymbol{\beta}}_{s1}^{(m)})-\widehat{\boldsymbol{\alpha}}^*\right\rVert =\left\lVert f(\widehat{\boldsymbol{\beta}}_{s1}^{(m)})-f(\widehat{\boldsymbol{\alpha}}^*)\right\rVert\leq \frac{1}{c_3}\left\lVert \widehat{\boldsymbol{\beta}}_{s1}^{(m)}-\widehat{\boldsymbol{\alpha}}^*\right\rVert.
\end{eqnarray*}
Note: $\widehat{\boldsymbol{\beta}}^{(m+1)}=\boldsymbol{\alpha}^*(\widehat{\boldsymbol{\beta}}^{(m)})$, i.e., $\widehat{\boldsymbol{\beta}}^{(m+1)}$ updates $\widehat{\boldsymbol{\beta}}^{(m)}$.
Now, let $h_m=\left\lVert \widehat{\boldsymbol{\beta}}_{s1}^{(m)}-\widehat{\boldsymbol{\alpha}}^* \right\rVert$, then
\begin{eqnarray*}
    h_{m+1}&=&\left\lVert \widehat{\boldsymbol{\beta}}_{s1}^{(m+1)}-\widehat{\boldsymbol{\alpha}}^* \right\rVert=\left\lVert \boldsymbol{\alpha}^*(\widehat{\boldsymbol{\beta}}^{(m)})-\widehat{\boldsymbol{\alpha}}^* \right\rVert\\
    &\leq& \left\lVert \boldsymbol{\alpha}^*(\widehat{\boldsymbol{\beta}}^{(m)})-f(\widehat{\boldsymbol{\beta}}_{s1}^{(m)}) \right\rVert+\left\lVert f(\widehat{\boldsymbol{\beta}}_{s1}^{(m)})-f(\widehat{\boldsymbol{\alpha}}^*) \right\rVert\\
    &\leq& \left\lVert \boldsymbol{\alpha}^*(\widehat{\boldsymbol{\beta}}_{s1}^{(m)},\widehat{\boldsymbol{\beta}}_{s2}^{(m)}) -f(\widehat{\boldsymbol{\beta}}_{s1}^{(m)})\right\rVert+\left\lVert f(\widehat{\boldsymbol{\beta}}_{s1}^{(m)})-f(\widehat{\boldsymbol{\alpha}}^*) \right\rVert\\
    &\leq&\eta_m+\frac{1}{c_3}\left\lVert \widehat{\boldsymbol{\beta}}_{s1}^{(m)}-\widehat{\boldsymbol{\alpha}}^*\right\rVert\\
    &\leq& \eta_m+\frac{1}{c_3} h_m.
\end{eqnarray*}
By (\ref{R20}), for any $\epsilon>0$, there exists an $N>0$ such that for all $m>N$, $\eta_m<\epsilon$. Therefore, for $m>N$, or $m-N>0$, we have
\begin{eqnarray*}
    h_{m+1}&\leq& \frac{1}{c_3}h_m+\eta_m\\
    &\leq& \frac{1}{c_3}(\frac{1}{c_3}h_{m-1}+\eta_{m-1})+\eta_m\\
    &=&\frac{1}{c_3^2}h_{m-1}+\frac{1}{c_3}\eta_{m-1}+\eta_m\\
    &\leq& \frac{h_1}{c_3^m}+\frac{\eta_1}{c_3^{m-1}}+\frac{\eta_2}{c_3^{m-2}}+\cdots+\frac{\eta_N}{c_3^{m-N}}+\frac{\eta_{N+1}}{c_3^{m-(N+1)}}+\cdots+\frac{\eta_{m-1}}{c_3}+\eta_m\\
    &=&\frac{h_1}{c_3^m}+\frac{\eta_1}{c_3^{m-1}}+\frac{\eta_2}{c_3^{m-2}}+\cdots+\frac{\eta_N}{c_3^{m-N}}+\left(\frac{\eta_{N+1}}{c_3^{m-(N+1)}}+\cdots+\frac{\eta_{m-1}}{c_3}+\eta_m \right)\\
    &\leq& (h_1+\eta_1+\cdots+\eta_N)\frac{1}{c_3^{m-N}}+\left( \frac{1}{c_3^{m-(N+1)}}+\cdots+\frac{1}{c_3}+1\right)\epsilon\\
    &=&(h_1+\eta_1+\cdots+\eta_N)\frac{1}{c_3^{m-N}}+\frac{1-(1/c_3)^{m-N}}{1-(1/c_3)},~ \text{(by sum of the geometric series)}
\end{eqnarray*}
Since $1/c_3^{m-N}\rightarrow 0$ and $\frac{1-(1/c_3)^{m-N}}{1-(1/c_3)}\rightarrow\frac{c_3}{c_3-1}\epsilon$, when $m\rightarrow\infty$, there exists $N_0>N$ such that when $m>N_0$,
\begin{eqnarray*}
    (h_1+\eta_1+\cdots+\eta_N)\frac{1}{c_3^{m-N}}<\epsilon,
\end{eqnarray*}
and 
\begin{eqnarray*}
    \frac{1-(1/c_3)^{m-N}}{1-(1/c_3)}<2\frac{c_3}{c_3-1}\epsilon,
\end{eqnarray*}
which implies 
\begin{eqnarray*}
    h_{m+1}<\left(1+\frac{2c_3}{c_3-1} \right)\epsilon=\frac{3c_3-1}{c_3-1}\epsilon.
\end{eqnarray*}
Then, $h_{m+1}\rightarrow 0$
when $m\rightarrow\infty$.
Hence, with probability tending to 1, we have $h_m=\left\lVert \widehat{\boldsymbol{\beta}}_{s1}^{(m)}-\widehat{\boldsymbol{\alpha}}^*\right\rVert\rightarrow 0$ as $m\rightarrow\infty$ because $\widehat{\boldsymbol{\beta}}_{s1}^*=\lim_{m\rightarrow\infty}\widehat{\boldsymbol{\beta}}_{s1}^{(m)}$ and 
\begin{eqnarray*}
    \left\lVert \widehat{\boldsymbol{\beta}}_{s1}^*-\widehat{\boldsymbol{\alpha}}^* \right\rVert \leq \left\lVert \widehat{\boldsymbol{\beta}}_{s1}^*-\widehat{\boldsymbol{\beta}}_{s1}^{(m)} \right\rVert+ \left\lVert \widehat{\boldsymbol{\beta}}_{s1}^{(m)} -\widehat{\boldsymbol{\alpha}}^* \right\rVert\longrightarrow 0, 
\end{eqnarray*}
when $m\rightarrow\infty$. This implies $P(\widehat{\boldsymbol{\beta}}_{s1}^*=\widehat{\boldsymbol{\alpha}}^*)=1$ and the proof of Theorem~\ref{theorem1} (ii) is completed.
\end{proof}
\begin{proof}[Proof of Theorem~\ref{theorem1}] 
(iii). From (\ref{A11}), we have
\begin{eqnarray*}
\widehat{\boldsymbol{\alpha}}^*=(\boldsymbol{\Omega}_n^{(1)}(\widehat{\boldsymbol{\alpha}}^*)+\lambda_n\boldsymbol{D}_1(\widehat{\boldsymbol{\alpha}}^*))^{-1}\boldsymbol{v}_n^{(1)}(\widehat{\boldsymbol{\alpha}}^*)
\end{eqnarray*}
and
\begin{eqnarray*}
    \sqrt{n}(\widehat{\boldsymbol{\alpha}}^*-\boldsymbol{\beta}_{0s1})=\pi_1+\pi_2,
\end{eqnarray*}
where
\begin{eqnarray*}
    \pi_1&\equiv& \sqrt{n}\left[ (\boldsymbol{\Omega}_n^{(1)}(\widehat{\boldsymbol{\alpha}}^*)+\lambda_n\boldsymbol{D}_1(\widehat{\boldsymbol{\alpha}}^*))^{-1}\boldsymbol{\Omega}_n^{(1)}(\widehat{\boldsymbol{\alpha}}^*)-\textbf{I}_{q_n}\right]\boldsymbol{\beta}_{0s1}, \\
    \pi_2&\equiv& \sqrt{n} (\boldsymbol{\Omega}_n^{(1)}(\widehat{\boldsymbol{\alpha}}^*)+\lambda_n\boldsymbol{D}_1(\widehat{\boldsymbol{\alpha}}^*))^{-1}\left( \boldsymbol{v}_n^{(1)}(\widehat{\boldsymbol{\alpha}}^*)-\boldsymbol{\Omega}_n^{(1)}(\widehat{\boldsymbol{\alpha}}^*)\boldsymbol{\beta}_{0s1}\right).
\end{eqnarray*}
Noticing that for any two conformable invertible matrices $\boldsymbol{\zeta}$ and $\boldsymbol{\Psi}$, we have
\begin{eqnarray*}
    (\boldsymbol{\zeta}+\boldsymbol{\Psi})^{-1}=\boldsymbol{\zeta}^{-1}-\boldsymbol{\zeta}^{-1}\boldsymbol{\Psi}(\boldsymbol{\zeta}+\boldsymbol{\Psi})^{-1}.
\end{eqnarray*}
Then
\begin{eqnarray*}
    \left(\boldsymbol{\Omega}_n^{(1)}(\widehat{\boldsymbol{\alpha}}^*)+\lambda_n\boldsymbol{D}_1(\widehat{\boldsymbol{\alpha}}^*) \right)^{-1}&=&\left(\boldsymbol{\Omega}_n^{(1)}(\widehat{\boldsymbol{\alpha}}^*) \right)^{-1}-\lambda_n\left(\boldsymbol{\Omega}_n^{(1)}(\widehat{\boldsymbol{\alpha}}^*) \right)^{-1}\\
    &&\boldsymbol{D}_1(\widehat{\boldsymbol{\alpha}}^*)\left(\boldsymbol{\Omega}_n^{(1)}(\widehat{\boldsymbol{\alpha}}^*)+\lambda_n\boldsymbol{D}_1(\widehat{\boldsymbol{\alpha}}^*) \right)^{-1}.
\end{eqnarray*}
Therefore, we obtain
\begin{eqnarray}
    \label{R25}
    &&(\boldsymbol{\Omega}_n^{(1)}(\widehat{\boldsymbol{\alpha}}^*)+\lambda_n\boldsymbol{D}_1(\widehat{\boldsymbol{\alpha}}^*))^{-1}(\boldsymbol{\Omega}_n^{(1)}(\widehat{\boldsymbol{\alpha}}^*))=\nonumber\\
    &&\textbf{I}_{q_n}-\lambda_n(\boldsymbol{\Omega}_n^{(1)}(\widehat{\boldsymbol{\alpha}}^*))^{-1}\boldsymbol{D}_1(\widehat{\boldsymbol{\alpha}}^*)\left(\boldsymbol{\Omega}_n^{(1)}(\widehat{\boldsymbol{\alpha}}^*)+\lambda_n\boldsymbol{D}_1(\widehat{\boldsymbol{\alpha}}^*) \right)^{-1}\boldsymbol{\Omega}_n^{(1)}(\widehat{\boldsymbol{\alpha}}^*)
\end{eqnarray}
and 
\begin{eqnarray*}
    \pi_1&=&\sqrt{n}\left[ -\lambda_n(\boldsymbol{\Omega}_n^{(1)}(\widehat{\boldsymbol{\alpha}}^*))^{-1}\boldsymbol{D}_1(\widehat{\boldsymbol{\alpha}}^*)(\boldsymbol{\Omega}_n^{(1)}(\widehat{\boldsymbol{\alpha}}^*)+\lambda_n\boldsymbol{D}_1(\widehat{\boldsymbol{\alpha}}^*))^{-1}\boldsymbol{\Omega}_n^{(1)}(\widehat{\boldsymbol{\alpha}}^*)\boldsymbol{\beta}_{0s1}\right]\\
    &=&-\frac{\lambda_n}{\sqrt{n}}\left(\frac{1}{n}\boldsymbol{\Omega}_n^{(1)}(\widehat{\boldsymbol{\alpha}}^*) \right)^{-1}\boldsymbol{D}_1(\widehat{\boldsymbol{\alpha}}^*)\left(\frac{1}{n}\boldsymbol{\Omega}_n^{(1)}(\widehat{\boldsymbol{\alpha}}^*)+\frac{\lambda_n}{n}\boldsymbol{D}_1(\widehat{\boldsymbol{\alpha}}^*)\right)^{-1}\frac{1}{n}\boldsymbol{\Omega}_n^{(1)}(\widehat{\boldsymbol{\alpha}}^*)\boldsymbol{\beta}_{0s1}.
\end{eqnarray*}
By conditions (C5) and (C6), we have 
\begin{eqnarray}
    \label{R30}
    \left\lVert\pi_1 \right\rVert=O_p(\lambda_n\sqrt{q_n/n})\longrightarrow 0.
\end{eqnarray}
Next, we consider $\pi_2$. It follows from (\ref{R25}) and Condition (C6): $\lambda_n/\sqrt{n}\rightarrow 0$, that
\begin{eqnarray*}
    \pi_2&\equiv& \sqrt{n}(\boldsymbol{\Omega}_n^{(1)}(\widehat{\boldsymbol{\alpha}}^*)+\lambda_n\boldsymbol{D}_1(\widehat{\boldsymbol{\alpha}}^*))^{-1}\left(\boldsymbol{v}_n^{(1)}(\widehat{\boldsymbol{\alpha}}^*) -\boldsymbol{\Omega}_n^{(1)}(\widehat{\boldsymbol{\alpha}}^*)\boldsymbol{\beta}_{0s1}\right)\\
    &=&\sqrt{n}\left[
    (\boldsymbol{\Omega}_n^{(1)}(\widehat{\boldsymbol{\alpha}}^*))^{-1}-\lambda_n(\boldsymbol{\Omega}_n^{(1)}(\widehat{\boldsymbol{\alpha}}^*))^{-1}\boldsymbol{D}_1(\widehat{\boldsymbol{\alpha}}^*)(\boldsymbol{\Omega}_n^{(1)}(\widehat{\boldsymbol{\alpha}}^*)+\lambda_n\boldsymbol{D}_1(\widehat{\boldsymbol{\alpha}}^*))^{-1}\right]\\
    &&\left( \boldsymbol{v}_n^{(1)}(\widehat{\boldsymbol{\alpha}}^*)-\boldsymbol{\Omega}_n^{(1)}(\widehat{\boldsymbol{\alpha}}^*)\boldsymbol{\beta}_{0s1}\right)\\
    &=&\sqrt{n}\left[
    \left(\frac{1}{n}\boldsymbol{\Omega}_n^{(1)}(\widehat{\boldsymbol{\alpha}}^*)\right)^{-1}-\frac{\lambda_n}{n}\left(\frac{1}{n}\boldsymbol{\Omega}_n^{(1)}(\widehat{\boldsymbol{\alpha}}^*)\right)^{-1}\boldsymbol{D}_1(\widehat{\boldsymbol{\alpha}}^*)\left(\frac{1}{n}\boldsymbol{\Omega}_n^{(1)}(\widehat{\boldsymbol{\alpha}}^*)\right)+\frac{\lambda_n}{n}\boldsymbol{D}_1(\widehat{\boldsymbol{\alpha}}^*))^{-1}\right]\\
    &&\left( \frac{1}{n}\boldsymbol{v}_n^{(1)}(\widehat{\boldsymbol{\alpha}}^*)-\frac{1}{n}\boldsymbol{\Omega}_n^{(1)}(\widehat{\boldsymbol{\alpha}}^*)\boldsymbol{\beta}_{0s1}\right).
\end{eqnarray*}
By Condition (C6), $\lambda_n/n=(\lambda_n/\sqrt{n})(1/\sqrt{n})=o(1)\cdot(1/\sqrt{n})=o(1/\sqrt{n})$, we have 
\begin{eqnarray*}
    \pi_2=\sqrt{n}\left[\left(\frac{1}{n}\boldsymbol{\Omega}_n^{(1)}(\widehat{\boldsymbol{\alpha}}^*)\right)^{-1}-o_p(1/\sqrt{n})
    \right]\left( \frac{1}{n}\boldsymbol{v}_n^{(1)}(\widehat{\boldsymbol{\alpha}}^*)-\frac{1}{n}\boldsymbol{\Omega}_n^{(1)}(\widehat{\boldsymbol{\alpha}}^*)\boldsymbol{\beta}_{0s1}\right).
\end{eqnarray*}
Using the first-order Taylor expansion on
\begin{eqnarray*}
\boldsymbol{v}_n(\widehat{\boldsymbol{\alpha}}^*)=\boldsymbol{v}_n(\boldsymbol{\beta})\Big|_{\boldsymbol{\beta}_{s1}=\widehat{\boldsymbol{\alpha}}^*,\boldsymbol{\beta}_{s2}=0}=\dot{\ell}_n(\widehat{\boldsymbol{\alpha}}^*|\widetilde{\boldsymbol{\Lambda}})-\Ddot{\ell}_n(\widehat{\boldsymbol{\alpha}}^*|\widetilde{\boldsymbol{\Lambda}})\begin{pmatrix}
       \widehat{\boldsymbol{\alpha}}^*\\
       0
   \end{pmatrix},
\end{eqnarray*}
we obtain
\begin{eqnarray*}
    \boldsymbol{v}_n^{(1)}(\widehat{\boldsymbol{\alpha}}^*)&=&\dot{\ell}_n^{(1)}(\widehat{\boldsymbol{\alpha}}^*|\widetilde{\boldsymbol{\Lambda}})+\boldsymbol{\Omega}_n^{(1)}(\widehat{\boldsymbol{\alpha}}^*)\widehat{\boldsymbol{\alpha}}^*\\
    &=&\dot{\ell}_n^{(1)}(\boldsymbol{\beta}_{0s1}|\widetilde{\boldsymbol{\Lambda}})+\Ddot{\ell}_n(\widetilde{\boldsymbol{\alpha}}^*|\widetilde{\boldsymbol{\Lambda}})(\widehat{\boldsymbol{\alpha}}^*-\boldsymbol{\beta}_{0s1})+\boldsymbol{\Omega}_n^{(1)}(\widehat{\boldsymbol{\alpha}}^*)\widehat{\boldsymbol{\alpha}}^*,
\end{eqnarray*}
where $\widetilde{\boldsymbol{\alpha}}^*$ is between $\widehat{\boldsymbol{\alpha}}^*$ and $\boldsymbol{\beta}_{0s1}$, $\left\lVert \widetilde{\boldsymbol{\alpha}}^*-\boldsymbol{\beta}_{0s1}\right\rVert=o_p(1)$, and $\left\lVert \widetilde{\boldsymbol{\alpha}}^*-\widehat{\boldsymbol{\alpha}}^*\right\rVert=o_p(1)$. By Condition (C4), we have 
\begin{eqnarray*}
    \frac{1}{n}\boldsymbol{\Omega}_n^{(1)}(\widehat{\boldsymbol{\alpha}}^*)- \frac{1}{n}\boldsymbol{\Omega}_n^{(1)}(\widetilde{\boldsymbol{\alpha}}^*)=o_p(1),
\end{eqnarray*}
then 
\begin{eqnarray*}
    &&\frac{1}{n}\boldsymbol{v}_n^{(1)}(\widehat{\boldsymbol{\alpha}}^*)-\frac{1}{n}\boldsymbol{\Omega}_n^{(1)}(\widehat{\boldsymbol{\alpha}}^*)\boldsymbol{\beta}_{0s1}\\
    &&=\frac{1}{n}\dot{\ell}_n^{(1)}(\boldsymbol{\beta}_{0s1}|\widetilde{\boldsymbol{\Lambda}})-\left(-\frac{1}{n}\Ddot{\ell}_n^{(1)}(\widetilde{\boldsymbol{\alpha}}^*|\widetilde{\boldsymbol{\Lambda}})\right)(\widehat{\boldsymbol{\alpha}}^*-\boldsymbol{\beta}_{0s1})+\left(\frac{1}{n}\boldsymbol{\Omega}_n^{(1)}(\widehat{\boldsymbol{\alpha}}^*) \right)(\widehat{\boldsymbol{\alpha}}^*-\boldsymbol{\beta}_{0s1})\\
    &&=\frac{1}{n}\dot{\ell}_n^{(1)}(\boldsymbol{\beta}_{0s1}|\widetilde{\boldsymbol{\Lambda}})+\left(\frac{1}{n}\boldsymbol{\Omega}_n^{(1)}(\widehat{\boldsymbol{\alpha}}^*)-\frac{1}{n}\boldsymbol{\Omega}_n^{(1)}(\widetilde{\boldsymbol{\alpha}}^*)\right)(\widehat{\boldsymbol{\alpha}}^*-\boldsymbol{\beta}_{0s1})\\
    &&=\frac{1}{n}\dot{\ell}_n^{(1)}(\boldsymbol{\beta}_{0s1}|\widetilde{\boldsymbol{\Lambda}})+o_p(1).
\end{eqnarray*}
Hence, we have 
\begingroup
\allowdisplaybreaks
\begin{eqnarray*}
    \sqrt{n}(\widehat{\boldsymbol{\alpha}}^*-\boldsymbol{\beta}_{0s1})&=&\pi_2+\pi_1\\
    &=&\sqrt{n}\left[(I^{(1)}(\boldsymbol{\beta}_{0s1}))^{-1}+o_p(1)-o_p(1/\sqrt{n}) \right]\\
    &&\left[\frac{1}{n}\dot{\ell}_n^{(1)}(\boldsymbol{\beta}_{0s1}|\widetilde{\boldsymbol{\Lambda}})+o_p(1)(\widehat{\boldsymbol{\alpha}}^*-\boldsymbol{\beta}_{0s1}) 
    \right]+o_p(1)\\
    &=&\left[(I^{(1)}(\boldsymbol{\beta}_{0s1}))^{-1}+o_p(1) \right]\left[n^{-1/2}\dot{\ell}_n^{(1)}(\boldsymbol{\beta}_{0s1 }|\widetilde{\boldsymbol{\Lambda}})\right]\\
    &+&o_p(1)\sqrt{n}(\widehat{\boldsymbol{\alpha}}^*-\boldsymbol{\beta}_{0s1})+o_p(1).
\end{eqnarray*}
\endgroup
Further, we obtain
\begin{eqnarray}
\label{extratheorem1}
    \sqrt{n}(\widehat{\boldsymbol{\alpha}}^*-\boldsymbol{\beta}_{0s1})(1+o_p(1))&=&\left[(I^{(1)}(\boldsymbol{\beta}_{0s1}))^{-1}+o_p(1) \right]\left[n^{-1/2}\dot{\ell}_n^{(1)}(\boldsymbol{\beta}_{0s1 }|\widetilde{\boldsymbol{\Lambda}})\right]\nonumber\\
    &+&o_p(1).
\end{eqnarray}
By simplifying (\ref{extratheorem1}), we have 
\begin{eqnarray*}
    \sqrt{n}(\widehat{\boldsymbol{\alpha}}^*-\boldsymbol{\beta}_{0s1})=(I^{(1)}(\boldsymbol{\beta}_{0s1}))^{-1}\left[
    n^{-1/2} \dot{\ell}_n^{(1)}(\boldsymbol{\beta}_{0s1 }|\widetilde{\boldsymbol{\Lambda}})\right]+o_p(1).
\end{eqnarray*}
Let $\boldsymbol{\Sigma}=\left(I^{(1)}(\boldsymbol{\beta}_{0s1})\right)^{-1}$, then for any $\boldsymbol{b}_n$ being a $q_n$-vector, assume $\left\lVert \boldsymbol{b}_n\right\rVert=1$ or $\boldsymbol{b}_n^\top \boldsymbol{b}_n=1$, we have
\begin{eqnarray*}
    \sqrt{n}\boldsymbol{b}_n^\top \boldsymbol{\Sigma}^{-\frac{1}{2}}(\widehat{\boldsymbol{\alpha}}^*-\boldsymbol{\beta}_{0s1})&=&\boldsymbol{b}_n^\top \boldsymbol{\Sigma}^{-\frac{1}{2}}(I^{(1)}(\boldsymbol{\beta}_{0s1}))^{-1}\left[n^{-1/2}\dot{\ell}_n^{(1)}(\boldsymbol{\beta}_{0s1 }|\widetilde{\boldsymbol{\Lambda}})\right]+o_p(1)\\
    &=&\boldsymbol{b}_n^\top 
    (I^{(1)}(\boldsymbol{\beta}_{0s1}))^{-\frac{1}{2}}
    \left[n^{-1/2}\dot{\ell}_n^{(1)}(\boldsymbol{\beta}_{0s1 }|\widetilde{\boldsymbol{\Lambda}})\right]+o_p(1).
\end{eqnarray*}
Since $\dot{\ell}_n^{(1)}(\boldsymbol{\beta}_{0s1 }|\widetilde{\boldsymbol{\Lambda}})$ is the partial score about $\boldsymbol{\beta}$ and can be considered as the semiparametric efficient score  \cite[see][]{bickel1993efficient}, we have 
\begin{eqnarray*}
\lefteqn{\text{Cov}\left\{ \boldsymbol{b}_n^\top(I^{(1)}(\boldsymbol{\beta}_{0s1}))^{-\frac{1}{2}}\left[n^{-1/2}\dot{\ell}_n^{(1)}(\boldsymbol{\beta}_{0s1 }|\widetilde{\boldsymbol{\Lambda}}) \right]\right\}
} \\
&=&\boldsymbol{b}_n^\top(I^{(1)}(\boldsymbol{\beta}_{0s1}))^{-\frac{1}{2}}I^{(1)}(\boldsymbol{\beta}_{0s1})(I^{(1)}(\boldsymbol{\beta}_{0s1}))^{-\frac{1}{2}}\boldsymbol{b}_n\\
    &=&\boldsymbol{b}_n^\top \boldsymbol{b}_n=1.
\end{eqnarray*}
Therefore, by the Central Limit Theorem and Slutsky's Theorem, we have
\begin{eqnarray*}
    \sqrt{n}\boldsymbol{b}_n^\top \boldsymbol{\Sigma}^{-\frac{1}{2}}(\widehat{\boldsymbol{\alpha}}^*-\boldsymbol{\beta}_{0s1})\longrightarrow N(0,1)
\end{eqnarray*}
in distribution, and equivalently, 
\begin{eqnarray*}
    \sqrt{n}\boldsymbol{b}_n^\top \boldsymbol{\Sigma}^{-\frac{1}{2}}(\widehat{\boldsymbol{\beta}}_1^*-
    \boldsymbol{\beta}_{0s1}) \longrightarrow N(0,1)
\end{eqnarray*}
in distribution.
The proof of Theorem \ref{theorem1} (iii) is completed.
\end{proof}
\end{document}